\documentclass[envcountsame,runningheads,a4paper]{llncs}

\usepackage[UKenglish]{babel}
\usepackage{wrapfig}

\pdfoutput=1

\usepackage[final]{microtype}
\usepackage{amssymb}
\usepackage{enumerate}
\usepackage{tikz}
\usepackage{xcolor,latexsym,amsmath,extarrows,alltt}
\usepackage{xspace}
\usepackage{booktabs}
\usepackage{mathtools}
\usepackage{enumitem}
\usepackage{stmaryrd}
\usepackage{bussproofs}
\usepackage{multirow}
\usepackage{url}

\newcommand{\N}{\mathbb{N}}
\newcommand{\F}{\mathcal{F}}
\newcommand{\sig}[1]{\funs(#1)}
\newcommand{\V}{\mathcal{V}}
\newcommand{\M}{\mathcal{M}}
\renewcommand{\P}{\mathcal{P}}
\newcommand{\Sorts}{\mathcal{S}}
\newcommand{\Terms}{\mathcal{T}}
\newcommand{\Rules}{\mathcal{R}}
\newcommand{\AlterRules}{{\mathcal{U}}}
\newcommand{\RulesEta}{\Rules^{\mathtt{ext}}}

\newcommand{\Defineds}{\mathcal{D}}
\newcommand{\FV}{\mathit{FV}}
\newcommand{\FMV}{\mathit{FMV}}
\newcommand{\domain}{\mathtt{dom}}

\newcommand{\SDP}{\mathit{SDP}}

\newcommand{\FR}{\mathit{FR}}

\newcommand{\UR}{\mathit{UR}}
\newcommand{\asort}{\iota}
\newcommand{\bsort}{\kappa}
\newcommand{\atype}{\sigma}
\newcommand{\btype}{\tau}
\newcommand{\ctype}{\pi}

\newcommand{\identifier}[1]{\mathtt{#1}}
\newcommand{\afun}{\identifier{f}}
\newcommand{\bfun}{\identifier{g}}
\newcommand{\cfun}{\identifier{h}}
\newcommand{\adpprob}{M}
\newcommand{\bdpprob}{Q}
\newcommand{\avarormeta}{b}

\newcommand{\no}{\texttt{NO}}

\newcommand{\wanda}{\textsf{WANDA}\xspace}

\newcommand{\aprove}{\textsf{AProVE}\xspace}
\newcommand{\minisat}{\textsf{MiniSat}\xspace}
\newcommand{\csiho}{$\textsf{CSI}^{\textsf{ho}}$\xspace}
\renewcommand{\csiho}{$\textsf{CSI\textasciicircum{}ho}$\xspace}
\newcommand{\acph}{\textsf{ACPH}\xspace}

\newcommand{\app}[2]{#1\ #2}
\newcommand{\apps}[3]{#1\ #2 \cdots #3}
\newcommand{\abs}[2]{\lambda #1.#2}
\newcommand{\meta}[2]{#1\langle #2\rangle}
\newcommand{\metaapply}[2]{#1 \langle\!\langle #2 \rangle\!\rangle}

\newcommand{\arity}{\mathit{arity}}
\newcommand{\minarity}[1]{\mathit{minar}}
\newcommand{\head}{\mathsf{head}}
\newcommand{\symb}[1]{\mathtt{#1}}

\newcommand{\arrtype}{\rightarrow}
\newcommand{\arrdp}{\Rrightarrow}
\newcommand{\arrz}{\Rightarrow}
\newcommand{\arr}[1]{\arrz_{#1}}
\newcommand{\arrr}[1]{\arr{#1}^*}

\newcommand{\supterm}{\rhd}
\newcommand{\suptermeq}{\unrhd}
\newcommand{\bsuptermeq}[1]{\unrhd_{#1}}
\newcommand{\bettersuptermeq}{\mathbin{\hspace{1pt}\underline{\hspace{-1pt}\blacktriangleright\hspace{-1pt}}\hspace{1pt}}}
\newcommand{\bettersupterm}{\blacktriangleright}
\newcommand{\safesup}{\unrhd_{\mathtt{safe}}^{\text{\tiny\cite{suz:kus:bla:11}}}}

\newcommand{\cand}{\mathsf{cand}}
\newcommand{\Proc}{\mathit{Proc}}
\newcommand{\static}{\mathtt{computable}}
\newcommand{\minimal}{\mathtt{minimal}}
\newcommand{\arbitrary}{\mathtt{arbitrary}}
\newcommand{\formative}{\mathtt{formative}}
\newcommand{\nonformative}{\mathtt{all}}
\newcommand{\metafy}{\mathit{metafy}}
\newcommand{\greqsort}{\succeq^{\Sorts}}
\newcommand{\leqsort}{\preceq^{\Sorts}}
\newcommand{\eqsort}{\approx^{\Sorts}}
\newcommand{\grsort}{\succ^{\Sorts}}
\newcommand{\gracsortup}{\succeq^{\Sorts}_+}
\newcommand{\gracsortdown}{\succ^{\Sorts}_-}
\newcommand{\gracc}{\unrhd_{\mathtt{acc}}}
\newcommand{\accreduce}[1]{\Rrightarrow_{#1}}
\newcommand{\Acc}{\mathit{Acc}}
\newcommand{\approxp}{\approx_\mac}
\newcommand{\project}{\overline{\nu}}

\newcommand{\funs}{\mathtt{funs}}

\newcommand{\unsharp}[1]{#1^\flat}
\newcommand{\etalong}[1]{#1\!\!\uparrow^\eta}
\newcommand{\halfetalong}[1]{\overline{#1}}

\newcommand{\pgt}{\succ} 
\newcommand{\pge}{\succcurlyeq} 
\newcommand{\rge}{\succsim} 

\newcommand{\nul}{\symb{0}}
\newcommand{\one}{\symb{1}}
\newcommand{\nil}{\symb{nil}}
\newcommand{\cons}{\symb{cons}}

\newcommand{\ssin}{\symb{sin}}
\newcommand{\suc}{\symb{s}}
\newcommand{\fix}{\symb{fix}}
\newcommand{\map}{\symb{map}}

\newcommand{\deriv}{\symb{deriv}}

\newcommand{\nat}{\symb{nat}}
\newcommand{\lijst}{\symb{list}}
\newcommand{\real}{\symb{real}}

\newcommand{\mac}{{\textcolor{red}{e}}}
\renewcommand{\mac}{k}
\newcommand{\mia}{{\textcolor{red}{k}}}
\renewcommand{\mia}{k}
\newcommand{\maa}{{\textcolor{red}{m}}}
\renewcommand{\maa}{m}

\newcommand{\refDef}[1]{Def.~\ref{#1}}
\newcommand{\refEx}[1]{Ex.~\ref{#1}}
\newcommand{\refLemma}[1]{Lemma~\ref{#1}}
\newcommand{\refThm}[1]{Thm.~\ref{#1}}
\newcommand{\refApp}[1]{Appendix~\ref{#1}}
\newcommand{\refSec}[1]{\S~\ref{#1}}

\newcommand{\onlypaper}[1]{}
\newcommand{\onlyarxiv}[1]{#1}

\begin{document}

\title{A static higher-order dependency pair framework \onlyarxiv{(extended version)}}

\author{Carsten Fuhs\inst{1} \and Cynthia Kop\inst{2}}
\institute{Dept.\ of Comp.\ Sci.\ and Inf.\ Sys.,
Birkbeck, University of London, UK \and
Dept.\ of Software Science,
Radboud University Nijmegen, The Netherlands}

\maketitle

\begin{abstract}
We revisit the static dependency pair method for proving termination
of higher-order term rewriting and extend it in a number of ways:
(1) We introduce a new rewrite formalism designed for general applicability
in termination proving of higher-order rewriting, Algebraic Functional
Systems with Meta-variables.
(2) We provide a syntactically checkable soundness criterion to make
  the method applicable to a large class of rewrite systems.
(3) We propose a modular dependency pair \emph{framework} for this
higher-order setting.
(4) We introduce a fine-grained notion of \emph{formative} and
\emph{computable} chains to render the framework more powerful.
(5) We formulate several existing and new termination proving techniques
in the form of processors within our framework.

The framework has been implemented in the (fully automatic)
higher-order termination tool \wanda.
\end{abstract}

\section{Introduction}

Term rewriting \cite{baa:nip:98,ter:03} is an important area of logic, with
applications in many different areas of computer science
\cite{bac:gan:94,der:kap:89,fuh:kop:nis:17,haf:nip:10,hoe:arv:99,mea:92,ott:bro:ess:gie:10}.
\emph{Higher-order} term rewriting -- which extends the traditional
\emph{first-order} term rewriting with higher-order types and binders
as in the $\lambda$-calculus -- offers a formal foundation of functional
programming
and a tool for equational reasoning in
higher-order logic.
A key question
in the analysis of both first- and
higher-order term rewriting is \emph{termination};
both for its own sake, and
as part of confluence and equivalence analysis.

In first-order term rewriting, a
hugely
effective method for proving
termination (both manually and automatically) is the
\emph{dependency pair (DP) approach} \cite{art:gie:00}.
This approach has been extended to the \emph{DP framework}
\cite{gie:thi:sch:05:2,gie:thi:sch:fal:06}, a highly
modular
methodology which new techniques for proving termination \emph{and
non-termination} can easily be plugged into
in the form of \emph{processors}.

In higher-order rewriting, two DP approaches
with distinct costs and benefits are used:
\emph{dynamic} \cite{sak:wat:sak:01,kop:raa:12} and
\emph{static}
\cite{bla:06,sak:kus:05,kus:iso:sak:bla:09,suz:kus:bla:11,kus:13,kus:18}
DPs.
Dynamic DPs are more broadly applicable, yet static DPs often
enable more powerful analysis techniques.
Still, neither approach has the modularity and extendability of the DP
framework, nor can they
be used to prove non-termination.
Also, these approaches consider different styles
of higher-order rewriting, which means that
for all results certain language features are not available.

In this paper, we
address these issues for the \emph{static} DP
approach by extending it to a full higher-order
\emph{dependency pair framework} for both termination and
non-termination analysis.  For broad applicability, we
introduce a new rewriting formalism, \emph{AFSMs},
to capture several
flavours of higher-order rewriting, including \emph{AFSs}
\cite{jou:rub:99} (used in the annual Termination Competition
\cite{termcomp}) and \emph{pattern HRSs} \cite{nip:91,mil:91} (used
in the annual Confluence Competition \cite{coco}).
To show the versatility and power of this methodology, we
define various processors
in the framework -- both
adaptations of existing processors from the literature and entirely
new ones.

\medskip
\emph{Detailed contributions.}
We
reformulate the results of
\cite{bla:06,sak:kus:05,kus:iso:sak:bla:09,suz:kus:bla:11,kus:13} into
a DP framework for AFSMs. In doing so, we
instantiate the applicability restriction of
\cite{kus:13} by a very liberal syntactic condition, and add two
new flags to track properties of DP problems: one
completely new, one
from an earlier work by the authors for the
\emph{first-order} DP framework \cite{fuh:kop:14}.
We
give eight \emph{processors} for reasoning
in our
framework: four
translations of techniques
from
static DP approaches, three
techniques from
first-order or dynamic DPs, and one completely new.

This is a \emph{foundational} paper, focused on defining a general
theoretical framework for higher-order termination analysis using
dependency pairs rather than questions of implementation.
We have, however, implemented most of these
results in the fully
automatic termination analysis
tool \wanda~\cite{wanda}.

\medskip
\emph{Related Work.}
There is a vast body of work in the first-order setting regarding the
DP approach \cite{art:gie:00} and framework
\cite{gie:thi:sch:05:2,gie:thi:sch:fal:06,hir:mid:07}.
We have drawn from the ideas in these works for the core structure of
the higher-order framework, but have added some new features of our
own and adapted results to the higher-order setting.

There is no true higher-order DP \emph{framework} yet:
both static and dynamic approaches actually lie halfway between the
original ``DP approach'' of first-order rewriting and a full DP
framework as in \cite{gie:thi:sch:05:2,gie:thi:sch:fal:06}.
Most of these works
\cite{kop:raa:11,kop:raa:12,kus:13,kus:iso:sak:bla:09,suz:kus:bla:11}
prove ``non-loopingness'' or ``chain-freeness'' of a set $\P$ of DPs
through a number of theorems.
Yet, there is no concept of \emph{DP problems}, and the set
$\Rules$ of rules cannot be altered.  They also fix assumptions on
dependency chains -- such as minimality \cite{kus:iso:sak:bla:09} or
being ``tagged'' \cite{kop:raa:12} -- which frustrate extendability
and are more naturally dealt with in a DP framework using flags.

The static DP approach for higher-order term rewriting
is discussed in, e.g.,
\cite{kus:iso:sak:bla:09,sak:kus:05,suz:kus:bla:11}.
The approach is limited to \emph{plain function passing (PFP)}
systems. The definition of PFP has been made more liberal
in later papers, but always concerns the position of higher-order
variables in the left-hand sides of rules.  These works include
non-pattern HRSs \cite{kus:iso:sak:bla:09,suz:kus:bla:11}, which we
do not consider, but do not employ formative rules or meta-variable
conditions, or consider non-termination, which we do.
Importantly, they do not consider strictly positive
inductive types, which could be used to significantly broaden the PFP
restriction.  Such types \emph{are} considered in an early paper which
defines a variation of static higher-order dependency
pairs~\cite{bla:06}
based on a computability closure~\cite{bla:jou:oka:02,bla:16}.
However, this work carries different restrictions (e.g., DPs
must be type-preserving and not introduce fresh variables) and
considers only one analysis technique (reduction pairs).

Definitions of DP approaches for \emph{functional programming} also
exist \cite{kus:13,kus:18}, which consider applicative systems with
ML-style polymorphism.  These works also employ a much broader,
semantic definition than PFP, which is actually more general than the
syntactic restriction we propose here.  However, like the static
approaches for term rewriting, they do not truly exploit the
computability \cite{tai:67} properties inherent in this restriction:
it is only used for the initial generation of dependency pairs.
In the present work, we will take advantage of our exact
computability notion by introducing a $\static$ flag that can be
used by the computable subterm criterion processor
(\refThm{thm:staticsubtermproc}) to handle benchmark systems that
would otherwise be beyond the reach of static DPs.  Also in these
works, formative rules, meta-variable conditions and non-termination
are not considered.

Regarding \emph{dynamic} DP approaches, a precursor of the present
work is \cite{kop:raa:12}, which provides a halfway framework
(methodology to prove ``chain-freeness'') for dynamic DPs, introduces
a notion of formative rules, and briefly translates a basic form of
static DPs to the same setting.  Our formative \emph{reductions}
consider the shape of reductions rather than the rules they use, and
they can be used as a flag in the framework to gain additional power
in other processors.
The adaptation of static DPs in \cite{kop:raa:12} was very limited,
and did not for instance consider strictly positive inductive types
or rules of functional type.

For a more elaborate discussion of both static and dynamic DP
approaches in the literature, we refer to \cite{kop:raa:12}
and the second author's PhD thesis~\cite{kop:12}.

\smallskip
\emph{Organisation of the paper.}
\refSec{sec:prelim} introduces
higher-order rewriting using AFSMs and recapitulates computability.
In \refSec{sec:restrictions} we impose restrictions on the input AFSMs
for which our framework is soundly applicable.
In \refSec{sec:dp} we define
static DPs for AFSMs, and derive the
key results on them.
\refSec{sec:framework} formulates the DP framework and a number of DP
processors for existing and new termination proving techniques.
\refSec{sec:conclusions} concludes.
Detailed proofs for all results in this paper
\onlyarxiv{(extending \cite{esop2019})}
and an experimental evaluation
are available in
\onlyarxiv{the appendix}\onlypaper{a technical report
\cite{technicalreport}}.
In addition, many of the results have been informally
published in the second author's PhD thesis \cite{kop:12}.

\section{Preliminaries}
\label{sec:prelim}

In this section, we first define our notation by introducing the AFSM
formalism.  Although not one of the standards of higher-order rewriting,
AFSMs combine features from various forms of higher-order rewriting and
can be seen as a form of IDTSs~\cite{bla:00} which includes application.
We will finish with a definition of \emph{computability}, a technique
often used for higher-order termination methods.

\subsection{Higher-order term rewriting using AFSMs}

Unlike first-order term rewriting, there is no single, unified approach
to higher-order term rewriting, but rather a number of
similar but not fully compatible systems
aiming to combine term rewriting and typed $\lambda$-calculi.
For generality, we will use
\emph{Algebraic Functional Systems with Meta-variables}: a formalism
which admits translations from the main formats of higher-order
term rewriting.

\begin{definition}[Simple types]
We fix a set $\Sorts$ of \emph{sorts}.
All sorts are simple types, and if $\atype,\btype$ are simple types,
then so is $\atype \arrtype \btype$.
\end{definition}

We let $\arrtype$ be right-associative.
Note that all types have a unique representation in the form
$\atype_1 \arrtype \dots \arrtype \atype_\maa \arrtype \asort$
with $\asort \in \Sorts$.

\begin{definition}[Terms and meta-terms]\label{def:terms}
We fix disjoint sets $\F$ of \emph{function symbols}, $\V$ of
\emph{variables} and $\M$ of \emph{meta-variables}, each symbol
equipped with a type.  Each meta-variable is additionally equipped
with a natural number.  We assume that both $\V$ and $\M$ contain infinitely
many symbols of all types.
The set $\Terms(\F,\V)$ of \emph{terms} over $\F,\V$
consists of
expressions $s$ where $s : \atype$ can be derived
for some type $\atype$ by the following clauses:

\begin{tabular}{llllcllll}
(\textsf{V}) & $x : \atype$ & if & $x : \atype \in \V$ & \quad\quad &
(\textsf{@}) & $\app{s}{t} : \btype$ & if & $s : \atype \arrtype
  \btype$ and $t : \atype$ \\
(\textsf{F}) & $\afun : \atype$ & if & $\afun : \atype \in \F$ & &
($\mathsf{\Lambda}$) & $\lambda x.s : \atype \arrtype \btype$ & if &
  $x : \atype \in \V$ and $s : \btype$
\end{tabular}

\noindent
\emph{Meta-terms} are expressions whose type can be derived by those
clauses and:

\begin{tabular}{ll}
(\textsf{M}) & $\meta{Z}{s_1,\dots,s_\mia} : \atype_{\mia+1}
  \arrtype \dots \arrtype \atype_\maa \arrtype \asort$ \\
  & if $Z :
  (\atype_1 \arrtype \dots \arrtype \atype_{\mac} \arrtype \dots
  \arrtype \atype_{\maa} \arrtype \asort,\ 
  \mia) \in \M$ and $s_1 : \atype_1,\dots,s_\mac : \atype_\mac$
\end{tabular}

\noindent
The $\lambda$ binds variables as in the $\lambda$-calculus; unbound
variables are called \emph{free}, and $\FV(s)$ is the set of free
variables in $s$.
Meta-variables cannot be bound; we write $\FMV(s)$ for the set
of meta-variables occurring in $s$.
A meta-term $s$ is called
\emph{closed} if $\FV(s) = \emptyset$ (even if $\FMV(s) \neq
\emptyset$).
Meta-terms are considered modulo $\alpha$-conversion.
Application (\textsf{@}) is left-associative;
abstractions ($\mathsf{\Lambda}$) extend as far to the right
as possible.
A meta-term $s$ \emph{has type} $\atype$ if
$s : \atype$;
it \emph{has base type} if $\atype \in \Sorts$.
We define $\head(s) = \head(s_1)$ if $s = \app{s_1}{s_2}$,
and $\head(s) = s$ otherwise.

A (meta-)term $s$ has a \emph{sub-(meta-)term} $t$,
notation $s \suptermeq t$, if either $s = t$ or $s \supterm t$,
where $s \supterm t$ if
(a) $s = \abs{x}{s'}$ and $s' \suptermeq t$, (b)
$s = \app{s_1}{s_2}$ and $s_2 \suptermeq t$ or (c)
$s = \app{s_1}{s_2}$ and $s_1 \suptermeq t$.
A (meta-)term $s$ has a \emph{fully applied sub-(meta-)term} $t$,
notation $s
\bettersuptermeq t$, if either $s = t$ or $s \bettersupterm t$,
where $s \bettersupterm t$ if (a) $s = \abs{x}{s'}$ and $s'
\bettersuptermeq t$, (b) $s = \app{s_1}{s_2}$ and $s_2
\bettersuptermeq t$ or (c) $s = \app{s_1}{s_2}$ and $s_1
\bettersupterm t$ (so if $s = \app{\app{x}{s_1}}{s_2}$, then $x$
and $\app{x}{s_1}$ are not fully applied subterms, but
$s$ and both $s_1$ and $s_2$ are).

For $Z : (\atype,\mia) \in \M$, we call $\mia$ the
\emph{arity} of $Z$, notation $\arity(Z)$.
\end{definition}

Clearly, all fully applied subterms are subterms, but not all subterms
are fully applied.
Every term $s$ has a form $\apps{t}{s_1}{s_n}$ with $n \geq 0$
and $t = \head(s)$ a variable, function symbol, or abstraction; in
meta-terms $t$ may also be a meta-variable application $\meta{F}{s_1,
\dots,s_\mia}$.
\emph{Terms} are the objects that we will rewrite;
\emph{meta-terms} are used to define rewrite rules.
Note that all our terms (and meta-terms) are, by definition,
well-typed.  For rewriting, we will employ \emph{patterns}:

\begin{definition}[Patterns]\label{def:pattern}
A meta-term is a \emph{pattern} if it has one of the forms $\meta{Z}{x_1,
\dots,x_\mia}$ with all $x_i$ distinct variables;
$\abs{x}{\ell}$ with $x \in \V$ and $\ell$ a pattern;
or $\apps{a}{\ell_1}{\ell_n}$ with $a \in \F \cup \V$ and all
$\ell_i$ patterns ($n \geq 0$).

\end{definition}

In rewrite rules, we will use meta-variables for \emph{matching}
and variables only with \emph{binders}.  In terms, variables can occur
both free and bound, and meta-variables cannot occur.
Meta-variables originate in very early forms of higher-order rewriting
(e.g., \cite{acz:78,klo:oos:raa:93}), but have also been used in later
formalisms (e.g.,~\cite{bla:jou:oka:02}).  They strike a balance between
matching modulo $\beta$ and syntactic matching.
By \linebreak
using meta-variables, we obtain the same expressive power as
with Miller patterns~\cite{mil:91}, but
do so without including a reversed $\beta$-reduction as part of matching.

\medskip
\emph{Notational conventions:}
We will use $x,y,z$ for variables, $X,Y,Z$ for
meta-variables, $\avarormeta$ for symbols that could be variables or
meta-variables, $\afun,\bfun,\cfun$ or more suggestive notation for
function symbols, and $s,t,u,v,q,w$ for (meta-)terms.  Types are denoted
$\sigma,\tau$, and $\asort,\bsort$ are sorts.
We will regularly overload notation and write $x \in \V$, $\afun \in \F$
or $Z \in \M$ without stating a type (or minimal arity).  For
meta-terms $\meta{Z}{}$ we will usually omit the brackets, writing
just $Z$.

\begin{definition}[Substitution]
A \emph{meta-substitution} is a type-preserving function $\gamma$
from variables and meta-variables to meta-terms.
Let the \emph{domain} of $\gamma$ be given by:
$\domain(\gamma) = \{ (x : \atype) \in \V \mid \gamma(x) \neq x \} \cup
\{ (Z : (\atype,\mia)) \in \M \mid  \gamma(Z) \neq \abs{y_1 \dots
y_{\mia}}{\meta{Z}{y_1,\dots,y_{\mia}}} \}$; this domain is allowed
to be infinite.
We let $[\avarormeta_1:=s_1,\dots,\avarormeta_n:=s_n]$ denote the
meta-substitution $\gamma$ with $\gamma(\avarormeta_i) = s_i$ and
$\gamma(z) = z$ for $(z : \atype) \in \V \setminus \{\avarormeta_1,
\dots,\avarormeta_n\}$, and $\gamma(Z) = \abs{y_1\dots y_{\mia}}{
\meta{Z}{y_1,\dots, y_{\mia}}}$ for $(Z : (\atype,\mia)) \in \M
\setminus \{\avarormeta_1,\dots,\avarormeta_n\}$.
We assume there are infinitely many variables $x$ of all types such
that (a) $x \notin \domain(\gamma)$ and (b) for all $\avarormeta \in
\domain(\gamma)$: $x \notin \FV(\gamma(\avarormeta))$.

A \emph{substitution} is a meta-substitution mapping everything in its
domain to terms.
The result $s\gamma$ of applying a meta-substitution $\gamma$ to a
term $s$ is obtained by:

\noindent
\begin{tabular}{rclllcrclll}
$x\gamma$ & $=$ & $\gamma(x)$ & if & $x \in \V$ & \quad\quad &
$(s\ t)\gamma$ & $=$ & $(s\gamma)\ (t\gamma)$ \\
$\afun\gamma$ & $=$ & $\afun$ & if & $\afun \in \F$ & &
$(\lambda x.s)\gamma$ & $=$ & $\lambda x.(s\gamma)$ & if &
  $\gamma(x) = x
  \wedge x \notin \bigcup_{y \in \domain(\gamma)} \FV(\gamma(y))$
\end{tabular}

For meta-terms, the result $s\gamma$ is obtained by the clauses above
and:

\noindent
\begin{tabular}{rcl}
$\meta{Z}{s_1,\dots,s_\mac}\gamma$ & $=$ & 
  $\meta{\gamma(Z)}{s_1\gamma,\dots,s_\mac\gamma}$ \quad
  if $Z \notin \domain(\gamma)$ \\
$\meta{Z}{s_1,\dots,s_\mac}\gamma$ & $=$ & 
  $\metaapply{\gamma(Z)}{s_1\gamma,\dots,s_\mac\gamma}$ \quad
  if $Z \in \domain(\gamma)$ \\
$\metaapply{(\abs{x_1 \dots x_\mac}{s})}{t_1,\dots,t_\mac}$ & $=$ &
  $s[x_1:=t_1,\dots,x_\mac:=t_\mac]$ \\
$\metaapply{(\abs{x_1 \dots x_n}{s})}{t_1,\dots,t_\mac}$ & $=$ &
  $\apps{s[x_1:=t_1,\dots,x_n:=t_n]}{t_{n+1}}{t_\mac}$ \quad
  if $n < \mac$ \\
  & & \phantom{x}\hfill and $s$ is not an abstraction \\
\end{tabular}
\end{definition}

Note that for fixed $\mac$, any term
has exactly one of the
two forms above ($\abs{x_1 \dots x_n}{s}$ with $n < \mac$ and $s$ not
an abstraction, or $\abs{x_1 \dots x_\mac}{s}$).

\medskip
Essentially, applying a meta-substitution that has meta-variables in
its domain combines a substitution with
(possibly several) $\beta$-steps.
For example, we have that:\linebreak
$\symb{deriv}\ (\abs{x}{\symb{sin}\ (\meta{F}{x})})[F
:=\abs{y}{\symb{plus}\ y\ x}]$ equals $\symb{deriv}\ 
(\abs{z}{\symb{sin}\ (\symb{plus}\ z\ x}))$.
We also have:
$\meta{X}{\symb{0},\nil}[X:=\abs{x}{\map\ (\abs{y}{x})}]$ equals
$\symb{map}\ (\abs{y}{\symb{0}})\ \nil$.

\begin{definition}[Rules and rewriting]\label{def:rule}
Let $\F,\V,\M$ be fixed sets of function symbols, variables
and meta-variables respectively.
A \emph{rule} is a pair $\ell \arrz r$ of closed meta-terms of the
same type such that $\ell$ is a pattern of the form
$\apps{\afun}{\ell_1}{\ell_n}$
with $\afun \in \F$ and $\FMV(r) \subseteq \FMV(\ell)$.
A set of rules $\Rules$ defines a rewrite
relation $\arr{\Rules}$
as the smallest monotonic relation on terms which includes:

\noindent
\begin{tabular}{l@{\hskip 2pt}c@{\hskip 8pt}l@{\hskip 8pt}cl}
(\textsf{Rule}) &
  $\ell\delta$ & $\arr{\Rules}$ & $r\delta$ & if $\ell \arrz r \in
  \Rules$ and $\domain(\delta) = \FMV(\ell)$ \\
(\textsf{Beta}) & $(\abs{x}{s})\ t$ & $\arr{\Rules}$ & $s[x:=t]$ \\
\end{tabular}

\noindent
We say $s \arr{\beta} t$ if $s \arr{\Rules} t$ is derived using a
(\textsf{Beta}) step.
A term $s$ is \emph{terminating} under $\arr{\Rules}$ if there is no
infinite reduction $s = s_0 \arr{\Rules} s_1 \arr{\Rules} \dots$,
is \emph{in normal form} if there is no $t$ such that $s \arr{\Rules}
t$, and is \emph{$\beta$-normal} if there is no $t$
with $s \arr{\beta} t$.
Note that we are allowed to reduce at any position of a term, even
below a $\lambda$.
The relation $\arr{\Rules}$ is terminating if all terms over $\F,\V$
are terminating.
The set $\Defineds \subseteq \F$ of \emph{defined symbols} consists
of those
$(\afun : \atype) \in \F$ such that a rule $\apps{\afun}{\ell_1}{
\ell_n} \arrz r$ exists;
all other
symbols are called \emph{constructors}.
\end{definition}

Note that $\Rules$ is allowed to be infinite, which is useful for
instance to model polymorphic systems.  Also, right-hand sides of
rules do not have to be in $\beta$-normal form.
While this is rarely used in practical examples, non-$\beta$-normal
rules may arise through transformations,
and we lose nothing by allowing them.

\begin{example}\label{ex:mapintro}
Let $\F \supseteq \{ \nul : \nat,\ \suc : \nat \arrtype \nat,\ \nil :
\lijst,
\cons : \nat \arrtype \lijst \arrtype \lijst,\ 
\map : (\nat \arrtype \nat) \arrtype \lijst \arrtype \lijst\}$ and
consider the following rules $\Rules$:
\vspace{-2pt}
\[
\begin{array}{rcl}
\map\ (\abs{x}{\meta{Z}{x}})\ \nil & \arrz & \nil \\
\map\ (\abs{x}{\meta{Z}{x}})\ (\cons\ H\ T) & \arrz &
  \cons\ \meta{Z}{H}\ (\map\ (\abs{x}{\meta{Z}{x}})\ T) \\
\end{array}
\vspace{-2pt}
\]
Then $\map\ (\abs{y}{\nul})\ (\cons\ (\suc\ \nul)\ \nil) \arr{\Rules} 
\cons\ \nul\ (\map\ (\abs{y}{\nul})\ \nil) \arr{\Rules} \cons\ \nul\ 
\nil$.
Note that the
bound variable $y$ does not need to occur in the body of $\abs{y}{\nul}$
to match $\abs{x}{\meta{Z}{x}}$.
However,
a term like
$\map\ \suc\ (\cons\ \nul\ \nil)$ \emph{cannot} be reduced, because
$\suc$ does not instantiate $\abs{x}{\meta{Z}{x}}$. We could
alternatively consider the rules:
\vspace{-5pt}
\[
\begin{array}{rcl}
\map\ Z\ \nil & \arrz & \nil \\
\map\ Z\ (\cons\ H\ T) & \arrz & \cons\ (Z\ H)\ (\map\ Z\ T) \\
\end{array}
\]
Where the system before had $(Z : (\nat \arrtype \nat, 1)) \in \M$,
here we assume $(Z : (\nat \arrtype \nat, 0)) \in \M$.
Thus, rather than meta-variable application $\meta{Z}{H}$ we use
explicit application $\app{Z}{H}$.
Then $\map\ \suc\ (\cons\ \nul\ \nil) \arr{\Rules} \cons\ (\suc\ 
\nul)\ (\map\ \suc\ \nil)$.  However, we will often need explicit
$\beta$-reductions;
e.g., $\map\ (\abs{y}{\nul})\ (\cons\ (\suc\ \nul)\ \nil)
\linebreak
\arr{\Rules}
\cons\ ((\abs{y}{\nul})\ (\suc\ \nul))\ (\map\ (\abs{y}{\nul})\ \nil)
\arr{\beta} \cons\ \nul\ (\map\ (\abs{y}{\nul})\ \nil)$.
\end{example}

\begin{definition}[AFSM]\label{def:afsm}
An \emph{AFSM} is a tuple $(\F,\V,\M,\Rules)$ of a signature
and a set of rules built from meta-terms over $\F,\V,\M$; as types of relevant
variables and meta-variables can always be derived from context, we
will typically just refer to the AFSM $(\F,\Rules)$.
An AFSM implicitly defines the abstract reduction system
$(\Terms(\F,\V), \arr{\Rules})$:
a set of terms and a rewrite
relation on this set.
An AFSM is terminating if $\arr{\Rules}$ is terminating (on all terms
in $\Terms(\F,\V)$).
\end{definition}

\emph{Discussion:}
The two most common
formalisms in
termination analysis of higher-order rewriting
are \emph{algebraic functional systems}~\cite{jou:rub:99} (AFSs)
and \emph{higher-order rewriting  systems}~\cite{nip:91,mil:91} (HRSs).
AFSs are very similar to our AFSMs, but
use variables
for matching rather than meta-variables; this
is
trivially
translated to the AFSM format, giving rules where all meta-variables
have
arity $0$, like the ``alternative'' rules in \refEx{ex:mapintro}.
HRSs use matching
modulo $\beta/\eta$, but the common restriction of \emph{pattern HRSs}
can be directly translated into AFSMs,
provided
terms are
$\beta$-normalised after every reduction step.
Even without this $\beta$-normalisation step, termination of the
obtained AFSM implies termination of the original HRS;
for second-order systems, termination is equivalent.
AFSMs can also naturally encode CRSs~\cite{klo:oos:raa:93}
and several applicative systems (cf. \cite[Chapter 3]{kop:12}).

\begin{example}[Ordinal recursion]\label{ex:ordrec}
A running example
is the AFSM
$(\F,\Rules)$ with
$\F \supseteq \{
\nul : \symb{ord},\ \suc : \symb{ord} \arrtype \symb{ord},
\symb{lim} : (\symb{nat} \arrtype \symb{ord}) \arrtype \symb{ord},\ 
\symb{rec} : \symb{ord} \arrtype \symb{nat} \arrtype (\symb{ord}
\arrtype \symb{nat} \arrtype \symb{nat}) \arrtype ((\symb{nat}
\arrtype \symb{ord}) \arrtype (\symb{nat} \arrtype \symb{nat}) \arrtype
\symb{nat}) \arrtype \symb{nat} \}$
and $\Rules$ given below.
As all meta-variables have arity $0$, this can be seen as an AFS.
\[
\begin{array}{rcl}
\symb{rec}\ \nul\ K\ F\ G & \arrz & K \\
\symb{rec}\ (\suc\ X)\ K\ F\ G & \arrz & F\ X\ (\symb{rec}\ X\ K\ F\ G) \\
\symb{rec}\ (\symb{lim}\ H)\ K\ F\ G & \arrz & G\ H\ (\abs{m}{
  \symb{rec}\ (H\ m)\ K\ F\ G}) \\
\end{array}
\]
\end{example}

Observant readers may notice that by the given constructors, the
type $\nat$ in \refEx{ex:ordrec} is not inhabited.  However,
as the given symbols are only a subset of $\F$, additional
symbols (such as constructors for the $\symb{nat}$ type) may be
included.
The presence of additional function symbols does not affect
termination of AFSMs:

\begin{theorem}[Invariance of termination under signature extensions]
\label{thm:invar}
For an AFSM $(\F,\Rules)$ with $\F$ at most countably
  infinite, let $\sig{\Rules} \subseteq \F$ be the set of function
  symbols occurring in some
  rule of $\Rules$. Then
  $(\Terms(\F,\V), \arr{\Rules})$ is terminating if and only if
  $(\Terms(\sig{\Rules},\V), \arr{\Rules})$ is terminating.
\end{theorem}

\begin{proof}
Trivial by replacing all function symbols in $\F \setminus
\sig{\Rules}$ by corresponding variables of the same type.
\qed
\end{proof}

Therefore,
we will typically only state the types of symbols occurring in the
rules, but may safely assume that infinitely many symbols of all types
are present (which for instance allows us to select unused constructors
in some proofs).

\subsection{Computability}\label{subsec:computability}

A common technique in higher-order termination is Tait and Gi\-rard's
\emph{computability} notion~\cite{tai:67}.
There are several ways to define computability predicates; here we
follow, e.g., \cite{bla:00,bla:jou:oka:02,bla:jou:rub:15,bla:16} 
in considering \emph{accessible meta-terms} using strictly positive
inductive types.
The definition
presented
below is adapted from these works, both to
account for the altered formalism and to introduce (and obtain
termination of) a relation $\accreduce{C}$ that we will use in
the ``computable subterm criterion processor'' of
\refThm{thm:staticsubtermproc}
(a termination criterion that allows us to handle systems
that would otherwise be beyond the reach of static DPs).
This allows for a minimal presentation that avoids the use of ordinals
that would otherwise be needed to obtain $\accreduce{C}$
(see, e.g., \cite{bla:jou:rub:15,bla:16}).

To define computability, we use the notion of an \emph{RC-set}:

\edef\defRCset{\number\value{theorem}}
\edef\defRCsetSec{\number\value{section}}
\begin{definition}\label{def:RCset}
A \emph{set of reducibility candidates}, or \emph{RC-set}, for
a rewrite relation $\arr{\Rules}$ of an AFSM
is a set $I$ of base-type terms $s$ such that:
every term in $I$ is terminating under $\arr{\Rules}$;
$I$ is closed under $\arr{\Rules}$ (so if $s \in I$ and $s
  \arr{\Rules} t$ then $t \in I$);
if $s = \apps{x}{s_1}{s_n}$ with $x \in \V$ or $s =
  \apps{(\abs{x}{u})}{s_0}{s_n}$ with $n \geq 0$, and for all $t$ with
  $s \arr{\Rules} t$ we have $t \in I$, then $s \in I$
  (for any $u,s_0,\dots,s_n \in \Terms(\F,\V)$).

We define $I$-computability for an RC-set $I$ by induction on types.
For $s \in \Terms(\F,\V)$, we say that $s$ is $I$-computable if
either
$s$ is of base type and $s \in I$;
or
 $s : \atype \arrtype \btype$ and for all
  $t : \atype$ that are $I$-computable, $\app{s}{t}$ is $I$-computable.
\end{definition}

The traditional notion of computability is obtained by taking for $I$
the set of
all terminating base-type terms.
Then, a term $s$ is computable if and only if (a) $s$ has
base type and is terminating; or (b) $s : \atype \arrtype \btype$ and
for all computable $t : \atype$ the term $\app{s}{t}$ is computable.
This choice is simple but, for reasoning, not ideal: we
do not have a property like: ``if $\symb{f}\ s_1 \cdots s_n$ is
computable then so is each $s_i$''.
Such a property would be valuable to have for generalising termination
proofs from first-order to higher-order rewriting, as it allows us to
use computability where the first-order proof uses termination.
While it is not possible to define a computability notion with this
property alongside case (b) (as such a notion would not be well-founded),
we can come \emph{close} to this property by choosing a different set
for $I$.  To define this set, we will use the notion of
\emph{accessible arguments}, which is used for the same purpose also
in the
\emph{General Schema} \cite{bla:jou:oka:02}, the
\emph{Computability Path Ordering}~\cite{bla:jou:rub:15}, and the
\emph{Computability Closure}~\cite{bla:16}.

\begin{definition}[Accessible arguments]\label{def:accArgs}
We fix a quasi-ordering $\greqsort$ on $\Sorts$ with well-founded
strict part $\grsort\ :=\ \greqsort \setminus
\leqsort$.\footnote{Well-foundedness is immediate if
  $\Sorts$ is finite, but we have not imposed that requirement.}
For a type
$\atype \equiv \atype_1 \!\arrtype\! \dots \!\arrtype\! \atype_\maa
\!\arrtype\! \bsort$ (with $\bsort \in \Sorts$) and sort $\asort$,
let $\asort \gracsortup \atype$ if $\asort \greqsort \bsort$ and
$\asort \gracsortdown \atype_i$ for all $i$, and let
$\asort \gracsortdown \atype$ if $\asort \grsort \bsort$ and $\asort
\gracsortup \atype_i$ for all $i$.\footnote{Here
  $\asort \gracsortup \atype$ corresponds to
  ``$\asort$ occurs only positively in $\atype$'' in
  \cite{bla:00,bla:jou:oka:02,bla:jou:rub:15}.}

For $\afun : \atype_1 \arrtype \dots \arrtype \atype_\maa \arrtype
\asort \in \F$, let $\Acc(\afun) = \{ i \mid 1 \leq i \leq \maa
\wedge \asort \gracsortup \atype_i \}$.
For $x : \atype_1 \arrtype \dots \arrtype \atype_\maa \arrtype \asort
\in \V$, let $\Acc(x) = \{ i \mid 1 \leq i \leq \maa \wedge \atype_i$
has the form $\btype_1 \arrtype \dots \arrtype \btype_n \arrtype
\bsort$ with $\asort \greqsort \bsort \}$.
We write $s \gracc t$ if either $s = t$, or $s = \abs{x}{s'}$ and $s'
\gracc t$, or $s = \apps{a}{s_1}{s_n}$ with $a \in \F \cup \V$ and
$s_i \gracc t$ for some $i\in \Acc(a)$
with $a \notin \FV(s_i)$.
\end{definition}

With this definition, we will be able to define a set $C$ such that,
roughly, $s$ is $C$-computable if and only if (a) $s : \atype \arrtype
\btype$ and $s\ t$ is $C$-computable for all $C$-computable $t$, or
(b) $s$ has base type, is terminating, and if $s = \apps{\afun}{s_1}{
s_m}$ then $s_i$ is $C$-computable for all \emph{accessible} $i$ (see
\refThm{thm:defC} below).
The reason that $\Acc(x)$ for $x \in \V$ is different is
proof-technical:
computability of $\abs{x}{\apps{x}{s_1}{s_\maa}}$ implies the
computability of more arguments $s_i$ than computability of
$\apps{\afun}{s_1}{s_\maa}$ does, since $x$ can be instantiated by
anything.

\begin{example}\label{ex:lim}
Consider a quasi-ordering $\greqsort$ such that
$\symb{ord} \grsort \symb{nat}$. In \refEx{ex:ordrec}, we then have
$\symb{ord} \: \gracsortup \: \symb{nat} \arrtype \symb{ord}$.
Thus, $1 \in \Acc(\symb{lim})$, which gives $\symb{lim}\ H \gracc H$.
\end{example}

\begin{theorem}\label{thm:defCoriginal}\label{thm:defC}
Let $(\F,\Rules)$ be an AFSM.
Let $\apps{\afun}{s_1}{s_\maa} \accreduce{I} \apps{s_i}{t_1}{t_n}$ if
both sides have base type, $i \in \Acc(\afun)$,
and all $t_j$ are $I$-computable.
There is an RC-set $C$ such that $C = \{ s \in \Terms(\F,\V) \mid
s$ has base type $\wedge\ s$ is terminating under
$\arr{\Rules} \cup \accreduce{C} \wedge {}$
if $s \arrr{\Rules} \apps{\afun}{s_1}{s_\maa}$
then $s_i$ is
$C$-computable for all $i \in \Acc(\afun) \}$.
\end{theorem}

\begin{proof}[sketch]
Note that we cannot \emph{define} $C$ as this set, as the set
relies on the notion of $C$-computability.  However, we
\emph{can} define $C$ as the
fixpoint of a monotone function operating on RC-sets.  This follows the
proof in, e.g., \cite{bla:jou:oka:02,bla:jou:rub:15}.
\qed

\onlyarxiv{The full proof (for the definitions in this paper) is available in
\refApp{app:computability}.}%
\onlypaper{The complete proof is available in
  \cite[Appendix~A]{technicalreport}.}
\end{proof}

\section{Restrictions}
\label{sec:restrictions}

The termination methodology in this paper is restricted to AFSMs that
satisfy certain limitations: they must be \emph{properly applied}
(a restriction on the number of terms each function symbol is applied
to) and \emph{accessible function passing} (a restriction on the
positions of variables of a functional type in the left-hand sides of
rules).  Both are syntactic restrictions that are easily checked by a
computer (mostly; the latter requires a search for a sort
ordering, but this is typically easy).

\subsection{Properly applied AFSMs}\label{subsec:properlyapplied}

In \emph{properly applied AFSMs}, function symbols are assigned a
certain, minimal number of arguments that they must always be applied
to.

\begin{definition}
An AFSM $(\F,\Rules)$ is \emph{properly applied} if for every
$\afun\in \Defineds$ there exists an integer $\mia$ such that
for all rules $\ell \arrz r \in \Rules$:
(1)
  if $\ell = \apps{\afun}{\ell_1}{\ell_n}$ then $n = \mia$;
and (2)
  if $r \bettersuptermeq \apps{\afun}{r_1}{r_n}$ then $n \geq \mia$.
We denote $\minarity{\Rules}(\afun) = \mia$.
\end{definition}

That is, every occurrence of a function symbol in the \emph{right-hand}
side of a rule has at least as many arguments as the occurrences in
the \emph{left-hand} sides of rules.  This means that partially applied
functions are often not allowed: an AFSM with rules such as
$\symb{double}\ X \arrz \symb{plus}\ X\ X$ and $\symb{doublelist}\ L
\arrz \symb{map}\ \symb{double}\ L$ is not properly applied, because
$\symb{double}$ is applied to one argument in the left-hand side of
some rule, and to zero in the right-hand side of another.

This restriction is not as severe as it may initially seem since partial
applications can be replaced by $\lambda$-abstractions; e.g., the rules
above can be made properly applied by replacing the second rule by:
$\symb{doublelist}\ L \arrz \symb{map}\ (\abs{x}{\symb{double}\ x})\ L$.
By using $\eta$-expansion, we can transform any AFSM to satisfy this
restriction:

\begin{definition}[$\Rules^\uparrow$]\label{def:etalong}
Given a set of rules
$\Rules$, let their \emph{$\eta$-expansion} be given by
$\Rules^\uparrow = \{ \etalong{(\apps{\ell}{Z_1}{Z_\maa})}
\ \arrz \etalong{(\apps{r}{Z_1}{Z_\maa})}
\mid \ell \arrz r \in \Rules$ with $r : \atype_1 \arrtype \dots \arrtype
\atype_\maa \arrtype \asort$,
$\asort \in \Sorts$,
and
$Z_1,\dots, Z_\maa$ fresh meta-variables$\}$,
where
\begin{itemize}
\item $\etalong{s} = \abs{x_1 \dots x_\maa}{\apps{\halfetalong{s}}{
  (\etalong{x_1})}{(\etalong{x_\maa})}}$ if $s$ is an application or
  element of $\V \cup \F$, and $\etalong{s} = \halfetalong{s}$
  otherwise;
\item $\halfetalong{\afun} = \afun$ for $\afun \in \F$
  and $\halfetalong{x} = x$ for $x \in \V$, while
  $\halfetalong{\meta{Z}{s_1,\dots,s_\mac}} =
  \meta{Z}{\halfetalong{s_1},\dots,\halfetalong{s_\mac}}$ and
  $\halfetalong{(\abs{x}{s})} = \abs{x}{(\etalong{s})}$ and
  $\halfetalong{s_1\ s_2} = \halfetalong{s_1}\ (\etalong{s_2})$.
\end{itemize}
\end{definition}

Note that $\etalong{\ell}$ is a pattern if $\ell$ is.  By
\cite[Thm.~2.16]{kop:12}, a relation $\arr{\Rules}$ is
terminating if $\arr{\Rules^\uparrow}$ is terminating, which allows
us to transpose any methods to prove termination of properly
applied AFSMs to all AFSMs.

However, there is a caveat: this transformation can introduce
non-termination in some special cases,
e.g.,
the terminating rule
$\symb{f}\ X \arrz \symb{g}\ \symb{f}$ with $\symb{f} : \symb{o}
\arrtype \symb{o}$ and $\symb{g} : (\symb{o} \arrtype \symb{o})
\arrtype \symb{o}$, whose $\eta$-expansion
$\symb{f}\ X \arrz \symb{g}\ (\abs{x}{(\symb{f}\ x)})$ is
non-terminating.
Thus, for a properly applied AFSM the methods in this paper
apply directly.  For an AFSM that is not properly applied, we can use
the methods to prove \emph{termination} (but not non-termination) by
first $\eta$-expanding the rules.  Of course, if this analysis leads to
a \emph{counterexample} for termination, we may still be able to verify
whether this counterexample applies in the original, untransformed
AFSM.

\begin{example}
Both AFSMs in \refEx{ex:mapintro} and the AFSM in
\refEx{ex:ordrec} are properly applied.
\end{example}

\begin{example}\label{ex:deriv}
Consider an AFSM $(\F,\Rules)$ with $\F \supseteq \{ \symb{sin},
\symb{cos} : \real \arrtype \real,\linebreak \symb{times} :
\real \arrtype \real \arrtype \real,\  \deriv
: (\real \arrtype \real) \arrtype \real \arrtype
\real \}$ and $\Rules = \{
\deriv\ (\abs{x}{\symb{sin}\ \meta{F}{x}})
\arrz
  \abs{y}{\symb{times}\ (\deriv\ (\abs{x}{\meta{F}{x}})\ y)\ 
  (\symb{cos}\ \meta{F}{y})} 
\}
$.
Although the one rule has a functional output type ($\real \arrtype
\real$), this AFSM is properly applied, with $\deriv$ having always
at least $1$ argument.  Therefore, we do not need to use $\Rules^\uparrow$.
However, if $\Rules$ were to additionally include some rules that did
not satisfy the restriction (such as the $\symb{double}$ and
$\symb{doublelist}$ rules above), then $\eta$-expanding \emph{all}
rules, including this one, would be necessary.  We have:
$\Rules^\uparrow =
\{ \deriv\ (\abs{x}{\symb{sin}\ \meta{F}{x}})\ Y \arrz
(\abs{y}{\symb{times}\ 
(\deriv\ (\abs{x}{\meta{F}{x}})\ y)\ (\symb{cos}\ \meta{F}{y})})\ 
Y \}$.
Note that the right-hand side of the $\eta$-expanded $\deriv$ rule
is not $\beta$-normal.
\end{example}

\subsection{Accessible Function Passing AFSMs}

In \emph{accessible function passing} AFSMs, variables of functional
type may not occur at arbitrary places in the left-hand sides of rules:
their positions are restricted using the sort ordering
$\greqsort$ and
accessibility relation $\gracc$ from \refDef{def:accArgs}.

\begin{definition}[Accessible function passing]\label{def:apfp}
An AFSM $(\F,\Rules)$ is \emph{accessible function passing (AFP)}
if there exists a sort ordering $\greqsort$ following
\refDef{def:accArgs} such that:
  for all $\apps{\afun}{\ell_1}{\ell_n} \arrz r \in \Rules$ and all
  $Z \in \FMV(r)$: there are variables $x_1,\dots,x_\mia$
  and some $i$ such that $\ell_i \gracc \meta{Z}{x_1,\dots,x_\mia}$.
\end{definition}

The key idea of this definition is that computability of
each $\ell_i$ implies computability of all meta-variables in $r$.
This excludes cases like Example \ref{ex:lambdadynamic} below.
Many common examples satisfy this restriction, including those we
saw before:

\begin{example}\label{ex:afpgood}\label{ex:derivafp}
Both systems from \refEx{ex:mapintro} are AFP: choosing the sort
ordering $\greqsort$ that equates $\nat$ and $\lijst$, we indeed have
$\cons\ H\ T \gracc H$ and $\cons\ H\ T \gracc T$
(as $\Acc(\cons) = \{1,2\}$) and both
$\abs{x}{\meta{Z}{x}} \gracc \meta{Z}{x}$ and $Z \gracc Z$.
The AFSM from \refEx{ex:ordrec} is AFP because we can choose
$\symb{ord} \grsort \symb{nat}$ and have $\symb{lim}\ H \gracc H$
following \refEx{ex:lim}
(and also $\symb{s}\ X \gracc X$ and $K \gracc K,\ F \gracc
F,\ G \gracc G$).
The AFSM from \refEx{ex:deriv} is AFP, because
$\abs{x}{\ssin\ \meta{F}{x}} \gracc \meta{F}{x}$ for any
$\greqsort$: $\abs{x}{\ssin\ \meta{F}{x}} \gracc \meta{F}{x}$ because
$\ssin\ \meta{F}{x} \gracc \meta{F}{x}$ because $1 \in \Acc(\ssin)$.
\end{example}

In fact, \emph{all} first-order AFSMs (where all fully applied
sub-meta-terms of the left-hand side of a rule have base type) are
AFP
via the sort
ordering $\greqsort$ that equates all sorts.  Also (with the same
sort ordering), an AFSM $(\F,\Rules)$ is
AFP
if, for all rules $\apps{\afun}{\ell_1}{\ell_\mia} \arrz r \in
\Rules$ and all $1 \leq i \leq \mia$, we can write:
$\ell_i = \abs{x_1 \dots x_{n_i}}{\ell'}$ where $n_i \geq 0$ and all
fully applied sub-meta-terms of $\ell'$ have base type.

This covers many practical systems, although for \refEx{ex:ordrec}
we need a non-trivial sort ordering.  Also, there are AFSMs
that cannot be handled with \emph{any} $\greqsort$.

\begin{example}[Encoding the untyped $\lambda$-calculus]\label{ex:lambdadynamic}
Consider an AFSM with $\F \supseteq \{ \symb{ap} : \symb{o} \arrtype
\symb{o} \arrtype \symb{o},\ \symb{lm} : (\symb{o} \arrtype \symb{o})
\arrtype \symb{o} \}$ and $\Rules = \{ \symb{ap}\ (\symb{lm}\ F)
\arrz F \}$ (note that the only rule has type $\symb{o} \arrtype
\symb{o}$).  This AFSM is not accessible function passing, because
$\symb{lm}\ F \gracc F$ cannot hold for any $\greqsort$ (as this
would require $\symb{o} \grsort \symb{o}$).

Note that this example is also not terminating. With
$t = \symb{lm}\ (\abs{x}{\symb{ap}\ x\ x})$,
we get this self-loop as evidence:
$\symb{ap}\ t\ t\ \arr{\Rules} (\abs{x}{\symb{ap}\ x\ x})\ t \arr{\beta}
\symb{ap}\ t\ t$.
\end{example}

Intuitively: in an accessible function passing AFSM, meta-variables of
a higher type may occur only in ``safe'' places in the left-hand sides
of rules.  Rules like the ones in \refEx{ex:lambdadynamic}, where a
higher-order meta-variable is lifted out of a base-type term, are not
admitted (unless the base type is greater than the higher type).

\medskip
In the remainder of this paper, we will refer to a \emph{properly
applied, accessible function passing} AFSM as a PA-AFP AFSM.

\medskip
\emph{Discussion:}
This definition is strictly more liberal than the notions of
``plain function passing'' in both \cite{kus:iso:sak:bla:09} and
\cite{suz:kus:bla:11} as adapted to AFSMs.  The notion in
\cite{suz:kus:bla:11} largely corresponds to AFP if $\greqsort$
equates all sorts, and the HRS formalism guarantees that rules
are properly applied (in fact, all fully applied sub-meta-terms of
both left- and right-hand sides of rules have base type).  The notion
in \cite{kus:iso:sak:bla:09} is more restrictive.  The current
restriction of PA-AFP AFSMs
lets us
handle examples like
ordinal recursion (\refEx{ex:ordrec})
which are not covered by \cite{kus:iso:sak:bla:09,suz:kus:bla:11}.
However, note that \cite{kus:iso:sak:bla:09,suz:kus:bla:11} consider
a different formalism, which does take rules whose left-hand side is
not a pattern into account (which we do not consider).
Our restriction also quite resembles the ``admissible'' rules
in \cite{bla:06} which are defined using a pattern computability
closure \cite{bla:00}, but that work carries additional restrictions.

In later work~\cite{kus:13,kus:18}, K.\ Kusakari extends the static
DP approach to forms of polymorphic functional programming, with a
very liberal restriction: the definition is parametrised with an
\emph{arbitrary} RC-set and corresponding accessibility (``safety'')
notion.  Our AFP restriction is actually an instance of this condition
(although a more liberal one than the example RC-set used in
\cite{kus:13,kus:18}).  We have chosen a specific instance
because it allows us to use dedicated techniques for the RC-set; for
example, our \emph{computable subterm criterion processor}
(\refThm{thm:staticsubtermproc}).

\section{Static higher-order dependency pairs}\label{sec:dp}

To obtain sufficient criteria for both termination and
non-termination of AFSMs, we will now
transpose the definition of static
dependency pairs~\cite{bla:06,kus:iso:sak:bla:09,suz:kus:bla:11,kus:18}
to AFSMs.
In addition,
we will add the new features of \emph{meta-variable conditions},
\emph{formative reductions}, and \emph{computable chains}.
Complete versions of all proof sketches in this section are
available in \onlyarxiv{\refApp{app:staticdp}}%
\onlypaper{\cite[Appendix~B]{technicalreport}}.

Although we retain the first-order terminology of dependency
\emph{pairs}, the setting with meta-variables makes it more suitable
to define DPs as \emph{triples}.

\begin{definition}[(Static) Dependency Pair]\label{def:dp}
A \emph{dependency pair (DP)} is a triple $\ell \arrdp p\ (A)$,
where $\ell$ is a closed pattern $\apps{\afun}{\ell_1}{\ell_\mia}$,
$p$ is a closed meta-term $\apps{\bfun}{p_1}{p_n}$,
and $A$ is a set of \emph{meta-variable conditions}: pairs $Z : i$
indicating that $Z$ regards its $i^{\text{th}}$ argument.
A DP is \emph{conservative} if $\FMV(p) \subseteq \FMV(\ell)$.

A substitution $\gamma$ \emph{respects} a set of meta-variable
conditions $A$ if for all $Z : i$ in $A$ we have $\gamma(Z) =
\abs{x_1 \dots x_j}{t}$ with either $i > j$, or $i \leq j$ and $x_i \in
\FV(t)$.
DPs will be used only with substitutions that respect
their meta-variable conditions.

For $\ell \arrdp p\ (\emptyset)$ (so a DP whose set of meta-variable
conditions is empty), we often omit the third component and just
write $\ell \arrdp p$.
\end{definition}

Like the first-order setting, the static
DP approach
employs \emph{marked function symbols} to obtain meta-terms whose
instances cannot be reduced at the root.

\begin{definition}[Marked symbols]\label{def:marked}
Let $(\F,\Rules)$ be an AFSM.
Define $\F^\sharp := \F \uplus \{ \afun^\sharp : \atype \mid \afun :
\atype \in \Defineds \}$.
For a meta-term $s = \apps{\afun}{s_1}{s_\mia}$ with $\afun \in
\Defineds$ and $\mia =
\minarity{\Rules}(\afun)$, we let $s^\sharp =
\apps{\afun^\sharp}{s_1}{s_\mia}$; for $s$ of other forms $s^\sharp$
is not defined.
\end{definition}

Moreover, we will consider \emph{candidates}.  In the first-order
setting, candidate terms are subterms of the right-hand sides of rules
whose root symbol is a defined symbol. Intuitively, these
subterms correspond to function calls. In the current setting, we
have to consider also meta-variables as well as rules whose right-hand
side is not $\beta$-normal (which might arise for instance due to
$\eta$-expansion).

\begin{definition}[$\beta$-reduced-sub-meta-term,
  $\bsuptermeq{\beta}$, $\bsuptermeq{A}$]
A meta-term $s$ has a fully applied
\emph{$\beta$-reduced-sub-meta-term} $t$ (shortly, \emph{BRSMT}),
notation $s \bsuptermeq{\beta} t$, if there exists a set of
meta-variable conditions $A$
with $s \bsuptermeq{A} t$.
Here $s \bsuptermeq{A} t$ holds if:
\begin{itemize}
\item $s = t$, or
\item $s = \abs{x}{u}$ and $u \bsuptermeq{A} t$, or
\item $s = \apps{(\abs{x}{u})}{s_0}{s_n}$ and
  some $s_i \bsuptermeq{A} t$, or
  $\apps{u[x:=s_0]}{s_1}{s_n} \bsuptermeq{A} t$, or
\item $s = \apps{a}{s_1}{s_n}$ with $a \in \F \cup \V$ and
  some $s_i \bsuptermeq{A} t$, or
\item $s = \apps{\meta{Z}{t_1,\dots,t_\mac}}{s_1}{s_n}$ and
  some $s_i \bsuptermeq{A} t$, or
\item $s = \apps{\meta{Z}{t_1,\dots,t_\mac}}{s_1}{s_n}$ and $t_i
  \bsuptermeq{A} t$ for some $i \in \{1,\dots,\mac\}$ with
  $(Z : i) \in A$.
\end{itemize}
\end{definition}

Essentially, $s \,\bsuptermeq{A}\, t$ means that $t$ can be reached from
$s$ by taking $\beta$-reductions at the root and ``subterm''-steps,
where $Z : i$ is in $A$ whenever we pass into argument $i$ of a
meta-variable $Z$.
BRSMTs are used to generate \emph{candidates}:

\begin{definition}[Candidates]\label{def:candidates}
For a meta-term $s$, the set $\cand(s)$ of \emph{candidates of $s$}
consists of those pairs $t\ (A)$
such that
(a) $t$ has the form
$\apps{\identifier{f}}{s_1}{s_\mia}$ with $\identifier{f} \in \Defineds$
and $\mia = \minarity{\Rules}(\identifier{f})$, and
(b) there are $s_{\mia+1},\dots,s_n$ (with $n \geq k$) such that
$s \bsuptermeq{A} \apps{t}{s_{\mia+1}}{s_n}$, and
(c) $A$ is minimal: there is no subset
$A' \subsetneq A$ with $s \bsuptermeq{A'} t$.
\end{definition}

\begin{example}\label{ex:metanonapplied}
In AFSMs where
all meta-variables have
arity $0$ and
the right-hand sides of rules are $\beta$-normal,
the set $\cand(s)$ for a meta-term $s$ consists exactly of the
pairs $t\ (\emptyset)$ where $t$ has the form
$\apps{\afun}{s_1}{s_{\minarity{\Rules}(\afun)}}$ and
$t$ occurs as part of $s$.
In
\refEx{ex:ordrec},
we thus have
$\cand(G\ H\ (\abs{m}{\symb{rec}\ (H\ m)\ K\ F\ G}))
=\{\,\symb{rec}\ (H\ m)\ K\ F\ G\ (\emptyset)\,\}$.
\end{example}

If some of the meta-variables \emph{do} take arguments, then the
meta-variable conditions
matter: candidates of $s$ are pairs
$t\  (A)$ where $A$ contains exactly those pairs $Z : i$ for which we
pass through the $i^{\text{th}}$ argument of $Z$ to reach $t$ in $s$.

\begin{example}\label{ex:metaapplied}
Consider an AFSM with the signature from \refEx{ex:ordrec} but a rule
using meta-variables with larger
arities:
\[
\begin{array}{c}
\symb{rec}\ (\symb{lim}\ (\abs{n}{\meta{H}{n}}))\ K\ 
  (\abs{x}{\abs{n}{\meta{F}{x,n}}})\ (\abs{f}{\abs{g}{\meta{G}{f,g}}})\ 
 \arrz \\
  \meta{G}{\abs{n}{\meta{H}{n}},\ 
  \abs{m}{\symb{rec}\ \meta{H}{m}\ K\ (\abs{x}{\abs{n}{\meta{F}{x,n}}})\ 
  (\abs{f}{\abs{g}{\meta{G}{f,g}}})}}
\end{array}
\]
The right-hand side has one candidate:
\[
\symb{rec}\ \meta{H}{m}\ K\ (\abs{x}{\abs{n}{\meta{F}{x,n}}})\ 
(\abs{f}{\abs{g}{\meta{G}{f,g}}})\ (\{G : 2 \})
\]
\end{example}

The original static approaches define
DPs as
pairs $\ell^\sharp \arrdp p^\sharp$ where $\ell \arrz r$ is a rule
and $p$ a subterm of $r$ of the form $\apps{\afun}{r_1}{r_\maa}$ --
as their rules are built using terms, not meta-terms.
This can set variables bound in $r$ free in $p$.
In the
current setting, we use candidates with their meta-variable conditions
and implicit $\beta$-steps rather than subterms, and we replace such
variables by meta-variables.

\begin{definition}[$\SDP$]\label{def:sdp}
Let $s$ be a meta-term and $(\F,\Rules)$ be an AFSM.
Let $\metafy(s)$ denote $s$ with all free variables replaced by
corresponding meta-variables.
Now $\SDP(\Rules) = \{ \ell^\sharp \arrdp \metafy(p^\sharp)\ 
(A) \mid \ell \arrz r \in \Rules \wedge p\ (A) \in \cand(r) \}$.
\end{definition}

Although static
DPs always have a pleasant form
$\apps{\afun^\sharp}{\ell_1}{\ell_\mia} \arrdp \apps{\bfun^\sharp}{p_1}{p_n}\ 
(A)$ (as opposed to the \emph{dynamic}
DPs of, e.g.,
\cite{kop:raa:12}, whose right-hand sides
can have a meta-variable
at the head, which complicates various techniques in the
framework),
they have two important complications
not present in first-order
DPs: the right-hand side $p$ of
a DP $\ell \arrdp p\ (A)$ may contain meta-variables that
do not occur in the left-hand side $\ell$ --
traditional analysis techniques are not really equipped for
this -- and
the left- and right-hand sides may have different types.
In \refSec{sec:framework} we will explore some methods to deal with
these features.

\begin{example}\label{ex:derivsdp}
For the non-$\eta$-expanded rules of \refEx{ex:deriv}, the set
$\SDP(\Rules)$ has one element: 
$\deriv^\sharp\ (\abs{x}{\symb{sin}\ \meta{F}{x}}) \arrdp
\deriv^\sharp\ (\abs{x}{\meta{F}{x}})$.
(As $\symb{times}$ and $\symb{cos}$ are not defined
symbols, they do not generate dependency pairs.)
The set $\SDP(\Rules^\uparrow)$ for the $\eta$-expanded rules
is
$\{
\deriv^\sharp\ (\abs{x}{\symb{sin}\ \meta{F}{x}})\ Y \arrdp
  \deriv^\sharp\ (\abs{x}{\meta{F}{x}})\ Y
\}$.
To obtain the relevant candidate, we used the $\beta$-reduction
step of BRSMTs.
\end{example}

\begin{example}\label{ex:ordrecstatic}
The AFSM from \refEx{ex:ordrec} is AFP following \refEx{ex:afpgood};
here $\SDP(\Rules)$ is:
\[
\begin{array}{rcll}
\symb{rec}^\sharp\ (\suc\ X)\ K\ F\ G & \arrdp & \symb{rec}^\sharp\ 
  X\ K\ F\ G\ (\emptyset) \\
\symb{rec}^\sharp\ (\symb{lim}\ H)\ K\ F\ G & \arrdp &
  \symb{rec}^\sharp\ (H\ M)\ K\ F\ G\ (\emptyset) \\
\end{array}
\]
Note that the right-hand side of the second DP contains a
meta-variable that
is not
on the left.
As we will see in
\refEx{ex:ordrecdone}, that is not problematic
here.
\end{example}

Termination analysis using dependency pairs importantly considers the
notion of a \emph{dependency chain}.  This notion is fairly similar to
the first-order setting:

\begin{definition}[Dependency chain]\label{def:chain}
Let $\P$ be a set of DPs and $\Rules$ a set of rules.
A (finite or infinite) \emph{$(\P,\Rules)$-dependency chain} (or just
$(\P,\Rules)$-chain) is a sequence
$[(\ell_0 \arrdp p_0\ (A_0),s_0,t_0),
  (\ell_1 \arrdp p_1\ (A_1),s_1,t_1),\ldots]$
where each $\ell_i \arrdp p_i\ (A_i) \in \P$ and all $s_i,t_i$ are
terms, such that for all $i$:
\begin{enumerate}
\item\label{depchain:dp}
  there exists a substitution $\gamma$ on domain $\FMV(\ell_i)
  \cup \FMV(p_i)$ such that $s_i = \ell_i\gamma,\ t_i = p_i\gamma$
  and for all $Z \in \domain(\gamma)$: $\gamma(Z)$ respects $A_i$;
\item\label{depchain:reduce}
  we can write $t_i = \apps{\afun}{u_1}{u_n}$ and $s_{i+1} =
  \apps{\afun}{w_1}{w_n}$ and each $u_j \arrr{\Rules} w_j$.
\end{enumerate}
\end{definition}

\begin{example}\label{ex:mapdp}
In the (first) AFSM from \refEx{ex:mapintro},
we have
$\SDP(\Rules) = \{ 
\map^{\sharp}\ (\abs{x}{\meta{Z}{x}}) \linebreak 
(\cons\ 
H\ T) \arrdp \map^{\sharp}\ (\abs{x}{\meta{Z}{x}})\ T \}$.
An example of a finite dependency chain is
$[(\rho,s_1,t_1),(\rho,s_2,t_2)]$ where $\rho$ is the one DP,
$s_1 = \map^\sharp\ (\abs{x}{\suc\ x})\ (\cons\ \nul\ (\cons\ 
(\suc\ \nul)\ \linebreak
(\map\ (\abs{x}{x})\ \nil)))$ and
$t_1 = \map^\sharp\ (\abs{x}{\suc\ x})\ (\cons\ (\suc\ \nul)\ 
(\map\ (\abs{x}{x})\ \nil))$ and
$s_2 = \map^\sharp\ (\abs{x}{\suc\ x})\ (\cons\ (\suc\ \nul)\ \nil)$ and
$t_2 = \map^\sharp\ (\abs{x}{\suc\ x})\ \nil$.

Note that here $t_1$ reduces to $s_2$ in a single step ($\map\ (\abs{x}{x})\ 
\nil \arr{\Rules} \nil$).
\end{example}

We have the following key result:

\edef\thmBasicChain{\number\value{theorem}}
\edef\thmBasicChainSec{\number\value{section}}
\newcommand{\basicChainTheThm}{
\begin{theorem}\label{thm:basicchain}
Let $(\F,\Rules)$ be a
PA-AFP
AFSM.
If $(\F,\Rules)$ is non-terminating, then there is an infinite
$(\SDP(\Rules),\Rules)$-dependency chain.
\end{theorem}
}\basicChainTheThm

\begin{proof}[sketch]
The proof is an adaptation of the
one in~\cite{kus:iso:sak:bla:09},
altered for the more permissive definition of
\emph{accessible function passing} over \emph{plain function passing}
as well as the meta-variable conditions; it also follows from
\refThm{thm:sschain} below.
\qed
\end{proof}

By this result we can use dependency pairs to prove termination of
a given properly applied and AFP AFSM: if we can prove that there is
no infinite $(\SDP(\Rules),\Rules)$-chain, then termination follows
immediately.
Note, however,  that the reverse result does \emph{not} hold:
it is possible to
have an infinite
$(\SDP(\Rules),\Rules)$-dependency chain
even for a terminating PA-AFP AFSM.

\begin{example}\label{ex:staticbad}
Let
$\F \supseteq \{
\nul,\one : \nat$,
$\symb{f} : \nat \arrtype \nat$,
$\symb{g} : (\nat \arrtype \nat)
\arrtype \nat\}$ and
$\Rules = \{ \symb{f}\ \nul \arrz \symb{g}\ 
(\abs{x}{\symb{f}\ x}),
\symb{g}\ (\abs{x}{
\meta{F}{x}}) \arrz
\meta{F}{\one} \}$.
This AFSM is PA-AFP, with
$\SDP(\Rules) = \{ \symb{f}^\sharp\ \nul \arrdp \symb{g}^\sharp\ 
  (\abs{x}{\symb{f}\ x}),\ 
\symb{f}^\sharp\ \nul \arrdp \symb{f}^\sharp\ X\}$; the second rule
does not cause the addition of any dependency pairs.  Although
$\arr{\Rules}$
is terminating, there is an infinite $(\SDP(\Rules),\Rules)$-chain
$[(\identifier{f}^\sharp\ \nul \arrdp \identifier{f}^\sharp\ X,
\identifier{f}^\sharp\ \nul, \identifier{f}^\sharp\ \nul),
(\identifier{f}^\sharp\ \nul \arrdp \identifier{f}^\sharp\ X,
\identifier{f}^\sharp\ \nul, \identifier{f}^\sharp\ \nul),\ldots
]$.
\end{example}

The problem in \refEx{ex:staticbad} is the \emph{non-conservative}
DP $\identifier{f}^\sharp\ \nul \arrdp \identifier{f}^\sharp\ X$,
with $X$
on the right but not on the left.
Such
DPs arise from \emph{abstractions} in
the right-hand sides of rules.  Unfortunately, abstractions are
introduced by the restricted $\eta$-expansion (\refDef{def:etalong})
that we may need to make an AFSM properly applied.  Even so,
often all DPs are conservative, like Ex.\ 
\ref{ex:mapintro} and \ref{ex:deriv}.  There, we do have the inverse
result:

\edef\thmReverseChain{\number\value{theorem}}
\edef\thmReverseChainSec{\number\value{section}}
\newcommand{\reverseChainTheThm}{
\begin{theorem}\label{thm:chainreverse}
For any AFSM $(\F,\Rules)$:
if there is an infinite $(\SDP(\Rules),\Rules)$-chain
$[(\rho_0, s_0, t_0),(\rho_1, s_1, t_1),\ldots]$ with all $\rho_i$
conservative, then $\arr{\Rules}$ is non-terminating.
\end{theorem}
}\reverseChainTheThm

\begin{proof}[sketch]
If $\FMV(p_i) \subseteq \FMV(\ell_i)$, then we can see that
$s_i \arr{\Rules} \cdot \arrr{\beta} t_i'$ for some term $t_i'$ of
which $t_i$ is a subterm.  Since also each $t_i \arrr{\Rules} s_{
i+1}$, the infinite chain induces an infinite reduction $s_0
\arr{\Rules}^+ t_0' \arrr{\Rules} s_1' \arr{\Rules}^+ t_1''
\arrr{\Rules} \dots$.
\qed
\end{proof}

The core of the dependency pair \emph{framework} is to
systematically simplify a set of pairs $(\P,\Rules)$ to prove either
absence or presence of an infinite $(\P,\Rules)$-chain, thus showing termination
or non-termination as appropriate.  By Theorems \ref{thm:basicchain}
and \ref{thm:chainreverse} we can do so, although with some conditions
on the non-termination result.
We can do better by tracking certain properties of dependency
chains.

\begin{definition}[Minimal and Computable chains]
Let $(\F,\AlterRules)$ be an AFSM and $C_{\AlterRules}$
an RC-set satisfying the properties of \refThm{thm:defC} for
$(\F,\AlterRules)$. Let $\F$ contain,
for every type $\sigma$, at least
countably many symbols $\afun : \sigma$
not used in $\AlterRules$.

A $(\P,\Rules)$-chain $[(\rho_0,s_0,t_0),(\rho_1,s_1,t_1),\ldots]$
is \emph{$\AlterRules$-computable} if:
$\arr{\AlterRules} \mathop{\supseteq} \arr{\Rules}$, and
for all $i \in \N$ there exists a substitution $\gamma_i$ such
  that $\rho_i = \ell_i \arrdp p_i\ (A_i)$ with $s_i = \ell_i
  \gamma_i$ and $t_i = p_i\gamma_i$,
  and $(\abs{x_1 \dots x_n}{v})\gamma_i$ is $C_{\AlterRules}$-computable
  for all $v$ and $B$ such that $p_i \bsuptermeq{B} v$,
  $\gamma_i$ respects $B$, and $\FV(v) = \{x_1,\dots,x_n\}$.

A chain is \emph{minimal} if the strict subterms of all $t_i$ are
terminating under $\arr{\Rules}$.
\end{definition}

In the first-order DP framework, \emph{minimal} chains
give access to several powerful techniques to
prove absence of infinite chains, such as
the \emph{subterm criterion}
\cite{hir:mid:07} and \emph{usable rules}
\cite{gie:thi:sch:fal:06,hir:mid:07}.
\emph{Computable} chains go a step further, by building on the
computability inherent in the proof of \refThm{thm:basicchain} and the
notion of \emph{accessible function passing} AFSMs.  In computable
chains, we can require that (some of) the subterms of all $t_i$ are
\emph{computable} rather than merely \emph{terminating}.
This property will be essential in the \emph{computable subterm
criterion processor} (\refThm{thm:staticsubtermproc}).

Another property of dependency chains is the use of \emph{formative
rules}, which has proven very useful
for dynamic DPs
\cite{kop:raa:12}.  Here
we go
further and
consider \emph{formative reductions}, which were
introduced
for the first-order DP framework in \cite{fuh:kop:14}.
This property will be essential in the \emph{formative rules
processor} (\refThm{thm:formativeproc}).

\begin{definition}[Formative chain, formative reduction]\label{def:formative}
A $(\P,\Rules)$-chain $[(\ell_0 \arrdp p_0\ (A_0),s_0,t_0),
(\ell_1 \arrdp p_1\ (A_1),s_1,t_1),\ldots]$
is \emph{formative} if for all $i$,
the reduction $t_i \arrr{\Rules} s_{i+1}$ is $\ell_{i+1}$-formative.
Here, for a pattern $\ell$, substitution $\gamma$ and
term $s$, a reduction $s \arrr{\Rules} \ell\gamma$ is
\emph{$\ell$-formative} if 
one of the following
holds:
\begin{itemize}
\item $\ell$ is not a fully extended linear pattern; that is:
  some meta-variable occurs more than once in $\ell$ or $\ell$ has a
  sub-meta-term $\abs{x}{C[\meta{Z}{\vec{s}}]}$ with $x \notin \{
  \vec{s}\}$
\item $\ell$ is a meta-variable application $\meta{Z}{x_1,\dots,x_\mia}$
  and $s = \ell\gamma$
\item $s = \apps{a}{s_1}{s_n}$ and $\ell = \apps{a}{\ell_1}{\ell_n}$
  with $a \in \F^\sharp \cup \V$ and each $s_i \arrr{\Rules} \ell_i
  \gamma$ by an $\ell_i$-formative reduction
\item $s = \abs{x}{s'}$ and $\ell = \abs{x}{\ell'}$ and $s' \arrr{\Rules}
  \ell'\gamma$ by an $\ell'$-formative reduction
\item $s = \apps{(\abs{x}{u})\ v}{w_1}{w_n}$ and
  $\apps{u[x:=v]}{w_1}{w_n} \arrr{\Rules} \ell\gamma$ by an
  $\ell$-formative reduction
\item $\ell$ is not a meta-variable application, and there
  are
  $\ell' \arrz r' \in \Rules$, meta-variables $Z_1 \dots Z_n$
  ($n \geq 0$) and $\delta$
  such that $s \arrr{\Rules} (\apps{\ell'}{Z_1}{Z_n})\delta$ by an
  $(\apps{\ell'}{Z_1}{Z_n})$-formative reduction, and
  $(\apps{r'}{Z_1}{Z_n})\delta
  \arrr{\Rules} \ell\gamma$ by an $\ell$-formative reduction.
\end{itemize}
\end{definition}

The idea of a formative reduction is to avoid redundant steps: if $s
\arrr{\Rules} \ell\gamma$ by an $\ell$-formative reduction, then this
reduction takes only the steps needed to obtain an instance of $\ell$.
Suppose
that we have rules
$\symb{plus}\ 
\nul\ Y \arrz Y,\ \symb{plus}\ (\suc\ X)\ Y \arrz \suc\ (\symb{plus}\ 
X\ Y)$.
Let $\ell := \symb{g}\ \nul\ X$ and $t := \symb{plus}\ \nul\ \nul$.
Then the reduction $\symb{g}\ t\ t \arr{\Rules} \symb{g}\ \nul\ t$
is $\ell$-formative: we
must reduce the first argument to
get an
instance of $\ell$.  The reduction $\symb{g}\ t\ t \arr{\Rules}
\symb{g}\ t\ \nul \arr{\Rules} \symb{g}\ \nul\ \nul$ is not
$\ell$-formative, because the reduction in the second argument
does not contribute to
the non-meta-variable positions
of $\ell$.
This matters when we consider $\ell$ as the
left-hand side of a rule, say $\symb{g}\ \nul\ X \arrz \nul$: if we
reduce $\symb{g}\ t\ t \arr{\Rules} \symb{g}\ t\ \nul \arr{\Rules}
\symb{g}\ \nul\ \nul \arr{\Rules} \nul$, then the first step was
redundant: removing this step gives a shorter reduction to
the same result: $\symb{g}\ t\ t \arr{\Rules} \symb{g}\ \nul\ t
\arr{\Rules} \nul$.
In an infinite reduction, redundant steps may also be postponed
indefinitely.

\medskip
We can now strengthen the result of \refThm{thm:basicchain}
with two new properties.

\edef\thmSSChain{\number\value{theorem}}
\edef\thmSSChainSec{\number\value{section}}
\newcommand{\ssChainTheThm}{
\begin{theorem}\label{thm:sschain}
Let $(\F,\Rules)$ be a \emph{properly applied}, \emph{accessible
function passing} AFSM.
If $(\F,\Rules)$ is non-terminating, then there is an infinite
$\Rules$-computable
formative $(\SDP(\Rules),\Rules)$-dependency chain.
\end{theorem}
}
\ssChainTheThm

\begin{proof}[sketch]
We select a \emph{minimal non-computable (MNC)} term $s :=
\apps{\afun}{s_1}{s_\mia}$ (where all $s_i$ are $C_\Rules$-computable)
and an infinite reduction starting in $s$.  Then we stepwise build an
infinite dependency chain, as follows.
Since $s$ is non-computable but each $s_i$ terminates
(as computability implies termination), there exist a rule
$\apps{\afun}{\ell_1}{\ell_\mia} \arrz r$ and substitution $\gamma$
such that each $s_i \arrr{\Rules} \ell_i\gamma$ and $r\gamma$ is
non-computable.
We can then identify a candidate $t\ (A)$ of $r$ such that $\gamma$
respects $A$ and $t\gamma$ is a MNC subterm of $r\gamma$; we continue
the process with $t\gamma$ (or a term at its head).
For the \emph{formative} property, we note that if $s \arrr{\Rules}
\ell\gamma$ and $u$ is terminating, then $u \arrr{\Rules} \ell\delta$
by an $\ell$-formative reduction for substitution $\delta$ such that
each $\delta(Z) \arrr{\Rules} \gamma(Z)$.  This follows by postponing
those reduction steps not needed to obtain an instance of $\ell$.
The resulting infinite chain is $\Rules$-computable because we can show, by
induction on the definition of $\gracc$, that if $\ell \arrz r$ is an
AFP rule and $\ell\gamma$ is a MNC term, then $\gamma(Z)$ is
$C_{\Rules}$-computable for all $Z \in \FMV(r)$.
\qed
\end{proof}

As
it is easily seen that
all $C_{\AlterRules}$-computable terms are
$\arr{\AlterRules}$-terminating and
therefore $\arr{\Rules}$-terminating,
every $\AlterRules$-computable $(\P,\Rules)$-dependency chain is also
minimal.
The notions of $\Rules$-computable and formative chains
still do
not suffice to obtain a true inverse result, however (i.e., to prove
that termination implies the absence of an infinite $\Rules$-computable
chain over $\SDP(\Rules)$): the infinite chain in \refEx{ex:staticbad} is
$\Rules$-computable.

\medskip
To see why the two restrictions that the AFSM must be \emph{properly
applied} and \emph{accessible function passing} are necessary, consider
the following examples.

\begin{example}
Consider
$\F \supseteq \{ \fix :
((\symb{o} \arrtype \symb{o}) \arrtype \symb{o} \arrtype \symb{o})
\arrtype \symb{o} \arrtype \symb{o} \}$ and
$\Rules = \{ \fix\ F\ X \arrz F\ (\fix\ F)\ X \}$.
This AFSM is not properly applied; it is also not terminating, as can
be seen by instantiating $F$ with $\abs{y}{y}$.  However, it does not
have any static DPs, since $\symb{fix}\ F$ is not a candidate.  Even
if we altered the definition of static DPs to admit a dependency
pair $\symb{fix}^\sharp\ F\ X \arrdp \symb{fix}^\sharp\ F$, this pair
could not be used to build an infinite dependency chain.

Note that the problem does not arise if we study the $\eta$-expanded
rules $\Rules^\uparrow = \{ \fix\ F\ X \arrz F\ (\abs{z}{\fix\ F\ z)\ 
X} \}$, as the dependency pair $\fix^\sharp\ F\ X \arrdp \fix^\sharp\ 
F\ Z$ does admit an infinite chain.  Unfortunately, as the one
dependency pair does not satisfy the conditions of
\refThm{thm:chainreverse}, we cannot use this to prove non-termination.
\end{example}

\begin{example}
The AFSM from \refEx{ex:lambdadynamic} is not
accessible function passing, since
$\Acc(\symb{lm}) = \emptyset$.
This is good because the set $\SDP(\Rules)$ is empty, which would lead
us to falsely conclude termination without the restriction.
\end{example}

\emph{Discussion:}
\refThm{thm:sschain} transposes the work of
\cite{kus:iso:sak:bla:09,suz:kus:bla:11}
to AFSMs and extends it by using a more liberal restriction, by
limiting interest to \emph{formative}, \emph{$\Rules$-computable}
chains, and by including meta-variable conditions.  Both of these
new properties of chains will support new termination techniques
within the DP framework.

The relationship with the works for functional programming
\cite{kus:13,kus:18} is less clear: they define a different form of
chains suited well to polymorphic systems, but which requires more
intricate reasoning for non-polymorphic systems, as DPs can be used
for reductions at the head of a term.
It is not clear whether there are non-polymorphic systems that can
be handled with one and not the other.  The notions of formative
and $\Rules$-computable chains are not considered there;
meta-variable conditions are not relevant to their $\lambda$-free
formalism.

\section{The static higher-order DP framework}\label{sec:framework}

In first-order term rewriting, the DP \emph{framework}
\cite{gie:thi:sch:05:2} is an extendable framework to prove
termination and non-termination.  As observed in the introduction,
DP analyses in higher-order rewriting typically go beyond the initial
DP \emph{approach} \cite{art:gie:00},
but fall short of the full \emph{framework}.
Here, we define the latter for static DPs.
Complete versions of all proof sketches in this section are
\onlyarxiv{available in \refApp{app:processors}}%
\onlypaper{in \cite[Appendix~C]{technicalreport}}.

\medskip
We have now reduced the problem of termination to
non-existence
of certain chains.  In the DP framework, we formalise this in the
notion of a \emph{DP problem}:

\begin{definition}[DP problem]\label{def:dpproblem}
A \emph{DP problem} is a tuple $(\P,\Rules,m,f)$ with $\P$ a set of
DPs, $\Rules$ a set of rules, $m \in \{ \minimal,
\arbitrary \} \cup \{ \static_\AlterRules \mid \text{any set of rules}\ 
\AlterRules \}$, and $f \in
\{ \formative, \nonformative \}$.%
\footnote{Our framework is implicitly parametrised by
  the signature $\F^\sharp$ used for term formation.
  As none of the processors we present modify this component
  (as indeed there is no need to by \refThm{thm:invar}),
  we leave it implicit.}

A DP problem $(\P,\Rules,m,f)$ is
\emph{finite} if there exists
no infinite $(\P,\Rules)$-chain that is $\AlterRules$-computable if
$m = \static_\AlterRules$, is minimal if $m = \minimal$,
and is formative if $f = \formative$.
It is \emph{infinite} if $\Rules$ is non-terminating, or if there
exists an infinite $(\P,\Rules)$-chain
where all DPs
used
in the chain are conservative.

To capture the
levels of permissiveness in the $m$ flag,
we use a transitive-reflexive relation $\succeq$ generated by
$\static_{\AlterRules} \succeq \minimal \succeq \arbitrary$.
\end{definition}

Thus, the combination of Theorems \ref{thm:sschain} and
\ref{thm:chainreverse} can be rephrased as: an AFSM $(\F,\Rules)$ is
terminating if $(\SDP(\Rules),\Rules,\static_\Rules,\formative)$ is
finite, and is non-terminating if $(\SDP(\Rules),\Rules,m,f)$ is
infinite for some $m \in \{\static_\AlterRules,\minimal,\arbitrary\}$
and $f \in \{\formative,\nonformative\}$.\footnote{The processors
  in this paper do not \emph{alter}
  the flag $m$, but some \emph{require} minimality or computability.
  We include the $\minimal$ option and the subscript $\AlterRules$
  for the sake of future generalisations, and for
  reuse of processors in
  the \emph{dynamic} approach of~\cite{kop:raa:12}.}

\medskip
The core idea of the DP framework is to iteratively simplify a set of
DP problems via \emph{processors} until nothing remains to be proved:

\begin{definition}[Processor]\label{def:proc}
A \emph{dependency pair processor} (or just \emph{processor}) is a
function that takes a DP problem and returns either \no\ or
a set of DP problems.
A processor $\Proc$ is \emph{sound} if a DP problem $\adpprob$ is finite
whenever $\Proc(\adpprob) \neq \no$ and all elements of $\Proc(\adpprob)$
are finite.
A processor $\Proc$ is \emph{complete} if a DP problem $\adpprob$ is
infinite whenever $\Proc(\adpprob) = \no$ or contains an infinite
element.
\end{definition}

To prove finiteness of a DP problem $\adpprob$ with the DP framework,
we proceed analogously to the first-order DP framework
\cite{gie:thi:sch:fal:06}: we repeatedly apply sound DP processors
starting from $\adpprob$ until none remain.  That is, we execute the
following rough procedure:
  (1) let $A := \{ \adpprob \}$;
  (2) while $A \neq \emptyset$: select a problem $\bdpprob \in A$ and a
  sound processor $\Proc$ with
  $\Proc(\bdpprob) \neq \no$, and let
  $A := (A \setminus \{ \bdpprob \}) \cup \Proc(\bdpprob)$.
If this procedure terminates, then $\adpprob$ is a finite DP problem.

To prove termination of an AFSM $(\F,\Rules)$, we would use as initial
DP problem $(\SDP(\Rules),\Rules,\static_{\Rules},\formative)$,
provided that $\Rules$ is properly applied and accessible function
passing (where $\eta$-expansion following \refDef{def:etalong}
may be applied first).
If the procedure terminates -- so finiteness of $\adpprob$ is proved
by the definition of soundness -- then \refThm{thm:sschain} provides
termination of $\arr{\Rules}$.

Similarly, we can use the DP framework to prove infiniteness:
  (1) let $A := \{ \adpprob \}$;
  (2) while $A \neq \no$: select a problem $\bdpprob \in A$ and a complete
    processor $\Proc$, and let $A := \no$ if $\Proc(\bdpprob) = \no$, or
    $A := (A \setminus \{ \bdpprob \}) \cup \Proc(\bdpprob)$ otherwise.
For non-termination of  $(\F,\Rules)$, the initial DP problem should be
$(\SDP(\Rules),\Rules,m,f)$, where $m,f$ can be any flag
(see \refThm{thm:chainreverse}).
Note that the algorithms coincide while processors are used that are
both
sound \emph{and} complete.
In
a tool,
automation (or the user)
must resolve
the non-determinism
and select
suitable processors.

\smallskip
Below, we will present a number of processors within the framework.
We will typically present processors by writing
``for a DP problem $\adpprob$ satisfying $X$, $Y$, $Z$,
$\Proc(\adpprob) = \dots$''.
In these cases, we let
$\Proc(\adpprob) = \{ \adpprob \}$ for any problem $\adpprob$ not
satisfying the given properties.
Many more processors are possible, but we have chosen to present a
selection which touches on all aspects of the DP framework:
\begin{itemize}
\item processors which map a DP problem to $\no$
  (\refThm{def:nontermproc}), a singleton set (most processors) and
  a non-singleton set (\refThm{def:depgraph});
\item changing
  the set $\Rules$ (\refThm{thm:formativeproc}, \ref{thm:usable})
  and
  various flags (\refThm{thm:usable});
\item using specific values of the
  $f$
  (\refThm{thm:formativeproc}) and
  $m$ flags (\refThm{thm:subtermproc}, \ref{thm:usable},
  \ref{thm:staticsubtermproc});
\item using term orderings (\refThm{thm:redpairproc},
  \ref{thm:basetriple}), a key part of many termination
  proofs.
\end{itemize}

\subsection{The dependency graph}\label{subsec:graph}

We can leverage reachability information to \emph{decompose} DP
problems. In first-order rewriting, a graph structure
is used to track which DPs can possibly follow one another in a
chain~\cite{art:gie:00}.
Here, we define this
\emph{dependency graph} as follows.

\begin{definition}[Dependency graph]\label{def:depgraph}
A DP problem $(\P,\Rules,m,f)$ induces a graph structure $\mathit{DG}$,
called its \emph{dependency graph}, whose nodes are the elements of
$\P$. There is a (directed) edge from $\rho_1$ to $\rho_2$
in $\mathit{DG}$ iff
there exist $s_1,t_1,s_2,t_2$ such that
$[(\rho_1,s_1,t_1),(\rho_2,s_2,t_2)]$ is a $(\P,\Rules)$-chain
with the properties for $m,f$.
\end{definition}

\begin{example}\label{ex:graphconditions}
Consider an AFSM with
$\F \supseteq \{ \symb{f} : (\nat \arrtype \nat) \arrtype \nat
\arrtype \nat \}$ and $\Rules = \{
\symb{f}\ (\abs{x}{\meta{F}{x}})\ (\suc\ Y) \arrz \meta{F}{
  \symb{f}\ (\abs{x}{\nul})\ (\symb{f}\ (\abs{x}{\meta{F}{x}})\ Y)}
\}$.
Let $\P := \SDP(\Rules) =$
\[
\left\{
\begin{array}{rrcll}
(1) & \symb{f}^\sharp\ (\abs{x}{\meta{F}{x}})\ (\suc\ Y) & \arrdp &
  \symb{f}^\sharp\ (\abs{x}{\nul})\ 
  (\symb{f}\ (\abs{x}{\meta{F}{x}})\ Y) & (\{ F : 1 \}) \\
(2) & \symb{f}^\sharp\ (\abs{x}{\meta{F}{x}})\ (\suc\ Y) & \arrdp &
  \symb{f}^\sharp\ (\abs{x}{\meta{F}{x}})\ Y & (\{ F : 1 \}) \\
\end{array}
\right\}
\]
The dependency graph of $(\P,\Rules,\minimal,\formative)$ is:

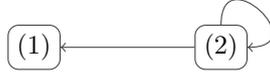
\vspace{-5ex}
\begin{center}
\begin{tikzpicture}[draw=darkgray]
  \tikzstyle{node} = [draw, fill=white, rounded corners]

  \node (pair1) [node] {$(1)$};
  \node (pair2) [node, right of=pair1, node distance=7em] {$(2)$};
  \draw [->] (pair2) -- (pair1);
  \draw [->] (pair2) to[out=90,in=0,looseness=5] (pair2);
\end{tikzpicture}
\end{center}
\vspace{-1.5ex}
There is no edge from (1) to itself or (2) because there is no
substitution $\gamma$ such that $(\abs{x}{\nul})\gamma$ can be reduced
to a term $(\abs{x}{\meta{F}{x}})\delta$ where $\delta(F)$ regards its
first argument (as $\arrr{\Rules}$ cannot introduce new variables).
\end{example}

In general, the dependency graph for a given DP problem is undecidable,
which is why we consider \emph{approximations}.

\begin{definition}[Dependency graph approximation \cite{kop:raa:12}]
A finite graph $G_\theta$ \emph{approximates}
$\mathit{DG}$ if $\theta$ is a function that maps the nodes of
$\mathit{DG}$ to the nodes of $G_\theta$ such that, whenever
$\mathit{DG}$ has
an edge from $\rho_1$ to $\rho_2$,
$G_\theta$
has an edge from
$\theta(\rho_1)$ to $\theta(\rho_2)$.
($G_\theta$ may have edges
that have no corresponding edge in $\mathit{DG}$.)

\end{definition}

Note that this definition allows for an \emph{infinite} graph to be
approximated by a \emph{finite} one; infinite graphs may occur if
$\Rules$ is infinite
(e.g., the union of all simply-typed instances of
polymorphic rules).

If $\P$ is finite, we can
take a graph approximation $G_{\mathtt{id}}$
with the same nodes as $\mathit{DG}$.
A simple approximation may have an edge from $\ell_1 \arrdp p_1\ (A_1)$
to $\ell_2 \arrdp p_2\ (A_2)$ whenever both $p_1$ and $\ell_2$ have the
form $\apps{\afun^\sharp}{s_1}{s_\mia}$ for the same $\afun$ and $\mia$.
However, one can also take the meta-variable conditions into account,
as we did in \refEx{ex:graphconditions}.

\edef\procDPGraph{\number\value{theorem}}
\edef\procDPGraphSec{\number\value{section}}
\newcommand{\DPGraphTheProc}{
\begin{theorem}[Dependency graph processor]\label{def:graphproc}
The processor $\Proc_{G_\theta}$\ that maps a DP problem $\adpprob =
(\P,\Rules,m,f)$
to $\{ (\{ \rho \in \P \mid \theta(\rho) \in C_i \},\Rules,m,f) \mid
1 \leq i \leq n \}$ if $G_\theta$ is an approximation of the
dependency graph of $\adpprob$ and $C_1,\dots,C_n$
are the (nodes of the) non-trivial strongly connected components (SCCs)
of $G_\theta$,
is both sound and complete.
\end{theorem}
}
\DPGraphTheProc

\begin{proof}[sketch]
In an infinite $(\P,\Rules)$-chain $[(\rho_0,s_0,t_0),(\rho_1,s_1,t_1),\ldots]$,
there is always a path from $\rho_i$ to $\rho_{i+1}$ in DG.
Since $G_\theta$ is finite, every infinite path in $DG$ eventually
remains in a cycle in $G_\theta$. This cycle is part of an SCC.
\qed
\end{proof}

\begin{example}\label{ex:finishgraph}
Let $\Rules$ be the set of rules from \refEx{ex:graphconditions}
and $G$ be the graph given there.
Then $\Proc_{G}(\SDP(\Rules),\Rules,
\mathtt{computable}_\Rules,\formative)
= \{ (\{ \symb{f}^\sharp\ (\abs{x}{\meta{F}{x}})\ (\suc\ Y) \arrdp
  \symb{f}^\sharp\ (\abs{x}{\meta{F}{x}})\ Y\ (\{ F : 1 \}) \},
\Rules,
\mathtt{computable}_\Rules,\formative) \}$.
\end{example}

\begin{example}\label{ex:staticgraph}
Let $\Rules$ consist of the rules for $\symb{map}$ from
\refEx{ex:mapintro} along with $\symb{f}\ L \arrz \symb{map}\ (\abs{x}{
\symb{g}\ x})\ L$ and $\symb{g}\ X \arrz X$.
Then $\SDP(\Rules)
= \{
(1)\ \symb{map}^\sharp\ (\abs{x}{\meta{Z}{x}})\ (\symb{cons}\ H\ T)$\linebreak
$\arrdp \symb{map}^\sharp\ (\abs{x}{\meta{Z}{x}})\ T,\
(2)\ \symb{f}^\sharp\ L \arrdp \symb{map}^\sharp\ (\abs{x}{\symb{g}\ x})\ L,\
(3)\ \symb{f}^\sharp\ L \arrdp \symb{g}^\sharp\ X \}$.
DP (3) is not conservative, but it is not on any cycle in the
graph approximation $G_{\mathtt{id}}$ obtained by considering head
symbols as described above:
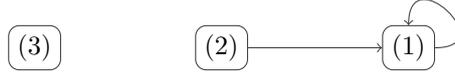
\vspace{-4ex}
\begin{center}
\begin{tikzpicture}[draw=darkgray]
  \tikzstyle{node} = [draw, fill=white, rounded corners]

  \node (pair4) [node] {$(3)$};
  \node (pair3) [node, right of=pair4, node distance=7em] {$(2)$};
  \node (pair2) [node, right of=pair3, node distance=7em] {$(1)$};
  \draw [->] (pair3) -- (pair2);
  \draw [->] (pair2) to[out=0,in=90,looseness=5] (pair2);
\end{tikzpicture}
\end{center}
As $(1)$ is the only DP on a cycle,
$\Proc_{\SDP_{G_{\mathtt{id}}}}(\SDP(\Rules),\Rules,
\mathtt{computable}_\Rules,\linebreak
\formative) =
\{\ (\{(1)\},\Rules,\mathtt{computable}_\Rules,\formative)\ \}$.
\end{example}

\emph{Discussion:} The dependency graph is a powerful tool for
simplifying DP problems,
used since early versions of the
DP approach \cite{art:gie:00}.  Our notion of a dependency graph
approximation, taken from \cite{kop:raa:12}, strictly generalises the
original
notion in \cite{art:gie:00}, which uses a graph
on the same node set as $DG$ with possibly further edges.
One can
get this notion here by using a graph $G_{\symb{id}}$.
The advantage of
our definition is that it ensures soundness
of the dependency graph processor also for \emph{infinite} sets of DPs.
This overcomes a restriction in the literature
\cite[Corollary 5.13]{kus:iso:sak:bla:09} to dependency graphs
without non-cyclic infinite paths.

\subsection{Processors based on reduction triples}\label{subsec:triple}

At the heart of most DP-based approaches to termination proving lie
well-founded orderings
to delete DPs (or rules).
For this, we use
\emph{reduction triples}~\cite{hir:mid:07,kop:raa:12}.

\begin{definition}[Reduction triple]\label{def:redpair}
A \emph{reduction triple} $(\rge,\pge,\pgt)$ consists of
two quasi-orderings $\rge$ and $\pge$ and a well-founded strict
ordering $\pgt$ on meta-terms such that $\rge$ is monotonic,
all of $\rge,\pge,\pgt$ are meta-stable (that is, $\ell \rge r$
implies $\ell\gamma \rge r\gamma$ if $\ell$ is a closed pattern and
$\gamma$ a substitution on domain $\FMV(\ell) \cup \FMV(r)$, and the
same for $\pge$ and $\pgt$),
$\arr{\beta} \mathop{\subseteq} \rge$,
and both $\rge \circ \pgt \mathop{\subseteq} \pgt$
and $\pge \circ \pgt \mathop{\subseteq} \pgt$.
\end{definition}

In the first-order DP framework, the reduction pair processor
\cite{gie:thi:sch:05:2} seeks to orient all rules with $\rge$ and all
DPs with either $\rge$ or $\pgt$;
if this succeeds, those pairs oriented with $\pgt$ may be removed.
Using reduction \emph{triples} rather than pairs, we obtain the following
extension to the higher-order setting: 

\edef\procRedpair{\number\value{theorem}}
\edef\procRedpairSec{\number\value{section}}
\newcommand{\redpairTheProc}{
\begin{theorem}[Basic reduction triple processor]\label{thm:redpairproc}
Let $M = (\P_1 \uplus \P_2,\Rules,m,\linebreak
f)$ be a DP problem.
If $(\rge,\pge,\pgt)$ is a reduction triple such that
\begin{enumerate}
\item
for all $\ell \arrz r \in \Rules$, we have $\ell \rge r$;
\item
for all $\ell \arrdp p\ (A) \in \P_1$, we have $\ell \pgt p$;
\item
for all $\ell \arrdp p\ (A) \in \P_2$, we have $\ell \pge p$;
\end{enumerate}
then the
processor that maps $M$ to $\{(\P_2,\Rules,m,f)\}$
is both sound and complete.
\end{theorem}
}
\redpairTheProc

\begin{proof}[sketch]
For an infinite $(\P_1 \uplus \P_2,
\Rules)$-chain $[(\rho_0,s_0,t_0),(\rho_1,s_1,t_1),\ldots]$
the requirements provide that,
for all $i$:
  (a) $s_i \pgt t_i$ if $\rho_i \in \P_1$;
  (b) $s_i \pge t_i$ if $\rho_i \in \P_2$; and
  (c) $t_i \rge s_{i+1}$.
Since $\pgt$ is well-founded, only finitely many DPs can be in $\P_1$,
so a tail of the chain is actually an infinite $(\P_2,\Rules,m,f)$-chain.
\qed
\end{proof}

\begin{example}\label{ex:derivhorpo}
Let $(\F,\Rules)$ be the (non-$\eta$-expanded) rules from
\refEx{ex:deriv}, and $\SDP(\Rules)$ the
DPs
from \refEx{ex:derivsdp}.
From \refThm{thm:redpairproc}, we
get the following ordering
requirements:
\[
\begin{array}{rcl}
\deriv\ (\abs{x}{\symb{sin}\ \meta{F}{x}}) & \rge &
  \abs{y}{\symb{times}\ (\deriv\ (\abs{x}{\meta{F}{x}})\ y)\ 
  (\symb{cos}\ \meta{F}{y})} \\
\deriv^\sharp\ (\abs{x}{\symb{sin}\ \meta{F}{x}}) & \pgt &
  \deriv^\sharp\ (\abs{x}{\meta{F}{x}}) \\
\end{array}
\]
We can handle both requirements by using a polynomial interpretation
$\mathcal{J}$ to $\mathbb{N}$ \cite{pol:96,fuh:kop:12}, by choosing
$\mathcal{J}_{\symb{sin}}(n) = n + 1$,\ 
$\mathcal{J}_{\symb{cos}}(n) = 0$,\ 
$\mathcal{J}_{\symb{times}}(n_1,n_2) = n_1$,\ 
$\mathcal{J}_{\symb{deriv}}(f) =
\mathcal{J}_{\symb{deriv}^\sharp}(f) = \lambda n.f(n)$.
Then the requirements are evaluated to: $\lambda n.f(n) + 1 \geq
\lambda n.f(n)$ and $\lambda n.f(n) + 1 > \lambda n.f(n)$, which
holds on $\N$.
\end{example}

\refThm{thm:redpairproc} is not ideal since, by definition,
the left- and right-hand side of a
DP may have
different types.  Such
DPs are
hard to handle with traditional
techniques
such as HORPO \cite{jou:rub:99} or
polynomial interpretations \cite{pol:96,fuh:kop:12}, as these methods
compare only (meta-)terms of the same type (modulo renaming of sorts).

\begin{example}\label{ex:differenttypes}
Consider the toy AFSM with
$\Rules = \{ \symb{f}\ (\suc\ X)\ Y \arrz \symb{g}\ X\ Y,\ 
\symb{g}\ X \arrz \abs{z}{\symb{f}\ X\ z} \}$ and
$\SDP(\Rules) = \{ \symb{f}^\sharp\ (\suc\ X)\ Y \arrdp \symb{g}^\sharp\ 
X,\ \symb{g}^\sharp\ X \arrdp \symb{f}^\sharp\ X\ Z \}$.
If $\symb{f}$ and $\symb{g}$ both have a type $\nat \arrtype \nat
\arrtype \nat$, then in the first DP, the left-hand side has type $\nat$
while the right-hand side has type $\nat\arrtype\nat$.  In the second
DP, the left-hand side has type $\nat \arrtype \nat$ and the right-hand
side has type $\nat$.
\end{example}

To be able to handle examples like the one above, we adapt
\cite[Thm.~5.21]{kop:raa:12} by altering the ordering requirements
to have base type.
\pagebreak

\edef\procBasetypeRedpair{\number\value{theorem}}
\edef\procBasetypeRedpairSec{\number\value{section}}
\newcommand{\basetypeRedpairTheProc}{
\begin{theorem}[Reduction triple processor]\label{thm:basetriple}
Let $\mathsf{Bot}$ be a set $ \{ \bot_\atype : \atype \mid$
$\atype \text{ a type}\} \subseteq \F^\sharp$ of unused constructors,
$\adpprob = (\P_1 \uplus \P_2,\Rules,m,f)$ a DP problem and
$(\rge,\pge,\pgt)$ a reduction triple such that:
(a) for all $\ell \arrz r \in \Rules$, we have $\ell \rge r$;
and (b)
  for all $\ell \arrdp p\ (A) \in \P_1 \uplus \P_2$ with $\ell :
  \atype_1 \arrtype \dots \arrtype \atype_\maa \arrtype \asort$
  and $p : \btype_1 \arrtype \dots \arrtype \btype_n \arrtype \bsort$
  we have, for fresh meta-variables
  $Z_1 : \atype_1,\dots,Z_\maa : \atype_\maa$:
  \begin{itemize}
  \item $\apps{\ell}{Z_1}{Z_\maa} \pgt \apps{p}{
    \bot_{\btype_1}}{\bot_{\btype_n}}$ if $\ell \arrdp p\ (A) \in \P_1$
  \item $\apps{\ell}{Z_1}{Z_\maa} \pge \apps{p}{
    \bot_{\btype_1}}{\bot_{\btype_n}}$ if $\ell \arrdp p\ (A) \in \P_2$
  \end{itemize}
Then the processor
that maps $M$ to
$\{(\P_2,\Rules,m,f)\}$ is both sound and complete.
\end{theorem}
}
\basetypeRedpairTheProc

\begin{proof}[sketch]
If $(\rge,\pge,\pgt)$ is such a triple, then for
$R \in \{\pge,\pgt\}$ define $R'$ as
follows: for $s : \atype_1 \arrtype \dots \arrtype \atype_\maa
\arrtype \asort$ and $t : \btype_1 \arrtype \dots \arrtype \btype_n
\arrtype \bsort$, let $s\ R'\ t$ if for all $u_1 : \atype_1,\dots,
u_\maa : \atype_\maa$ there exist $w_1 : \btype_1,\dots,w_n :
\btype_n$ such that $\apps{s}{u_1}{u_\maa}\ R\ \apps{t}{w_1}{w_n}$.
Now apply \refThm{thm:redpairproc} with the triple
$(\rge,\pge',\pgt')$.
\qed
\end{proof}

Here, the elements of $\mathsf{Bot}$ take the role of minimal terms for
the ordering.  We use them to flatten the type of the right-hand
sides of ordering requirements, which makes it easier to use
traditional methods to generate a reduction triple.

While $\pgt$ and $\pge$ may still have to orient meta-terms of
distinct types,
these
are always \emph{base} types, which we
could
collapse
to a single sort.  The only relation
required to be monotonic, $\rge$,
regards
pairs of
meta-terms of the \emph{same} type.
This makes it feasible to apply orderings
like HORPO or polynomial
interpretations.

Both the basic and non-basic reduction triple processor are
difficult to use for \emph{non-conservative} DPs, which generate
ordering requirements whose right-hand side contains a meta-variable
not occurring on the left.
This is typically difficult for traditional techniques,
although possible to overcome, by choosing triples that do not
regard such meta-variables (e.g., via an argument filtering
\cite{kus:nak:toy:99,suz:kus:bla:11}):

\begin{example}\label{ex:disregard}
We apply \refThm{thm:basetriple} on the DP problem $(\SDP(\Rules),
\Rules,\static_\Rules,\linebreak
\formative)$ of \refEx{ex:differenttypes}.  This gives for instance
the following ordering requirements:
\[
\begin{array}{rclcrcl}
\symb{f}\ (\suc\ X)\ Y & \rge & \symb{g}\ X\ Y & \quad &
\symb{f}^\sharp\ (\suc\ X)\ Y & \pgt & \symb{g}^\sharp\ X\ \bot_\nat \\
\symb{g}\ X & \rge &  \abs{z}{\symb{f}\ X\ z} & \quad &
\symb{g}^\sharp\ X\ Y & \pge & \symb{f}^\sharp\ X\ Z \\
\end{array}
\]
The right-hand side of the last DP uses a meta-variable $Z$
that does not occur on the left.
As
neither $\pgt$ nor $\pge$ are
required to be monotonic (only $\rge$ is),
function symbols do not have to regard all their arguments.
Thus, we can
use a polynomial
interpretation $\mathcal{J}$ to $\mathbb{N}$ with
$\mathcal{J}_{\bot_\nat} = 0$,\ 
$\mathcal{J}_{\suc}(n) = n + 1$ and
$\mathcal{J}_{\symb{h}}(n_1,n_2) = n_1$ for $\symb{h} \in \{
\symb{f},\symb{f}^\sharp,\symb{g},\symb{g}^\sharp \}$.
The ordering requirements then translate to $X + 1 \geq X$ and
$\lambda y.X \geq \lambda z.X$ for the rules, and $X + 1 > X$ and
$X \geq X$ for the DPs.  All these
inequalities on $\N$ are clearly satisfied, so
we can remove the first DP.
The remaining problem is quickly dispersed with the dependency graph
processor.
\end{example}

\subsection{Rule removal without search for orderings}\label{subsec:ruleremove}

While processors often simplify only $\P$, they can also simplify
$\Rules$.
One of the most powerful techniques in first-order DP approaches
that can do this are \emph{usable rules}.
The idea is that for a given set $\P$ of DPs, we only need to consider
a \emph{subset} $\UR(\P,\Rules)$ of $\Rules$.  Combined with the
dependency graph processor, this makes it possible to split a large
term rewriting system into a number of small problems.

\pagebreak
In the higher-order setting, simple versions of usable rules have also
been defined \cite{suz:kus:bla:11,kop:raa:12}.  We can easily extend
these definitions to AFSMs:

\begin{theorem}\label{thm:usable}
Given a DP problem $\adpprob = (\P,\Rules,m,f)$ with $m \succeq
\minimal$ and $\Rules$ finite,
let $\UR(\P,\Rules)$ be
the smallest
subset of $\Rules$ such that:
\begin{itemize}
\item if a symbol $\afun$ occurs in the right-hand side of an element
  of $\P$ or $\UR(\P,\Rules)$, and there is a rule $\apps{\afun}{
  \ell_1}{\ell_\mia} \arrz r$, then this rule is also in $\UR(\P,
  \Rules)$;
\item if there exists $\ell \arrz r \in \Rules$ or $\ell \arrdp r\ (A)
  \in \P$ such that $r \suptermeq \apps{\meta{F}{s_1,\dots,
  s_\mac}}{t_1}{t_n}$ with $s_1,\dots,s_\mac$ not all distinct
  variables or with $n > 0$, then $\UR(\P,\Rules) = \Rules$.
\end{itemize}
Then the processor
that maps $\adpprob$ to $\{(\P,
\UR(\P,\Rules),\arbitrary,\nonformative)\}$ is sound.
\end{theorem}

For the proof we refer to the very similar proofs in
\cite{suz:kus:bla:11,kop:raa:12}.

\begin{example}\label{ex:maprecusable}
For the set $\SDP(\Rules)$ of the ordinal recursion example
(Ex.\ \ref{ex:ordrec} and \ref{ex:ordrecstatic}), all rules are
usable due to the occurrence of $H\ M$ in the second DP.
For the set $\SDP(\Rules)$ of the map example (Ex.\ \ref{ex:mapintro}
and \ref{ex:mapdp}), there are no usable rules, since the one DP contains no
defined function symbols or applied meta-variables.
\end{example}

This higher-order processor is much less powerful than its first-order
version: if any DP or usable rule has a sub-meta-term of the
form $F\ s$ or $\meta{F}{s_1,\dots,s_\mac}$ with $s_1,\dots,s_\mac$ not
all distinct variables, then \emph{all} rules are usable.
Since applying a higher-order meta-variable to some argument is
extremely
common in
higher-order rewriting, the technique is usually not applicable.
Also, this processor imposes a heavy price on the flags:
minimality (at least) is required, but is lost; the formative flag
is also lost.
Thus, usable rules are often 
combined with reduction triples to temporarily disregard rules,
rather than as a way to permanently remove rules.

\medskip
To address these weaknesses, we
consider a processor that uses
similar ideas to usable rules, but operates from the \emph{left-hand}
sides of rules and
DPs rather than the right.  This
adapts the technique from \cite{kop:raa:12} that relies on
the new \emph{formative} flag.  As in the first-order case
\cite{fuh:kop:14}, we use a semantic characterisation of formative
rules. In practice, we
then work with over-approximations of
this characterisation, analogous to the use of dependency graph
approximations in \refThm{def:graphproc}.

\begin{definition}
\label{def:formativerules}
A function $\FR$ that maps a pattern $\ell$ and a set of rules $\Rules$
to a set $\FR(\ell,\Rules) \subseteq \Rules$ is a \emph{formative
  rules approximation} if for all $s$ and $\gamma$: if
$s \arrr{\Rules} \ell\gamma$ by an $\ell$-formative reduction,
then this reduction can be done using only rules in $\FR(\ell,\Rules)$.

We let $\FR(\P,\Rules) = \bigcup \{ \FR(\ell_i,
\Rules) \mid \apps{\afun}{\ell_1}{\ell_n} \arrdp p\, (A)\: \in\, \P
\wedge 1 \leq i \leq n \}$.
\end{definition}

Thus, a formative rules approximation is a subset of $\Rules$ that
is \emph{sufficient} for a formative reduction: if
$s \arrr{\Rules}\ell\gamma$,
then
$s \arrr{\FR(\ell,\Rules)} \ell\gamma$.  It is
allowed for there to exist other formative reductions that do use
additional rules.

\begin{example}\label{ex:formative}
We define a simple formative rules approximation:
(1)
$\FR(Z,\Rules) = \emptyset$ if $Z$
is a meta-variable;
(2)
$\FR(\apps{\afun}{\ell_1}{\ell_\maa},\Rules) =
  \FR(\ell_1,\Rules) \cup \dots \cup \FR(\ell_\maa,\Rules)$ if
  $\symb{f} : \atype_1 \arrtype \dots \arrtype \atype_\maa \arrtype
  \asort$ and no rules have type $\asort$;
(3) $\FR(s,\Rules) = \Rules$
otherwise.
This is a formative rules approximation: if $s \arrr{\Rules} Z\gamma$
by a $Z$-formative reduction, then $s = Z\gamma$, and if $s \arrr{
\Rules} \apps{\afun}{\ell_1}{\ell_\maa}$ and no rules have
the same output type as $s$, then $s = \apps{\afun}{s_1}{s_\maa}$
and each $s_i \arrr{\Rules} \ell_i\gamma$ (by an $\ell_i$-formative
reduction).
\end{example}

\pagebreak
The following result follows directly from the definition
of formative rules.

\edef\procFormative{\number\value{theorem}}
\edef\procFormativeSec{\number\value{section}}
\newcommand{\formativeTheProc}{
\begin{theorem}[Formative rules processor]
\label{thm:formativeproc}
For a formative rules approximation $\FR$, the processor $\Proc_{\FR}$
that maps a DP problem $(\P,\Rules,m,\formative)$ to $\{ (\P,
\FR(\P,
\Rules),m,\formative) \}$ is both
sound and complete.
\end{theorem}
}
\formativeTheProc

\begin{proof}[sketch]
A processor that only removes rules (or DPs) is always complete.
For soundness, if the chain is formative then each step
$t_i \arrr{\Rules} s_{i+1}$ can be replaced by $t_i \arrr{\FR(\P,
\Rules)} s_{i+1}$.  Thus, the chain can be seen as a $(\P,\FR(\P,
\Rules))$-chain.
\qed
\end{proof}

\begin{example}\label{ex:plode}
For our ordinal recursion example
(Ex.\ \ref{ex:ordrec} and \ref{ex:ordrecstatic}),
\emph{none} of the rules
are included when we use the approximation of
\refEx{ex:formative} since all rules have output type
$\symb{ord}$.
Thus,
$\Proc_\FR$ maps
$(\SDP(\Rules),\Rules,\static_\Rules,
\formative)$ to
$(\SDP(\Rules),\emptyset,\static_\Rules,\formative)$.
\emph{Note:}
this example can also be completed without formative rules (see
Ex.~\ref{ex:ordrecdone}).
Here we illustrate that, even
with a simple formative rules approximation,
we can often delete all rules of a given type.
\end{example}

Formative rules are introduced in~\cite{kop:raa:12}, and the
definitions can be adapted to a more powerful formative rules
approximation than the one sketched in \refEx{ex:plode}.  Several
examples and deeper intuition for the first-order setting are
given in \cite{fuh:kop:14}.

\subsection{Subterm criterion processors}

Reduction triple processors are powerful, but
they exert a computational
price:
we must orient all rules in $\Rules$.  The subterm criterion processor
allows us to remove DPs without considering
$\Rules$ at all.
It is based on a \emph{projection function} \cite{hir:mid:07},
whose higher-order counterpart
\cite{kus:iso:sak:bla:09,suz:kus:bla:11,kop:raa:12} is the following:

\begin{definition}
For $\P$ a set of DPs, let
$\mathtt{heads}(\P)$ be the set of all symbols $\afun$ that occur as
the head of a left- or right-hand side of a DP in $\P$.
A \emph{projection function} for $\P$ is a function $\nu :
\mathtt{heads}(\P) \to \N$ such that for all DPs
$\ell \arrdp p\ (A) \in \P$, the function $\project$ with
$\project(\apps{\afun}{s_1}{s_n}) = s_{\nu(\afun)}$ is well-defined
both for $\ell$ and for $p$.
\end{definition}

\edef\procSubtermCriterion{\number\value{theorem}}
\edef\procSubtermCriterionSec{\number\value{section}}
\newcommand{\subtermCriterionTheProc}{
\begin{theorem}[Subterm criterion processor]\label{thm:subtermproc}
The processor
$\Proc_{\mathtt{subcrit}}$ that maps a DP problem
$(\P_1 \uplus \P_2,\Rules,m,f)$ with $m \succeq \minimal$
to $\{(\P_2,\Rules,m,f)\}$
if a projection function $\nu$ exists such that
  $\project(\ell) \supterm \project(p)$ for all
  $\ell \arrdp p\ (A) \in \P_1$ and
  $\project(\ell) = \project(p)$ for all
  $\ell \arrdp p\ (A) \in \P_2$,
is sound and complete.
\end{theorem}
}
\subtermCriterionTheProc

\begin{proof}[sketch]
If the
conditions are satisfied, every infinite $(\P,
\Rules)$-chain induces an infinite $\suptermeq \mathop{\cdot}
\arrr{\Rules}$ sequence that starts in a strict subterm of $t_1$,
contradicting minimality unless all but finitely many steps are
equality.  Since every occurrence of a pair in $\P_1$ results in a
strict $\supterm$ step, a tail of the chain lies in $\P_2$.
  \qed
\end{proof}

\begin{example}\label{ex:mapfinish}
Using
$\nu(\symb{map}^\sharp) = 2$,
$\Proc_{\mathtt{subcrit}}$ maps the DP problem $(\{(1)\},\Rules,
\linebreak
\static_\Rules,\formative)$ from \refEx{ex:staticgraph} to
$\{(\emptyset,\Rules,\static_\Rules,\formative)\}$.
\end{example}

The subterm criterion can be strengthened, following
\cite{kus:iso:sak:bla:09,suz:kus:bla:11}, to
also
handle
DPs like the one in \refEx{ex:derivsdp}.
Here, we focus on a new idea.
For \emph{computable} chains, we can build on the idea of
the subterm criterion to get something more.

\edef\procStaticSubtermCriterion{\number\value{theorem}}
\edef\procStaticSubtermCriterionSec{\number\value{section}}
\newcommand{\staticSubtermCriterionTheProc}{
\begin{theorem}[Computable subterm criterion processor]\label{thm:staticsubtermproc}
The processor $\Proc_{\mathtt{statcrit}}$  that maps a DP problem
$(P_1 \uplus \P_2,\Rules,\static_\AlterRules,f)$ to
$\{(\P_2,\Rules,\linebreak
\static_\AlterRules,
f)\}$ if a projection function $\nu$ exists such that
$\project(\ell) \sqsupset \project(p)$ for all $\ell \arrdp
  p\ (A) \in \P_1$ and
$\project(\ell) = \project(p)$ for all $\ell \arrdp p\ (A) \in
  \P_2$,
 is
 sound and complete.
Here, $\sqsupset$ is the relation on base-type terms
with
$s
\sqsupset t$ if $s \neq t$ and (a) $s \gracc t$ or (b)
a meta-variable $Z$ exists
with
 $s \gracc \meta{Z}{x_1,\dots,x_\mia}$
and $t = \apps{\meta{Z}{t_1,\dots,t_\mac}}{s_1}{s_n}$.
\end{theorem}
}
\staticSubtermCriterionTheProc

\begin{proof}[sketch]
By the conditions,
every infinite $(\P,
\Rules)$-chain induces an infinite $(\accreduce{C_\AlterRules} \cup
\arr{\beta})^* \cdot \arrr{\Rules}$ sequence (where $C_\AlterRules$
is
defined following \refThm{thm:defC}).  This contradicts
computability unless there are only finitely
many inequality steps.  As pairs in $\P_1$ give rise to a strict
decrease, they may occur only finitely often.
\qed
\end{proof}

\begin{example}\label{ex:ordrecdone}
Following \refEx{ex:ordrec} and \ref{ex:ordrecstatic},
consider the projection function $\nu$ with $\nu(\symb{rec}^\sharp) =
1$.
As $\suc\ X \gracc X$ and $\symb{lim}\ H \gracc H$,
both $\suc\ X \sqsupset X$ and $\symb{lim}\ H \sqsupset H\ M$
hold.
Thus $\Proc_{\mathtt{statc}}(\P,\Rules,\static_\Rules,\formative) =
\{ (\emptyset,\Rules,\static_\Rules,\linebreak
\formative) \}$.
By the dependency graph processor, the AFSM is terminating.
\end{example}

The computable subterm criterion processor fundamentally relies on the
new $\static_\AlterRules$ flag, so it has no counterpart in the literature
so far.

\subsection{Non-termination}

While (most of) the processors presented
so far are complete,
none of them can
actually
return \no.
We have not yet implemented
such a processor;
however, we can already provide a general specification of
a \emph{non-termination processor}.

\edef\procInfinite{\number\value{theorem}}
\edef\procInfiniteSec{\number\value{section}}
\newcommand{\infiniteTheProc}{
\begin{theorem}[Non-termination processor]\label{def:nontermproc}
Let $M = (\P,\Rules,m,f)$ be a DP problem.
The processor
that
maps $M$ to
\no\ if it
determines that a sufficient criterion for non-termination of
$\arr{\Rules}$
or for existence of an infinite conservative
$(\P,\Rules)$-chain
according to the flags $m$ and $f$ holds is sound and complete.
\end{theorem}
}
\infiniteTheProc

\begin{proof}
Obvious.
\qed
\end{proof}

This is a very general processor, which does not tell us \emph{how} to
determine
such a sufficient criterion.
However, it allows us to conclude
non-termination as part of the framework by identifying a
suitable infinite chain.

\begin{example}\label{ex:nontermproc}
If we can find a finite $(\P,\Rules)$-chain $[(\rho_0,s_0,t_0),\dots,
(\rho_n,s_n,t_n)]$ with $t_n = s_0\gamma$ for some substitution
$\gamma$ which uses only conservative DPs, is formative if
$f = \formative$ and is $\AlterRules$-computable if
$m = \static_\AlterRules$,
such a chain
is clearly a sufficient criterion: there is an
infinite chain
$[(\rho_0,s_0,t_0),\dots,
  (\rho_0,s_0\gamma,t_0\gamma), \dots,\linebreak
  (\rho_0,s_0\gamma\gamma,t_0\gamma\gamma),\ldots]$.
If $m = \minimal$
  and we find such a chain that is however not minimal, then note that
$\arr{\Rules}$ is non-terminating, which also suffices.

For example, for a DP problem $(\P,\Rules,\minimal,\nonformative)$
with $\P = \{ \symb{f}^\sharp\ F\ X \arrdp \symb{g}^\sharp\ (F\ X),\ 
\symb{g}^\sharp\ X \arrdp \symb{f}^\sharp\ \symb{h}\ X \}$, there is
a finite dependency chain:
$
[(\symb{f}^\sharp\ F\ X \arrdp \symb{g}^\sharp\ (F\ X),\ 
  \symb{f}^\sharp\ \symb{h}\ x,\ \symb{g}^\sharp\ (\symb{h}\ x)),\ \;
(\symb{g}^\sharp\ X \arrdp \symb{f}^\sharp\ \symb{h}\ X,\ 
  \symb{g}^\sharp\ (\symb{h}\ x),\ \symb{f}^\sharp\ \symb{h}\ 
  (\symb{h}\ x))]
$.
As $\symb{f}^\sharp\ \symb{h}\ (\symb{h}\ x)$ is an instance of
$\symb{f}^\sharp\ \symb{h}\ x$, the processor maps this DP problem to
$\no$.
\end{example}

To instantiate \refThm{def:nontermproc}, we can borrow non-termination
criteria from first-order rewriting
\cite{gie:thi:sch:05:1,pay:08,emm:eng:gie:12}, with minor adaptions to
the typed setting.
Of course, it is worthwhile to also investigate dedicated
higher-order non-termination criteria.

\section{Conclusions and Future Work}\label{sec:conclusions}

We have built on the static dependency pair approach
\cite{bla:06,kus:iso:sak:bla:09,suz:kus:bla:11,kus:18} and
formulated it in the language of the DP \emph{framework} from
first-order rewriting \cite{gie:thi:sch:05:2,gie:thi:sch:fal:06}.
Our formulation is based on AFSMs, a dedicated formalism designed
to make termination proofs transferrable to
various higher-order rewriting formalisms.

This framework has two important additions over existing higher-order
DP approaches in the literature.
First, we consider not only arbitrary and minimally non-terminating
dependency chains, but also minimally \emph{non-computable} chains;
this is tracked by the $\static_\AlterRules$ flag.
Using the flag,
a dedicated processor allows us to efficiently
handle rules like \refEx{ex:ordrec}.
This flag has no counterpart in the first-order setting.
Second, we have generalised the idea of formative rules in
\cite{kop:raa:12} to a notion of formative \emph{chains}, tracked by
a $\formative$ flag.  This makes it possible to define a corresponding
processor that permanently removes rules.

\smallskip
\emph{Implementation and experiments.}
To provide a strong formal groundwork, 
we have presented several
processors in a general
way, using semantic definitions of, e.g., the dependency
graph approximation and
formative rules rather than
syntactic definitions using functions like
$\mathit{TCap}$ \cite{gie:thi:sch:05:1}.
Even so, most parts of the DP framework for AFSMs have been
implemented in the open-source termination prover
\wanda~\cite{wanda},
alongside a dynamic DP framework~\cite{kop:raa:12} and
a mechanism to delegate some ordering constraints to a first-order
tool~\cite{fuh:kop:11}.  For reduction triples, polynomial
interpretations \cite{fuh:kop:12} and a version of HORPO
\cite[Ch.~5]{kop:12} are used.  To solve the constraints arising
in the search for these orderings, and also to determine sort
orderings (for the accessibility relation) and projection functions
(for the subterm criteria), \wanda employs an external SAT-solver.
\wanda\ has won the higher-order category of the
International Termination Competition~\cite{termcomp}
four times.
In the International Confluence Competition~\cite{coco},
the
tools \acph~\cite{acph17} and \csiho~\cite{csiho18}
use \wanda\ as their ``oracle''
for
termination proofs on HRSs.

We have tested \wanda\ on the \emph{Termination Problems Data
Base} \cite{tpdb}, using
\aprove\ \cite{gie:asc:bro:emm:fro:fuh:hen:ott:plu:sch:str:swi:thi:17}
and 
\minisat \cite{minisat}
as back-ends.
When no additional features are enabled, 
\wanda proves termination of 124 (out of 198) benchmarks with static DPs,
versus 92 with only a search for reduction orderings; a 34\% increase.
When all features except static
DPs are enabled, \wanda\ succeeds on 153 benchmarks, versus 166 with also
static DPs; an 8\% increase, or alternatively, a 29\% decrease in failure
rate.
The full evaluation is available in \onlypaper{\cite[Appendix D]{technicalreport}}{\onlyarxiv{\refApp{sec:experiments}}}.

\smallskip

\emph{Future work.}
While the static and the dynamic DP approaches each have their own
strengths, there has thus far been little progress on a \emph{unified}
approach, which could
take advantage of the syntactic benefits of both styles.
We plan to combine the present work
with the ideas of \cite{kop:raa:12} into such a
unified DP framework.

In addition, we plan to extend the higher-order DP framework
to rewriting with \emph{strategies}, such as
implicit $\beta$-normalisation or strategies inspired by functional
programming languages like OCaml and Haskell.
Other natural directions
are dedicated automation to detect non-termination,
and reducing
the number of term constraints
solved by the reduction triple
processor via a tighter integration with usable and formative rules
with respect to argument filterings.

\bibliography{references}

\newpage
\appendix
\edef\savedcounter{\number0}
\edef\savedcounterSec{\number0}
\newcommand{\startappendixcounters}{
\setcounter{theorem}{\savedcounter}
}
\newcommand{\oldcounter}[2]{
\edef\savedcounter{\number\value{theorem}}
\setcounter{theorem}{#1}
}

\newcommand{\comprel}{\sqsupset}

\begin{center}
\textbf{\huge -- APPENDIX --}
\end{center}

This appendix contains detailed proofs for all results in
the paper.  Proofs in higher-order rewriting are typically intricate
and subject to errors in the small details, so we have strived to be
very precise.  However, aside from \refApp{app:computability} (which
is simply an adaptation of an existing technique to the present
setting), the main idea of all proofs is captured by the proof
sketches in the paper.

In addition, Appendix~\ref{sec:experiments} presents an experimental
evaluation that considers the power of the techniques in this paper on
the termination problem database~\cite{tpdb}.

\section{Computability: the set $C$}\label{app:computability}

In this appendix, we prove \refThm{thm:defC}: the existence of an
RC-set $C$ that provides an accessibility relation $\gracc$ that
preserves computability, and a base-type accessibility step
$\accreduce{C}$ that preserves both computability and termination.

As we have said before, $\V$ and $\F$ contain infinitely many symbols
of all types.  We will use this to select variables or constructor
symbol of any given type without further explanation.

These proofs \emph{do not require} that computability is considered
with respect to a rewrite relation: other relations (such as recursive
path orderings) may be used as well.  To make this explicit, we will
use an alternative relation symbol, $\comprel$.

The proofs here consider a computability notion over the set
$\Terms(\F,\V)$ of terms without restrictions.  However, they could
easily be extended to subsets of a different set of terms $T$,
provided $T$ is closed under $\arr{\Rules}$.  This could for instance
be used to obtain a computability result for terms that satisfy
certain arity restrictions.  To make this generality clear, each
quantification over terms is explicitly marked with $\Terms(\F,\V)$.

Note: a more extensive discussion of computability can be found
in~\cite{bla:16}.  Our notion of accessibility largely corresponds
to
membership of the computability closure defined there (although not
completely).

\subsection{Definitions and main computability result}

\begin{definition}\label{def:comprel}
In \refApp{app:computability}, $\comprel$ is assumed to be a given
relation on terms of the same type, with respect to which we consider
computability.  We require that:
\begin{itemize}
\item $\comprel$ is monotonic (that is, $s \comprel t$ implies that
  $s\ u \comprel t\ u$ and $u\ s \comprel u\ t$ and $\abs{x}{s}
  \comprel \abs{x}{t}$);
\item for all variables $x$: $\apps{x}{s_1}{s_n} \comprel t$ implies
  that $t$ has the form $\apps{x}{s_1}{s_i'} \cdots s_n$ with
  $s_i \comprel s_i'$;
\item if $s \arrr{\mathtt{head}\beta} u$ and $s \comprel t$, then
  there exists $v$ such that $u \comprel^* v$ and $t \arrr{\mathtt{
  head}\beta} v$;
  here, $\arr{\mathtt{head}\beta}$ is the relation generated by the
  step $\apps{(\abs{x}{u})\ v}{w_1}{w_n}
  \arr{\mathtt{head}\beta} \apps{u[x:=v]}{w_1}{w_n}$;
\item if $t$ is the $\mathtt{head}\beta$-normal form of $s$, then
  $s \comprel^* t$.
\end{itemize}
We call a term \emph{neutral} if it has the form $\apps{x}{s_1}{s_n}$
or $\apps{(\abs{x}{u})}{s_0}{s_n}$.
\end{definition}

The generality obtained by imposing only the minimal requirements
on $\comprel$ is not needed in the current paper (where we only consider
computability with respect to a rewrite relation), but could be used to
extend the method to other domains.
First note:

\begin{lemma}\label{lem:rewriterelationsuffices}
A rewrite relation $\arr{\Rules}$ satisfies the requirements of
$\comprel$ stated in \refDef{def:comprel}.
\end{lemma}

\begin{proof}
Clearly $\arr{\Rules}$ is monotonic, applications with a variable at
the head cannot be reduced at the head, and
moreover
$\arr{\Rules}$ includes
$\arr{\mathtt{head}\beta}$.

The third property we prove by induction
on $s$ with $\arr{\beta}$, using $\arrr{\Rules}$ instead of
$\arr{\Rules}$ for a stronger induction hypothesis.  If $s = u$, then we
are done choosing $v := t$.  Otherwise we can write
$s = \apps{(\abs{x}{q})\ w_0}{w_1}{w_n}$
and $s \arr{\mathtt{head}\beta} s' := \apps{q[x:=w_0]}{w_1}{w_n}$,
and $s' \arrr{\mathtt{head}\beta} u$.
If the reduction $s \arrr{\Rules} t$ does not take any head steps, then
\[
t = \apps{(\abs{x}{q'})\ w_0'}{w_1'}{w_n'} \arrr{\mathtt{head}\beta}
\apps{q'[x:=w_0']}{w_1'}{w_n'} =: v
\]
and indeed $u \arrr{\Rules} v$ by
monotonicity.  Otherwise, by the same argument we can safely assume
that the head step is done first, so $s' \arrr{\Rules} t$; we complete
by the induction hypothesis.
\qed
\end{proof}

\medskip
Recall \refDef{def:RCset} from the text.

\oldcounter{\defRCset}{\defRCsetSec}
\begin{definition}[with $\comprel$ rather than $\arr{\Rules}$]
A \emph{set of reducibility candidates}, or \emph{RC-set}, for a
relation $\comprel$ as in \refDef{def:comprel} is a set $I$ of
base-type terms $s \in \Terms(\F,\V)$ such that:
\begin{itemize}
\item every term in $I$ is terminating under $\comprel$
\item $I$ is closed under $\comprel$ (so if $s \in I$ and $s \comprel t$
  then $t \in I$)
\item if $s \in \Terms(\F,\V)$ is neutral, and for all $t$ with $s
  \comprel t$ we have $t \in I$, then $s \in I$
\end{itemize}
We define $I$-computability for an RC-set $I$ by induction on types;
for $s \in \Terms(\F,\V)$ we say $s$ is $I$-computable if:
\begin{itemize}
\item $s : \asort$ for some $\asort \in \Sorts$ and
  $s \in I$ ($\asort \in \Sorts$)
\item $s : \atype \arrtype \btype$ and for all terms $t : \atype \in
  \Terms(\F,\V)$ that are $I$-computable, $\app{s}{t}$ is $I$-computable
\end{itemize}
\end{definition}
\startappendixcounters

For $\asort$ a sort and $I$ an RC-set, we will write $I(\asort) = \{
s \in I \mid s : \asort \}$.

Let us illustrate \refDef{def:RCset} with two examples:

\begin{lemma}\label{lem:MINSN}
The set \textsf{SN} of all terminating base-type terms in $\Terms(\F,
\V)$ is an RC-set.
The set \textsf{MIN} of all terminating base-type terms in $\Terms(\F,
\V)$ whose $\mathtt{head}\beta$-normal form can be written
$\apps{x}{s_1}{s_\maa}$ with $x \in \V$ is also an RC-set.
\end{lemma}

\begin{proof}
It is easy to verify that the requirements hold for \textsf{SN}.  For
\textsf{MIN}, clearly termination holds.  If $s \in \textsf{MIN}$,
then $s \arrr{\mathtt{head}\beta} \apps{x}{s_1}{s_\maa} =: s'$, so for
any $t$ with $s \comprel^* t$ the assumptions on $\comprel$ provide
that $t \arrr{\mathtt{head}\beta} v$ for some $\comprel^*$-reduct of
$s'$, which can only have the form $\apps{x}{t_1}{t_\maa}$.
Finally, we prove that a neutral term $s \in \Terms(\F,\V)$ is in
\textsf{MIN} if all
its $\comprel^+$-reducts are, by induction on $s$ with $\arr{\beta}$
(this suffices because we have already seen that \textsf{MIN} is
closed under $\comprel$).
If $s = \apps{x}{s_1}{s_\maa}$ then it is included in \textsf{MIN} if
it is terminating, which is the case if all its reducts are
terminating, which is certainly the case if they are in \textsf{MIN}.
If $s = \apps{(\abs{x}{u})\ v}{w_1}{w_\maa}$ then it is included if
(a) all its reducts are terminating (which is satisfied if they are in
\textsf{MIN}), and (b) the $\mathtt{head}\beta$-normal form $s'$ of
$s$ has the right form, which holds because $s \comprel^+ s'$
(as $\arr{\mathtt{head}\beta}$ is included in $\comprel$) and
therefore $s' \in \textsf{MIN}$ by assumption.
\qed
\end{proof}

In fact, we have that \textsf{MIN} $\subseteq I \subseteq$ \textsf{SN}
for all RC-sets $I$.  The latter inclusion is obvious by the
termination requirement in the definition of RC-sets.  The former
inclusion follows easily:

\begin{lemma}\label{lem:mincom}
For all RC-sets $I$, \textsf{MIN} $\subseteq I$.
\end{lemma}

\begin{proof}
We prove by induction on $\comprel$ that all elements of \textsf{MIN}
are also in $I$.  It is easy to see that if $s \in \textsf{MIN}$ then
$s$ is neutral.  Therefore, $s \in I$ if $t \in I$ whenever $s
\comprel t$.  But since \textsf{MIN} is closed by
\refLemma{lem:MINSN}, each such $t$ is in \textsf{MIN}, so also in
$I$ by the induction hypothesis.
\qed
\end{proof}

Aside from minimality of \textsf{MIN}, \refLemma{lem:mincom}
actually provides $I$-computability of all variables, regardless of $I$.
We prove this alongside termination of all $I$-computable terms.

\begin{lemma}\label{lem:compresults}
Let $I$ be an RC-set.  The following statements hold for all types
$\atype$:
\begin{enumerate}
\item\label{lem:compresults:vars}
  all variables $x : \atype$ are $I$-computable
\item\label{lem:compresults:term}
  all $I$-computable terms $s : \atype$ are terminating (w.r.t.\ $\comprel$)
\end{enumerate}
\end{lemma}

\begin{proof}
By a mutual induction on the form of $\atype$, which we may safely
write $\atype_1 \arrtype \dots \arrtype \atype_\maa \arrtype \asort$
(with $\maa \geq 0$ and $\asort \in \Sorts$).

(\ref{lem:compresults:vars}) By definition of $I$-computability,
$x : \atype$ is computable if and only if $\apps{x}{s_1}{s_\maa} \in I$
for all $I$-computable terms $s_1 : \atype_1,\dots,s_\maa : \atype_\maa$
in $\Terms(\F,\V)$.
However, as all $\atype_i$ are smaller types, we know that such terms
$s_i$ are terminating, so \refLemma{lem:mincom} gives the required
result.

(\ref{lem:compresults:term}) Let $x_1 : \atype_1,\dots,x_\maa :
\atype_m$ be variables; by the induction hypothesis they are
computable, and therefore $\apps{s}{x_1}{x_\maa}$ is in $I$ and
therefore terminating. Then the head, $s$, cannot itself be
non-terminating (by monotonicity of $\comprel$).
\qed
\end{proof}

While \textsf{SN} is indisputably the easiest RC-set to define and
work with, it will be beneficial for the strength of the method to
consider a set strictly between \textsf{MIN} and \textsf{SN}.  To this
end, we assume given an ordering on types, and a function mapping each
function symbol $\afun$ to a set $\Acc(\afun)$ of arguments positions.
Here, we deviate from the text by not fixing $\Acc$; again, this
generality is not needed for the current paper, but is done with an
eye on future extensions.

\begin{definition}[Generalisation of Def.~\ref{def:accArgs}]
Assume given a quasi-ordering $\greqsort$ on $\Sorts$ whose strict part
$\grsort\ :=\ \greqsort \setminus \leqsort$ is well-founded.  Let $\eqsort$
denote the corresponding equivalence relation $\eqsort\ :=\ \greqsort \cap
\leqsort$.

For a type $\atype\ \equiv\ \atype_1 \arrtype \dots \arrtype \atype_\maa
\arrtype \bsort$ (with $\bsort \in \Sorts$) and sort $\asort$, we
write $\asort \gracsortup \atype$ if $\asort \greqsort \bsort$ and
$\asort \gracsortdown \atype_i$ for each $i$, and we
write $\asort
\gracsortdown \atype$ if $\asort \grsort \bsort$ and $\asort
\gracsortup \atype_i$ for each $i$.

For $\afun : \atype_1 \arrtype \dots \arrtype \atype_\maa \arrtype
\asort$ \textcolor{red}{\emph{we assume given a set $\Acc(\afun) \subseteq \{ i \mid 1
\leq i \leq \maa \wedge \asort \gracsortup \atype_i \}$.}}
For $x : \atype_1 \arrtype \dots \arrtype \atype_\maa \arrtype \asort
\in \V$, we write $\Acc(x) = \{ i \mid 1 \leq i \leq \maa \wedge
\atype_i = \btype_1 \arrtype \dots \arrtype \btype_n \arrtype \bsort
\wedge \asort \greqsort \bsort \}$.
We
write $s \gracc t$ if either $s = t$, or $s = \abs{x}{s'}$ and $s'
\gracc t$, or $s = \apps{a}{s_1}{s_n}$ with $a \in \F \cup \V$ and $s_i
\gracc t$ for some $i \in \Acc(a)$
with $a \notin \FV(s_i)$.
\end{definition}

\emph{Remark:} This definition of the accessibility relations deviates
from, e.g., \cite{bla:jou:rub:15} by using a pair of relations ($\gracsortup$ and
$\gracsortdown$) rather than positive and negative positions.  This is not
an important difference, but simply a matter of personal preference;
using a pair of relations avoids the need to discuss type positions in
the text, allowing for a shorter presentation.  It is also not common
to allow a choice in
$\Acc(\afun)$, but rather to fix $\Acc(\afun) = \{ \atype_i \mid 1 \leq i
\leq \maa \wedge \asort \gracsortup \atype \}$ for \emph{some} symbols
(for instance constructors) and $\Acc(\afun) = \emptyset$ for the
rest.  We elected to leave the choice open for greater generality.

\medskip
The interplay of the positive and negative relations $\gracsortup$ and
$\gracsortdown$ leads to an important result on RC-sets.

\begin{lemma}\label{lem:gracsortinterplay}
Fix a sort $\asort \in \Sorts$.
Suppose $I,J$ are RC-sets such that $I(\bsort) = J(\bsort)$ for all
$\bsort$ with $\asort \grsort \bsort$ and $I(\bsort) \subseteq
J(\bsort)$ if $\asort \eqsort \bsort$.  Let $s : \atype$.  Then we have:
\begin{itemize}
\item If $\asort \gracsortup \atype$, then if $s$ is $I$-computable
  also $s$ is $J$-computable.
\item If $\asort \gracsortdown \atype$, then if $s$ is $J$-computable
  also $s$ is $I$-computable.
\end{itemize}
\end{lemma}

\begin{proof}
We prove both statements together by a shared induction on the form of
$\atype$.  We can always
write $\atype\ \equiv\ \atype_1 \arrtype
\dots \arrtype \atype_\maa \arrtype \bsort$ with $\bsort \in \Sorts$.

First suppose $\asort \gracsortup \atype$; then $\asort \greqsort
\bsort$ -- so $I(\bsort) \subseteq J(\bsort)$ -- and each $\asort
\gracsortdown \atype_i$.  Assume that $s$ is $I$-computable; we must show
that it is $J$-computable, so that for all $J$-computable $t_1 :
\atype_1,\dots,t_\maa : \atype_\maa$ we have:
$\apps{s}{t_1}{t_\maa} \in J$.  However, by the induction hypothesis each
$t_i$ is also $I$-computable, so $\apps{s}{t_1}{t_\maa} \in I(\bsort)
\subseteq J(\bsort)$ by the assumption.

For the second statement, suppose $\asort \gracsortdown \atype$; then
$\asort \grsort \bsort$, so $I(\bsort) = J(\bsort)$.  Assume that $s$
is $J$-computable; $I$-computability follows if $\apps{s}{t_1}{t_\maa}
\in I(\bsort) = J(\bsort)$ whenever $t_1,\dots,t_\maa$ are $I$-computable.
By the induction hypothesis they are $J$-computable, so this holds
by assumption.
\qed
\end{proof}

The RC-set $C$ whose existence is asserted below offers computability
with a notion of accessibility.  It is worth noting that this is
\emph{not} a standard definition, but is designed to provide an
additional relationship $\accreduce{I}$ that is terminating on computable
terms.  This re\-lation will be useful in termination proofs using static
DPs.

\begin{theorem}\label{thm:defC2}
Let $\accreduce{I}$ be the relation on base-type terms where $
\apps{\afun}{s_1}{s_\maa} \accreduce{I} \apps{s_i}{t_1}{t_n}$ whenever $
i \in \Acc(\afun)$ and $s_i : \atype_1 \arrtype \dots \arrtype \atype_n
\arrtype \asort $ and  each $t_j$ is $I$-computable.

There exists an RC-set $C$ such that $C = \{ s \in \Terms(\F,\V) \mid$
$s$ has base type $\wedge\ s$ is terminating under
$\comprel \cup \accreduce{C}$ and if $s \comprel^*
\apps{\afun}{s_1}{s_\maa}$ then $s_i$ is $C$-computable for
all $i \in \Acc(\afun) \}$.
\end{theorem}

\begin{proof}
We will define, by well-founded induction on $\asort$ using $\greqsort$,
a set $A_\asort$ of terms as follows.

Assume $A_\bsort$ has already been defined for all $\bsort$ with
$\asort \grsort \bsort$, and let $X_\asort$ be the set of RC-sets $I$
such that $I(\bsort) = A_\bsort$ whenever $\asort \grsort \bsort$.
We observe that $X_\asort$ is a complete lattice with respect to
$\subseteq$: defining the bottom element $\sqcup \emptyset := \bigcup \{
A_\bsort \mid \asort \grsort \bsort \} \cup \textsf{MIN}$ and the
top element $\sqcap \emptyset := \bigcup \{ A_\bsort \mid \asort
\grsort \bsort \} \cup \bigcup \{$ \textsf{SN}$(\bsort) \mid \neg
(\asort \grsort \bsort) \}$, and letting
$\sqcup Z := \bigcup Z,\ \sqcap Z := \bigcap Z$ for non-empty $Z$, it
is easily checked that $\sqcap$ and $\sqcup$ give a greatest lower and
least upper bound within $X_\asort$ respectively.
Now for an RC-set $I \in X_\asort$, we let:
\[
\begin{array}{rcl}
F_\asort(I) & = & \{ s \in I \mid s : \bsort \not\eqsort \asort \} \\
& \cup & \{s \in \Terms(\F,\V) \mid s : \bsort \eqsort \asort
   \wedge s\ \text{is terminating} \\
& & \ \ \text{under}\ \comprel \cup \accreduce{I}
   \wedge\ \text{if}\ s \comprel^* \apps{\afun}{s_1}{s_\maa}\ 
   \text{for a symbol} \\
& & \ \ \afun \in \F\ \text{then}\ 
   \forall i \in \Acc(\afun):\, s_i\ \text{is}\ I\text{-computable}]
   \}
\end{array}
\]

Clearly, $F_\asort$ maps elements of $X_\asort$ to $X_\asort$: terms of
type $\bsort \not\eqsort \asort$ are left alone, and $F_\asort(I)$
satisfies the properties to be an RC-set.  Moreover, $F_\asort$ is
monotone.  To see
this, let $I,J \in X_\asort$ such that $I \subseteq J$; we must see that
$F_\asort(I) \subseteq F_\asort(J)$.  To this end, let $s \in
F_\asort(I)$; we will see that also $s \in F_\asort(J)$.  This is
immediate if $s : \bsort \not\eqsort \asort$, as membership in $X_\asort$
guarantees that $F_\asort(I)(\bsort) = I(\bsort) \subseteq J(\bsort) =
F_\asort(J)(\bsort)$.  So assume $s : \bsort \eqsort \asort$.  We must see
two things:
\begin{itemize}
\item $s$ is terminating under $\comprel \cup \accreduce{J}$.  We
  show that $\comprel \cup \accreduce{J}\ \subseteq\ 
  \comprel \cup \accreduce{I}$; as $s$ is terminating in the
  latter, the requirement follows.  Clearly $\comprel\ \subseteq\ 
  \comprel \cup \accreduce{I}$, so assume $s \accreduce{J} s'$.
  Then $s = \apps{\identifier{f}}{t_1}{t_\maa}$ and $s' =
  \apps{t_i}{u_1}{u_n}$ for $i \in \Acc(\identifier{f})$ and
  $J$-computable $u_1,\dots,u_n$.
  We can write $t_i : \atype_1 \arrtype \dots \arrtype \atype_n \arrtype
  \bsort$ and since $i \in \Acc(\identifier{f})$ we have $\asort
  \greqsort \bsort$ and $\asort \gracsortdown \atype_j$ for each $j$.
  By \refLemma{lem:gracsortinterplay} then each $u_j$ is also
  $I$-computable, so also $s \accreduce{I} s_1$.
\item If $s \comprel^* \apps{\identifier{f}}{s_1}{s_\maa}$ for some
  symbol $\identifier{f}$ then for all $i \in \Acc(\identifier{f})$:
  $s_i$ is $J$-computable.  But this is obvious: as $s \in F_\asort(I)$,
  we know that such $s_i$ are $I$-computable, and since $\asort
  \gracsortup \atype_i$ for $i \in \Acc(\identifier{f})$,
  \refLemma{lem:gracsortinterplay} provides $J$-computability.
\end{itemize}
Thus, $F$ is a monotone function on a complete lattice; by
Tarski's fixpoint theorem there is a fixpoint, so an RC-set $I$ such
that for all sorts $\bsort$:
\begin{itemize}
\item if $\asort \grsort \bsort$ then $I(\bsort) = A_\bsort$;
\item if $\asort \eqsort \bsort$ then $I(\bsort) = \{ s \in
  \Terms(\F,\V) \mid s : \bsort \wedge s$ is terminating under
  $\comprel \cup \accreduce{I} \wedge$ if $s \comprel^*
  \apps{\identifier{f}}{s_1}{s_\maa}$ for a symbol $\identifier{f}$ then
  $\forall i \in \Acc(\identifier{f}):\, s_i$ is $I$-computable$ \}$
\end{itemize}
We define $A_\bsort := I(\bsort)$ for all $\bsort \eqsort \asort$.

Now we let $C := \bigcup_{\asort \in \Sorts} A_\asort$.  Clearly, $C$
satisfies the given requirement.
\qed
\end{proof}

\refThm{thm:defC2} easily gives the proof of \refThm{thm:defC}
in the text:

\begin{proof}[Proof of \refThm{thm:defC}]
\refThm{thm:defC} follows by taking $\arr{\Rules}$ for $\comprel$
(which satisfies the requirements by \refLemma{lem:rewriterelationsuffices}) 
and taking for each
$\Acc(\afun)$ the maximum set $\{ i \mid 1 \leq i \leq \maa \wedge
 \asort \gracsortup \atype_i \}$.
\qed
\end{proof}

\subsection{Additional properties of computable terms}\label{subsec:compprop}

For reasoning about computable terms (as we will do when defining
static DPs and reasoning about computable chains), there are a number
of properties besides those in \refLemma{lem:compresults} that will
prove very useful to have.  In the following, we fix the RC-set $C$ as
obtained from \refThm{thm:defC2}.

\begin{lemma}\label{lem:preservecomp}
If $s$ is $C$-computable and $s \comprel t$, then $t$ is
also
$C$-computable.
\end{lemma}

(This actually holds for any RC-set, but we will only use
it for $C$.)

\begin{proof}
By induction on the type of $s$.  If $s$ has base type, then
$C$-computability implies that $s \in C$, and because $C$
is an RC-set all reducts of $s$ are also in $C$.
Otherwise, $s : \atype \arrtype \btype$ and computability of $s$ implies
computability of $s\ u$ for all computable $u : \atype$.  By the
induction hypothesis, the fact that $s\ u \comprel t\ u$ by
monotonicity of $\comprel$ implies that $t\ u$ is computable for
all computable $u$, and therefore by definition $t$ is computable.
\qed
\end{proof}

Thus, computability is preserved under $\comprel$; the following result
shows that it is also preserved under $\accreduce{C}$.

\begin{lemma}\label{lem:preservecompaccreduce}
If $s$ is $C$-computable and $s \accreduce{C} t$, then $t$ is also
$C$-computable.
\end{lemma}

\begin{proof}
If $s \accreduce{C} t$, then both terms have base type, so
$C$-compu\-tability is simply membership in $C$.  We have
$s = \apps{\afun}{s_1}{s_\maa}$ and $t = \apps{s_i}{t_1}{t_n}$
with each $t_j$ $C$-computable.  Since, by definition of $C$, also
$s_i$ is $C$-computable, $C$-computability of $t$ immediately follows.
\qed
\end{proof}

Finally, we will see that $C$-computability is also preserved under
$\gracc$.  For this, we first make a more general statement,
which will also handle variables below binders (which are freed in
subterms).

\begin{lemma}\label{lem:preservecompacchelper}
Let $s : \atype_1 \arrtype \dots \arrtype \atype_\maa \arrtype \asort$
and $t : \btype_1 \arrtype \dots \arrtype \btype_n \arrtype \bsort$ be
meta-terms, such that $s \gracc t$.
Let $\gamma$ be a substitution with $\FMV(s) \subseteq \domain(\gamma)
\subseteq \M$.

Let $u_1 : \btype_1,\dots,u_n : \btype_n$ be $C$-computable terms, and
$\delta$ a substitution with $\domain(\delta) \subseteq \V$ such that
each $\delta(x)$ is $C$-computable, and for $t' := \apps{(t(\gamma \cup
\delta))}{u_1}{u_n}$ there is no overlap between $\FV(t')$ and the
variables bound in $s$.

Then there exist a $C$-computable substitution $\xi$ with
$\domain(\xi) \subseteq \V$ and $C$-computable terms $v_1 : \atype_1,
\dots,v_\maa : \atype_\maa$ such that we have $\apps{(s(\gamma \cup
\xi))}{v_1}{v_\maa}\ (\accreduce{C} \cup \arr{\mathtt{head}\beta})^*\ 
t'$.
\end{lemma}

\begin{proof}
We prove the lemma by induction on the derivation of $s \gracc t$; in
this, we can assume (by $\alpha$-conversion) that variables that occur
bound in $s$ do not also occur free or occur in $\gamma$.

If $s = t$, then we are done choosing $\xi$ and $\vec{v}$
equal to $\delta$ and $\vec{u}$.

If $s = \abs{x}{s'}$ with $x : \atype_1$ and $s' \gracc t$, then by the
assumption based on $\alpha$-conversion above, $x$ does not occur free
in $s$ or in the range of $\gamma$.
By the induction hypothesis, there exist a computable substitution $\xi'$
and computable terms $v_2,\dots,v_\maa$ such that $\apps{(s'(\gamma
\cup \xi'))}{v_2}{v_\maa}\ (\accreduce{C} \cup \arr{\mathtt{head}
\beta})^*\ t'$.  We
can safely assume that $x$ does not occur in the range of $\xi'$, since
$x$ does not occur in $t'$ either
(if $x$ does occur, we can replace it by a different
variable).
Therefore, if we define $\xi := [x:=x] \cup [y:=\xi'(y) \mid y \in \V
\wedge y \neq x]$, we have $s'(\gamma \cup \xi') = (s'(\gamma \cup \xi))
[x:=\xi'(x)]$.  Choosing $v_1 := \xi'(x)$,
we get
$\apps{(s(\gamma \cup \xi))}{v_1}{v_\maa} \arr{\mathtt{head}\beta}
\apps{(s'(\gamma \cup \xi'))}{v_2}{v_\maa}\ (\accreduce{C} \cup
\arr{\mathtt{head}\beta})^*\ t'$.

If $s = \apps{x}{s_1}{s_j}$ for $s_i : \ctype_1 \arrtype \dots
\arrtype \ctype_{n'} \arrtype \bsort'$ with $\asort \greqsort \bsort'$
and $x \notin \FV(s_i)$ and $s_i \gracc t$, then the induction
hypothesis provides $C$-computable terms $w_1 : \ctype_1,\dots,w_{n'}
: \ctype_{n'}$ and a substitution $\xi'$ such that $\apps{(s_i(\gamma
\cup \xi'))}{w_1}{w_{n'}}\ (\accreduce{C} \cup \arr{\mathtt{head}\beta}
)^*\ t'$.  Since $x \notin \FV(s_i)$ by definition of 
$\gracc$, we can safely assume that $x
\notin \domain(\xi')$.
Now recall that by assumption $\F$ contains infinitely many constructors
of all types; let $\symb{c} : \bsort \arrtype \asort$ be a symbol that
does not occur anywhere in $\Rules$.  We can safely assume that
$\Acc(\symb{c}) = \{ 1 \}$.  Then $w := \abs{y_1 \dots y_j z_1 \dots
z_\maa}{\symb{c}\ (\apps{y_i}{w_1}{w_{n'}})}$ is $C$-computable.
Now let $\xi := [x:=w] \cup [y:=\xi'(y) \mid y \in \V \wedge y \neq x]$,
and let $v_1,\dots,v_\maa$ be variables (which are $C$-computable by
Lemma~\ref{lem:compresults}(\ref{lem:compresults:vars})).
Then $\apps{(s(\gamma \cup \xi))}{v_1}{v_\maa} \arr{\mathtt{head}\beta}^{
j+\maa} \apps{s_i(\gamma \cup \xi')}{w_1}{w_n'}\ (\accreduce{C} \cup
\arr{\mathtt{head}\beta})^*\ t'$.

Otherwise, $s = \apps{\identifier{f}}{s_1}{s_n}$ and $s_i \gracc t$
for some $i \in \Acc(\identifier{f})$; by the induction hypothesis
there exist $\xi$ and $C$-computable terms $w_1,\dots,w_{n'}$ such that
$s' := \apps{(s_i(\gamma \cup \xi))}{w_1}{w_{n'}}\ (\accreduce{C} \cup
\arr{\mathtt{head}\beta})^* t'$.
We have $\apps{(s(\gamma \cup \xi))}{v_1}{v_\maa} \accreduce{C} s'$
for any $\vec{v}$ (e.g., variables).
\qed
\end{proof}

From this we conclude:

\begin{lemma}\label{lem:preservecompacc}
Let $s$ be a closed meta-term, $\gamma$ a substitution with
$\FMV(s) \subseteq \domain(\gamma) \subseteq \M$ and $t$ such that
$s \gracc t$ and $s\gamma$ is $C$-computable.  Then for all
substitutions $\delta$ mapping $\FV(t)$ to computable terms:
$t(\gamma \cup \delta)$ is $C$-computable.
\end{lemma}

\begin{proof}
$t(\gamma \cup \delta)$ is $C$-computable if $\apps{(t(\gamma \cup
\delta))}{u_1}{u_n}$ is $C$-computable for all computable $u_1,\dots,
u_n$.  By \refLemma{lem:preservecompacchelper} and the fact that $s$
is closed, there exist $C$-computable terms $v_1,\dots,v_\maa$ such
that $\apps{(s\gamma)}{v_1}{v_\maa}\ (\accreduce{C} \cup \arr{\mathtt{
head}\beta})^*\ \apps{(t(\gamma \cup \delta))}{u_1}{u_n}$.  But
$s\gamma$ is $C$-computable, and therefore so is
$\apps{(s\gamma)}{v_1}{v_\maa}$.  Since $\accreduce{C}$ and
$\arr{\mathtt{head}\beta}$ are both computability-preserving by
Lemmas \ref{lem:preservecompaccreduce} and \ref{lem:preservecomp}
respectively (as $\arr{\mathtt{head}\beta}$ is included in $\comprel$)
we are done.
\qed
\end{proof}

\begin{lemma}\label{lem:neutralcomp}
A neutral term is $C$-computable if and only if all its
$\comprel$-reducts are $C$-computable.
\end{lemma}

\begin{proof}
That $C$-computability of a term implies $C$-computability of its
reducts is given by \refLemma{lem:preservecomp}.  For the other
direction, let $s : \atype_1 \arrtype \dots \arrtype \atype_\maa \arrtype
\asort$ be neutral and suppose that all its reducts are
$C$-computable.  To prove that also $s$ is $C$-computable, we must see
that for all $C$-computable terms $t_1 : \atype_1,\dots,t_\maa :
\atype_\maa$ the term $u := \apps{s}{t_1}{t_\maa}$ is in $C$.  We prove
this by induction on $(t_1,\dots,t_\maa)$ ordered by
$\comprel_{\mathtt{prod}}$.
Clearly, since $s$ does not have the form $\apps{\afun}{s_1}{s_n}$ with
$\Acc(\afun) \neq \emptyset$, nor does $u$, so $u \in C$ if all its
reducts are in $C$.  But since $s$ is neutral, all reducts of
$u$ either have the form $\apps{s'}{t_1}{t_\maa}$ with $s \comprel s'$
-- which is in $C$ because all $t_i$ are $C$-computable and $s'$ is
computable as a reduct of $s$ -- or the form $\apps{s}{t_1}{t_i'} \cdots
t_\maa$ with $t_i \comprel t_i'$ -- which is in $C$ by the induction
hypothesis.
\qed
\end{proof}

Using the $\arr{\mathtt{head}\beta}$-restrictions on $\comprel$, we
obtain the following result:

\begin{lemma}\label{lem:abscomputable}
Let $x : \atype \in \V$.
A term $\abs{x}{s} \in \Terms(\F,\V)$ is $C$-computable if and only if
$s[x:=t]$ is computable for all $C$-computable $t : \atype$.
\end{lemma}

\begin{proof}
If $\abs{x}{s}$ is $C$-computable, then by definition so is
$(\abs{x}{s})\ t$ for all $C$-computable $t$; by
\refLemma{lem:preservecomp} and inclusion of $\arr{\mathtt{head}\beta}$
in $\comprel$, this implies $C$-computability of the reducts $s[x:=t]$.

For the other direction, suppose $s[x:=t]$ is $C$-computable for all
$C$-computable $t : \atype$.
To obtain $C$-computability of $\abs{x}{s}$, we must see that
$(\abs{x}{s})\ t$ is $C$-computable for all $C$-computable $t : \atype$.
As $(\abs{x}{s})\ t$ is neutral, this holds if all its
$\comprel$-reducts $u$ are $C$-computable by
\refLemma{lem:neutralcomp}, and certainly if all its
$\comprel^+$-reducts are $C$-computable, which we prove by induction
on $u$ oriented with $\comprel$.  But by definition of $\comprel$ (and
induction on the derivation $(\abs{x}{s})\ t \comprel^+ u$) there
exists a term $v$ such that $s[x:=t] \comprel^* v$ and $u
\arrr{\mathtt{head}\beta} v$.  If $u = v$, we thus obtain the
required property, and if $u \arr{\mathtt{head}\beta}^+ v$, then $u$
is neutral and hence is $C$-computable if all its
$\comprel$-reducts are, which is the case by the induction hypothesis.
\qed
\end{proof}

\section{Static dependency pairs}\label{app:staticdp}

In this appendix, we will first prove the main result from
\refSec{sec:dp}: \refThm{thm:sschain}.  Then, we will prove the
``inverse'' result, \refThm{thm:chainreverse}.  Finally, to provide
a greater context to the current work, we will discuss how the
definitions in \cite{kus:iso:sak:bla:09,suz:kus:bla:11} relate to
the definitions here.

\subsection{Static dependency pairs: the main result}\label{app:sdp}

To start, we prove \refThm{thm:sschain}, which states that a properly
applied, accessible function passing AFSM with rules $\Rules$ is
terminating if it admits no $\Rules$-computable formative
$(\SDP(\Rules),\Rules) $-dependency chains.
\refThm{thm:basicchain}, which states that an AFSM is terminating if
it admits no $(\SDP(\Rules),\Rules)$-dependency chains, follows as a
corollary.

In the following, let $C = C_{\Rules}$ be a computability predicate
following \refDef{thm:defC2} for $\comprel$ the rewrite relation
$\arr{\Rules}$.
We will briefly call a term ``computable'' if it is $C$-computable.

Henceforth, we will assume without explicitly stating it that
$(\F,\Rules)$ is properly applied, so we can speak of
$\minarity{\Rules}(\afun)$ without disclaimers; we let
$\minarity{\Rules}(\afun) = 0$ for $\afun \notin \Defineds$.
We start with an observation on the consequences of accessible
function passingness:

\begin{lemma}\label{lem:pfp}
Let $\ell$ be a closed pattern, $Z$ a meta-variable and $x_1,\dots,
x_\mia$ variables such that $\ell \gracc \meta{Z}{x_1,\dots,x_\mia}$.
If $\ell\gamma$ is a computable term, then so is $\gamma(Z)$.
\end{lemma}

\begin{proof}
Since $\ell$ is closed, $\ell(\gamma \cup \delta) = \ell\gamma$ is
computable for all computable substitutions $\delta$ whose domain is
contained in $\V$.  By \refLemma{lem:preservecompacc}, we
thus have computability of $\meta{Z}{x_1,\dots,x_\mia}(\gamma
\cup \delta)$ for all such $\delta$.
Since $\ell\gamma$ is a term, $Z \in \domain(\gamma)$ so
we can either write $\gamma(Z) = \abs{x_1 \dots x_\mia}{s}$
or $\gamma(Z) = \abs{x_1 \dots x_i}{s'}$ with $i < \mia$ and $s'$ not
an abstraction.

In the first case, if we let $\delta := [x_1:=u_1,\dots,x_\mia:=
u_\mia]$ for computable terms $u_1,\dots,u_\mia$ we have computability
of $\meta{Z}{x_1,\dots,x_\mia}(\gamma \cup \delta) =
s[x_1:=u_1,\dots,x_k:=u_\mia]$.  Since this holds for \emph{all}
computable $u_1,\dots,u_\mia$, \refLemma{lem:abscomputable} provides
computability of $\abs{x_1 \dots x_\mia}{s} = \gamma(Z)$.
In the second case, the same substitution $\delta$ provides
computability of $\apps{s'[x_1:=u_1,\dots,x_i:=u_i]}{u_{i+1}}{u_n}$
which (since this holds for \emph{any} $u_{i+1},\dots,u_n$) implies
computability of $s'[x_1:=u_1,\dots,x_i:=u_i]$, and this in turn
implies computability of $\gamma(Z)$ as before.
\qed
\end{proof}

Thus, computability of the left-hand side of an instantiated DP
implies computability of all instantiated meta-variables.
To transfer this property to the right-hand side of the instantiated
pair, we have a closer look at the relation $\bsuptermeq{A}$.

In the following, we say that a meta-term $s$ \emph{respects}
$\minarity{\Rules}$ if $s \bsuptermeq{\beta} \apps{\afun}{t_1}{t_n}$
implies $n \geq \minarity{\Rules}(\afun)$.

\begin{lemma}\label{lem:scandcomplete}
Let $s$ be a meta-term that respects $\minarity{\Rules}$ and $\gamma$
a substitution on a finite domain with $\FMV(s) \subseteq
\domain(\gamma) \subseteq \M$, such that all $\gamma(Z)$ are computable.
If there exists a computable substitution $\delta$ on a variable
domain (that is, $\domain(\gamma) \subseteq \V$) such that $s(\gamma
\cup \delta)$ is not computable, then there exists a pair $t\ (A) \in
\cand(s)$ such that all of the following hold:
\begin{itemize}
\item there is a computable substitution $\delta$ on variable
  domain such that $t(\gamma \cup \delta)$ is not computable;
\item $\gamma$ respects $A$;
\item for all $t' \neq t$ such that $t \bsuptermeq{B} t'$ holds for
  some $B$ respected by $\gamma$: $t'(\gamma \cup \delta)$ is
  computable for all computable substitutions $\delta$ on variable
  domain.
\end{itemize}
\end{lemma}

\vspace{-4pt}
\begin{proof}
Let $S$ be the set of all pairs $t\ (A)$ such that (a) $s
\bsuptermeq{A} t$, (b) there exists a computable substitution $\delta$
on variable domain such that $t(\gamma \cup \delta)$ is
not computable, and (c) $\gamma$ respects $A$.
This set is non-empty, as it contains $\{ s\ (\emptyset) \}$.
As the relations $\bsuptermeq{\beta}$ and $\supseteq$ are both
well-founded quasi-orderings (the latter on finite sets),
we can select a pair $t\ (A)$ that is,
in a sense,
\emph{minimal} in $S$: for all $t'\ (A') \in S$: if $t
\bsuptermeq{\beta} t'$ then $t = t'$ and not $A' \subsetneq A$
(it is possible that $A$ and $A'$ are incomparable).
We observe that for all $t',B$ such that $t' \neq t$ and $t
\bsuptermeq{B} t'$ and $\gamma$ respects $B$ we cannot have $t'\ (A
\cup B) \in S$ by minimality of $t\ (A)$, so since $s \bsuptermeq{
A \cup B} t \bsuptermeq{A \cup B} t'$ and clearly $\gamma$ respects
$A \cup B$, it can only follow that requirement (b) is not satisfied
for $t'$.

Now suppose that $t$ has the form $\apps{\afun}{t_1}{t_n}$.  Then by
the above reasoning, all $t_i(\gamma \cup \delta)$ are computable, and
by definition of ``$s$ respects $\minarity{\Rules}$'' we know that $n
\geq \mia := \minarity{\Rules}(\afun)$.  Thus, $(\apps{\afun}{t_1}{
t_\mia})(\gamma \cup \delta)$ is not computable (since otherwise
$t(\gamma \cup \delta)$ would be computable), and by definition of
$S$ as a set of BRSMTs of $s$ (and minimality of $A$) we have
$\apps{\afun}{t_1}{t_\mia}\ (A) \in \cand(s)$.  By minimality of $t$
in $S$, we see that $\apps{\afun}{t_1}{t_\mia}\ (A)$ satisfies all
the requirements for the lemma to hold.

Thus, if $t$ has the form $\apps{\afun}{t_1}{t_n}$, we are done;
towards a contradiction we will show that if $t$ does \emph{not}
have this form, then $t\ (A)$ is not minimal.

Consider the form of $t$:
\begin{itemize}
\item $t = \abs{x}{t'}$: by \refLemma{lem:abscomputable},
  non-computability of $t(\gamma \cup \delta)$ implies
  non-computabi\-lity of $t'(\gamma \cup \delta)[x:=u]$ for some
  computable $u$.  Since, by $\alpha$-conversion, we can assume that
  $x$ does not occur in domain or range of $\gamma$ or $\delta$,
  we have non-computability of $t'(\gamma \cup \delta \cup [x:=u])$,
  and $\delta \cup [x:=u]$ is a computable substitution on
  variable domain while $t \bsuptermeq{A} t'$.
\item $t = \apps{x}{t_1}{t_n}$ with $x \in \V$: whether $x \in
  \domain(\delta)$ or not, $\delta(x)$ is computable (either by the
  assumption on $\delta$ or by
  \refLemma{lem:compresults}(\ref{lem:compresults:vars})).  Therefore,
  the only way for $t(\gamma \cup \delta)$ to not be computable is if
  some $s_i(\gamma \cup \delta)$ is not computable, and $s
  \bsuptermeq{A} s_i$.
\item $t = \apps{\identifier{c}}{t_1}{t_n}$ with $\identifier{c}
  \in \F \setminus \Defineds$: $t(\gamma \cup \delta)$ is
  non-computable only if there exist computable terms $u_{n+1},\dots,
  u_\maa$ such that the term
  $\apps{\apps{\identifier{c}}{t_1}{t_n}}{u_{n+1}}{u_\maa}$
  of base type
  is not in
  $C$.  This can only be the case if it is non-terminating or some
  $t_i(\gamma \cup \delta)$ is not computable.  Since head-reductions
  are impossible, non-termination implies non-termination of some
  $t_i(\gamma \cup \delta)$ or $u_j$, which by
  \refLemma{lem:compresults}(\ref{lem:compresults:term}) implies
  non-computability; as all $u_j$ are computable by assumption, this
  means some $t_i(\gamma \cup \delta)$ is non-computable.  We are
  done because $t \bsuptermeq{A} t_i$.
\item $t = \apps{\afun}{t_1}{t_n}$ with $\afun \in \Defineds$
  but $n < \arity(\afun)$: same as above, because terms of this form
  cannot be reduced at the head (or the root).
\item $t = \apps{(\abs{x}{u})}{t_0}{t_n}$: $t(\gamma \cup \delta)$ is
  neutral, so by \refLemma{lem:neutralcomp} non-computability implies
  the non-computability of a reduct.  If the reduct
  $\apps{u(\gamma \cup \delta)[x:=t_0(\gamma \cup \delta)]}{
  (t_1(\gamma 
  \cup \delta))}{(t_n(\gamma \cup \delta))} =
  (\apps{u[x:=t_0]}{t_1}{t_n})(\gamma \cup \delta)$ is non-computable,
  we are done because $t \bsuptermeq{A} \apps{u[x:=t_0]}{t_1}{t_n}$.
  Otherwise, note that all many-step reducts of $t(\gamma \cup \delta)$
  are either also a reduct of $(\apps{u[x:=t_0]}{t_1}{t_n})(\gamma \cup
  \delta)$ -- and therefore computable -- or have the form
  $\apps{(\abs{x}{u'})}{t_0'}{t_n'}$ with $u(\gamma \cup \delta)
  \arrr{\Rules} u'$ and each $t_i(\gamma \cup \delta) \arrr{\Rules}
  t_i'$.  Thus, at least one of $u(\gamma \cup \delta)$ or $t_i(\gamma
  \cup \delta)$ has to be non-terminating.  But if $u(\gamma \cup
  \delta)$ is non-terminating, then so is $u[x:=u'](\gamma \cup \delta)$,
  contradicting computability of $(\apps{u[x:=t_0]}{t_1}{t_n})(\gamma
  \cup \delta)$.  The same holds if $t_i(\gamma \cup \delta)$ is
  non-terminating for some $i \geq 1$.  Thus, $t_0(\gamma \cup \delta)$
  is non-terminating and therefore non-computable, and we indeed have
  $t \bsuptermeq{A} t_0$.
\item $t = \apps{\meta{Z}{s_1,\dots,s_\mac}}{t_1}{t_n}$: we either
  have $\gamma(Z) = \abs{x_1 \dots x_\mac}{u}$ or $\gamma(Z) =
  \abs{x_1 \dots x_i}{u'}$ with $i < \mac$ and $u'$ not an abstraction;
  in the latter case let $u := \abs{x_{i+1} \dots x_\mac}{\apps{u'}{
  x_{i+1}}{x_\mac}}$.  Either way, $t(\gamma \cup \delta) =
  \apps{u[x_1:=s_1(\gamma \cup \delta),\dots,x_\mac:=s_\mac(\gamma \cup
  \delta)]}{(t_1(\gamma \cup \delta))}{(t_n(\gamma \cup \delta))}$.

  For this term to be non-computable, either some $t_i(\gamma \cup
  \delta)$ should be non-computable, or
  $u[x_1:=s_1(\gamma \cup \delta),\dots,x_\mac:=s_\mac(\gamma \cup
  \delta)]$.  The former case immediately contradicts minimality,
  since $t \bsuptermeq{\emptyset} t_i$, so we assume the latter.
  However, if all $s_i(\gamma \cup \delta)$ are computable, then so
  is $u[x_1:=s_1(\gamma \cup \delta),\dots,x_\mac:=s_\mac(\gamma \cup
  \delta)]$:
  \begin{itemize}
  \item if $\gamma(Z) = \abs{x_1 \dots x_\mac}{u}$ then this holds by
    computability of all $\gamma(Z)$ and \refLemma{lem:abscomputable};
  \item if $\gamma(Z) = \abs{x_1 \dots x_i}{u'}$ with $i < \mac$ and
    $u = \apps{u'}{x_{i+1}}{x_\mac}$, then computability of
    $\gamma(Z)$ and \refLemma{lem:abscomputable} provide computability
    of $u'[x_1:=q_1,\dots,x_i:=q_i]$, which by definition of
    computability for higher-order terms implies computability
    for
    $\apps{u'[x_1:=q_1,\dots,x_i:=q_i]}{q_{i+1}}{q_n} = u[x_1:=q_1,
    \dots,x_n:=q_n]$.
  \end{itemize}

  Thus, some $s_i(\gamma \cup \delta)$ must be non-computable, and
  since substituting an unused variable has no effect, this must be
  the case for some $i$ with $x_i \in \FV(u)$.  So in this case,
  $\gamma$ respects $B := A \cup \{ Z : i \}$ and we obtain $s
  \bsuptermeq{B} t \bsuptermeq{B} s_i$.
\qed
\end{itemize}
\end{proof}

Next, let us consider formative reductions.  We will prove that
reductions from a terminating term to some instance of a pattern may
be assumed to be formative.

\begin{lemma}\label{lem:formative}
Let $\ell$ be a pattern and $\gamma$ a substitution on domain
$\FMV(\ell)$ such that a meta-variable $Z$ with $\arity(Z) = \mia$
is mapped to a term $\abs{x_1 \dots x_\mia}{t}$.
Let $s$ be a terminating term.
If $s \arrr{\Rules} \ell\gamma$, then there exists a substitution
$\delta$ on the same domain as $\gamma$ such that each $\delta(Z)
\arrr{\Rules} \gamma(Z)$ and $s \arrr{\Rules} \ell\delta$ by an
$\ell$-formative reduction.
\end{lemma}

Note that the restriction on $\gamma$ is very light: every substitution
$\gamma$ on domain $\FMV(\ell)$ can be altered to map meta-variables
with arity $\mia$ to terms with $\mia$ abstracted variables:
if $\gamma(Z) = \abs{x_1 \dots x_\mia}{t}$ with $\mia = \arity(Z)$
then let $\gamma'(Z) = \gamma(Z)$, and if $\gamma(Z) = \abs{x_1 \dots
x_i}{t}$ with $i < \arity(Z)$ and $t$ not an abstraction, then replace
this by setting $\gamma'(Z) := \abs{x_1 \dots x_\mia}{\apps{t}{x_{i+
1}}{x_\mia}}$.  Note that $\meta{Z}{x_1,\dots,x_\mia}\gamma =
\meta{Z}{x_1,\dots,x_\mia}\gamma'$.  Therefore, we always have
$\ell\gamma = \ell\gamma'$.

\begin{proof}
We prove the lemma by induction first on $s$ ordered by $\arr{\Rules}
\mathop{\cup} \supterm$, second on the length of the reduction $s
\arrr{\Rules} \ell\gamma$.  If $\ell$ is not a
fully extended linear pattern, then we are done choosing $\delta :=
\gamma$. Otherwise, we consider four cases:
\begin{enumerate}
\item $\ell$ is a meta-variable application $\meta{Z}{x_1,\dots,x_\mia}$;
\item $\ell$ is not a meta-variable application, and the reduction
  $s \arrr{\Rules} \ell\gamma$ does not contain any headmost steps;
\item $\ell$ is not a meta-variable application, and the reduction
  $s \arrr{\Rules} \ell\gamma$ contains headmost steps, the first of
  which is a $\arr{\beta}$ step;
\item $\ell$ is not a meta-variable application, and the reduction
  $s \arrr{\Rules} \ell\gamma$ contains headmost steps, the first of
  which is not a $\arr{\beta}$ step.
\end{enumerate}

In the first case, if $\ell$ is a meta-variable application
$\meta{Z}{x_1,\dots,x_\mia}$, then by $\alpha$-conversion we may
write $\gamma = [Z:=\abs{x_1\dots x_\mia}{t}]$ with $\ell\gamma = t$.
Let $\delta$ be $[Z := \abs{x_1 \dots x_\mia}{s}]$.  Then $\delta$
has the same domain as $\gamma$, and indeed $\delta(Z) = \abs{x_1
\dots x_\mia}{s} \arrr{\Rules} \abs{x_1 \dots x_\mia}{(\ell\gamma)} =
\gamma(Z)$.

In the second case, a reduction without any headmost steps, note
that $s$ has the same outer shape as $\ell$: either (a) $s = \abs{x}{s'}$
and $\ell = \abs{x}{\ell'}$, or (b) $s = \apps{a}{s_1}{s_n}$ and $\ell =
\apps{a}{\ell_1}{\ell_n}$ for some $a \in \V \cup \F$ (since $\ell$ is
a pattern, $a$ cannot be a meta-variable application or abstraction if
$n > 0$).
In case (a), we obtain $\delta$ such that $s' \arrr{\Rules}
\ell'\delta$ by an $\ell'$-formative reduction and $\delta
\arrr{\Rules} \gamma$ by the induction hypothesis (as sub-meta-terms of
linear patterns are still linear patterns).  In case (b), we let
$\gamma_i$ be the restriction of $\gamma$ to $\FMV(\ell_i)$ for $1
\leq i \leq n$; by linearity of $\ell$, all $\gamma_i$ have
non-overlapping domains and $\gamma = \gamma_1 \cup \dots \cup
\gamma_n$.  The induction hypothesis provides $\delta_1,\dots,
\delta_n$ on the same domains such that each $s_i \arrr{\Rules}
\ell_i\delta_i$ by an $\ell_i$-formative reduction and $\delta_i
\arrr{\Rules} \gamma_i$; we are done choosing $\delta := \delta_1 \cup
\dots \cup \delta_n$.

In the third case, if the first headmost step is a
$\beta$-step,
note that $s$ must have the form $\apps{(\abs{x}{t})\ u}{q_1}{q_n}$,
and moreover $s \arrr{\Rules} \apps{(\abs{x}{t'})\ u'}{q_1'}{q_n'} \arr{\beta}
\apps{t'[x:=u']}{q_1'}{q_n'} \arrr{\Rules} \ell\gamma$ by steps in the
respective subterms.  But then also $s \arr{\beta} \apps{t[x:=u]}{
q_1}{q_n} \arrr{\Rules} \apps{t'[x:=u']}{q_1'}{q_n'} \arrr{\Rules}
\ell\gamma$, and we can
get
$\delta$ and an $\ell$-formative
reduction for $\apps{t[x:=u]}{q_1}{q_n} \arrr{\Rules} \ell\delta$ by
the induction hypothesis.

In the last case, if the first headmost step is not a
$\beta$-step, then we can write $s = \apps{\afun}{s_1}{s_n}
\arrr{\Rules} \apps{\afun}{s_1'}{s_n'} = \apps{(\ell'\eta)}{s_{i+1}'
}{s_n'} \arr{\Rules} \apps{(r\eta)}{s_{i+1}'}{s_n'} \arrr{\Rules}
\ell\gamma$ for some $\afun \in \Defineds$, terms
$s_j \arrr{\Rules} s_j'$ for $1 \leq j \leq n$,
rule $\ell' \arrz r$
and substitution $\eta$ on domain $\FMV(\ell')$.  But then
$\apps{\ell'}{Z_{i+1}}{Z_n} \arrz \apps{r}{Z_{i+1}}{Z_n} \in
\RulesEta$, and for $\eta' := \eta \cup [Z_{i+1}:=s_{i+1}',
\dots,Z_n:=s_n']$ we both have $s \arrr{\Rules} (\apps{\ell'}{Z_{i+1}
}{Z_n})\eta'$ without any headmost steps, and $(\apps{r}{Z_{i+1}
}{Z_n})\eta' \arrr{\Rules} \ell\gamma$.
By the second induction hypothesis, there exists a substitution
$\xi$ such that $s \arrr{\Rules} (\apps{\ell'}{Z_{i+1}}{Z_n})\xi$ by
a $(\apps{\ell'}{Z_{i+1}}{Z_n})$-formative reduction and $\xi
\arrr{\Rules} \eta'$.  This gives $s \arr{\Rules}^+
(\apps{r}{Z_{i+1}}{Z_n})\xi \arrr{\Rules} (\apps{r}{Z_{i+1}}{Z_n})
\eta' \arrr{\Rules} \ell\gamma$, so by the first induction
hypothesis we obtain $\delta$ such that $(\apps{r}{Z_{i+1}}{Z_n})\xi
\arrr{\Rules} \ell\delta$ by an $\ell$-formative reduction, and
$\delta \arrr{\Rules} \gamma$.
\qed
\end{proof}

Essentially, \refLemma{lem:formative} states that we can postpone
reductions that are not needed to obtain an instance of the given
pattern.  This is not overly surprising, but will help us eliminate
some proof obligations later in the termination proof.

From this, we have the main result on static dependency chains.

\oldcounter{\thmSSChain}{\thmSSChainSec}
\ssChainTheThm
\startappendixcounters

\begin{proof}
In the following, let a \emph{minimal non-computable term} be a term
$s := \apps{\afun}{s_1}{s_\mia}$ with $\mia = \minarity{\Rules}(\afun)$,
such that $\afun \in \Defineds$, and $s$ is not computable but all $s_i$
are $C$-computable.  We say that $s$ is MNC.

We first observe that if $\arr{\Rules}$ is non-terminating, then there
exists a MNC term.  After all, if $\arr{\Rules}$ is non-terminating,
then there is a non-terminating term $s$, which (by
\refLemma{lem:compresults}(\ref{lem:compresults:term})) is also
non-computable.
Let $t\ (A)$ be the element of $\cand(s)$ that is given by
\refLemma{lem:scandcomplete} for $\gamma = \delta = []$.  Then $A =
\emptyset$ and $t$ has the form $\apps{\afun}{t_1}{t_\mia}$ with
$\mia = \minarity{\Rules}(\afun)$, and there exists a computable
substitution $\delta$ such that $t\delta$ is not computable but all
$t_i\delta$ are.

Thus, assuming $\arr{\Rules}$ is non-terminating, we can select a MNC
term $t_{-1}$.

Now for $i \in \N$, let a MNC term $t_{i-1} =
\apps{\afun}{q_1}{q_\mia}$ be given.  By definition of computability,
there are computable $q_{\mia+1},\dots,q_{\maa}$ such that
$\apps{\afun}{q_1}{q_\mia}$ has base type and is not computable.
Since all $q_i$ are computable this implies that $\apps{\afun}{q_1}{
q_\mia}$ is non-terminating, and since they are terminating, there
is eventually a reduction at the head: there exist a rule $\apps{
\afun}{\ell_1}{\ell_\mia} \arrz r$ and a substitution $\gamma$ such
that $\apps{\afun}{q_1}{q_\maa} \arrr{\Rules,in} \apps{\afun}{q_1'}{
q_\maa'}$ (where $\arrr{\Rules,in}$ indicates a reduction in the
argument terms $q_j$)
$= \apps{\apps{\afun}{(\ell_1\gamma)}{(\ell_\mia\gamma)}}{q_{\mia+1}'
}{q_\maa'} \arr{\Rules} \apps{(r\gamma)}{q_{\mia+1}'}{q_\maa'}$,
which latter term is still non-terminating.  But then also
$\apps{(r\gamma)}{q_{\mia+1}}{q_\maa}$ is non-terminating (as it
reduces to the term above), so $r\gamma$ is not computable.

From the above we conclude: $t_{i-1} \arrr{\Rules,in}
(\apps{\afun}{\ell_1}{\ell_\mia})\gamma
\arr{\Rules} r\gamma$, and $r\gamma$ is not computable.
By \refLemma{lem:formative}, we can safely assume that the reductions
$q_j \arrr{\Rules} \ell_j\gamma$ are $\ell_j$-formative if
$\apps{\identifier{f}}{\ell_1}{\ell_\mia}$ is a fully extended linear
pattern; and since $\apps{\afun}{\ell_1}{\ell_\mia}$ is closed we can
safely assume that $\domain(\gamma) =
\FMV(\apps{\afun}{\ell_1}{\ell_\mia})$.

Let $s_i := \apps{\afun^\sharp}{(\ell_1\gamma)}{(\ell_\mia\gamma)}$, and
note that all $\ell_j\gamma$ are computable by
\refLemma{lem:preservecomp}.
We observe that
for all $Z$ occurring in $r$ we have that $\gamma(Z)$ is $C$-computable
by a combination of accessible function passingness, computability
of $\ell_j\gamma$ and \refLemma{lem:pfp}.
As $r\gamma$ is non-computable, \refLemma{lem:scandcomplete} provides
an element $t\ (A)$ of $\cand(r)$ with pleasant minimality properties
and a computable substitution $\delta$ on domain $\FV(t)$ such that
$\gamma$ respects $A$ and $t(\gamma \cup \delta)$ is not computable.
For $\FV(t) = \{ x_1,\dots,x_n\}$, let $Z_1,\dots,Z_n$ be fresh
meta-variables; then $p :=  t[x_1:=Z_1,\dots,x_n:=Z_n] = \mathit{
metafy}(t)$, and $p\eta = t\delta$ for $\eta$ the substitution
mapping $X \in \FMV(\ell)$ to $\gamma(\ell)$ and each $Z_j$ to
$\delta(x_j)$.

Set $\rho_i := \ell^\sharp \arrdp p^\sharp\ (A)$ and $t_i :=
p^\sharp\eta$.  Then $t_i$ is MNC, because the meta-term
$t$ supplied by \refLemma{lem:scandcomplete} has the form
$\apps{\bfun}{u_1}{u_\maa}$ with $\maa = \arity(\bfun)$ and $u_j(
\gamma \cup \delta)$ is $C$-computable for each $j$ because $t
\bsuptermeq{A} u_j$.  Thus, we can continue the construction.

The chain $[(\rho_0,s_0,t_0), (\rho_1,s_1,t_1), \ldots]$ thus constructed is an
infinite formative $(\SDP(\Rules),\Rules)$-dependency chain.  That it
is a $(\SDP(\Rules),
\Rules)$-dependency chain is obvious because each $\rho_i \in
\SDP(\Rules)$ (since $t\ (A) \in \cand(r)$), because $\gamma$ respects
$A$ and $\eta$ corresponds with $\gamma$ on all meta-variables that
take arguments, and because $\FV(p) = \emptyset$ and $\domain(\eta) =
\FMV(\ell) \cup \{ Z_1,\dots, Z_n \} = \FMV(\ell) \cup \FMV(p)$.
That it is a formative chain follows by the initial selection of
$\gamma$, as we assumed formative reductions to each $\ell_j\gamma$.

It is also a computable chain: clearly we have $t_i = p\eta$ in step $i$ in
the construction above.  Suppose $p^\sharp \bsuptermeq{B} v$ and
$\eta$ respects $B$, but $(\abs{y_1 \dots y_n}{v})\eta$ is not
computable for $\FV(v) = \{ y_1,\dots,y_n \}$ -- so by
\refLemma{lem:abscomputable}, $v(\eta \cup \zeta)$ is not computable for
some computable substitution $\zeta$ on domain $\FV(v)$.  Since the
meta-variables $Z_j$ do not occur applied in $p$, we can safely
assume that $B$ contains only conditions for the meta-variables
in $\domain(\gamma)$.  By renaming each $Z_j$ back to $x_j$, we
obtain that $\gamma$ respects $B$ and $t \bsuptermeq{B} v'$ with
$v = v'[x_1:=Z_1,\dots,x_n:=Z_n]$.  But then $r\gamma \bsuptermeq{A
\cup B} t \bsuptermeq{A \cup B} v'$ and $\gamma$ respects $A \cup
B$ and $v'(\gamma \cup \delta \cup \zeta)$ is non-computable.  By
minimality of the choice $t\ (A)$, we have $v' = t$, so $v = p^\sharp$.
However, $p^\sharp\eta$ has a marked symbol $\symb{g}^\sharp$ as a
head symbol, and thus cannot be reduced at the top; by
$C$-computability of its immediate subterms, it is computable.
\qed
\end{proof}

We also prove the statement that $\AlterRules$-computability
implies minimality:

\begin{lemma}\label{lem:computableminimal}
Every $\AlterRules$-computable $(\P,\Rules)$-dependency chain is 
 minimal.
\end{lemma}

\begin{proof}
Let
$[(\rho_0,s_0,t_0), (\rho_1,s_1,t_1), \ldots]$
be a $\AlterRules$-computable
$(\P,\Rules)$-chain and let $i \in \N$; we must prove that the strict
subterms of $t_i$ are terminating under $\arr{\Rules}$.  By
definition, since $\bsuptermeq{\emptyset}$ is a
reflexive relation,
$t_i$ is $C_{\AlterRules}$-computable where $C_{\AlterRules}$ is given by
\refThm{thm:defC} for a relation $\arr{\AlterRules} \mathop{\supseteq}
\arr{\Rules}$.  By
\refLemma{lem:compresults}(\ref{lem:compresults:term}), $t_i$ is
therefore terminating under $\arr{\AlterRules}$, so certainly under
$\arr{\Rules}$ as well.  The strict subterms of a terminating term are
all terminating.
\qed
\end{proof}

\subsection{Static dependency pairs: the inverse result}\label{app:sdp:reverse}

In this section, we prove \refThm{thm:chainreverse}, which states that
the existence of certain kinds of dependency chains implies
non-termination of the original AFSM.  This is not a true inverse of
\refThm{thm:sschain} or even \refThm{thm:basicchain}: as observed in
the text, there are terminating AFSMs that do admit an infinite chain.
It does, however, give us a way to use the DP framework to prove
non-termination in some cases.

We begin by exploring the connection between $\bsuptermeq{\beta}$ and
reduction steps.  Note that this result is \emph{not} limited to
PA-AFP AFSMs, so in the following we do not assume that the rules or
dependency pairs involved satisfy arity restrictions.

\begin{lemma}\label{lem:candidatereduce}
Let $s,t$ be meta-terms and suppose $s \bsuptermeq{A} t$ for some set
$A$ of meta-variable conditions.  Then for any substitution $\gamma$
that respects $A$ and has a finite domain with $\FMV(s) \subseteq
\domain(\gamma) \subseteq \M$: $s\gamma\ (\bettersupterm \mathop{\cup}
\arr{\beta})^*\ t\gamma$.
\end{lemma}

\begin{proof}
By induction on the definition of $\bsuptermeq{A}$.  Consider the last
step in its derivation.
\begin{itemize}
\item If $s = t$ then $s\gamma = t\gamma$.
\item If $s = \abs{x}{u}$ and $u \bsuptermeq{A} t$, then by
  $\alpha$-conversion we can assume that $x \notin \FV(\gamma(Z))$ for
  any $Z \in \FMV(s)$. Thus, $s\gamma = \abs{x}{(u\gamma)}
  \bettersupterm u\gamma\ (\bettersupterm \mathop{\cup}
  \arr{\beta})^*\ t\gamma$ by the induction hypothesis.
\item If $s = \apps{(\abs{x}{u})}{s_0}{s_n}$ and $\apps{u[x:=s_0]}{
  s_1}{s_n} \bsuptermeq{A} t$, then by $\alpha$-conversion we can
  safely assume that $x$ is fresh w.r.t.\ $\gamma$ as above; thus,
  $s\gamma = \apps{(\abs{x}{(u\gamma)})}{(s_0\gamma)}{\linebreak
  (s_n\gamma)}
  \arr{\beta} \apps{(u\gamma[x:=s_0\gamma])}{(s_1\gamma)}{(s_n\gamma)}
  = (\apps{u[x:=s_0]}{s_1}{s_n})\gamma$, which reduces to $t\gamma$ by
  the induction hypothesis.
\item If $s = \apps{u}{s_1}{s_n}$ for $u$ an abstraction, variable,
  function symbol or meta-variable application, and $s_i \bsuptermeq{A}
  t$, then $s\gamma = \apps{(u\gamma)}{(s_1\gamma)}{(s_n\gamma)}
  \bettersupterm s_i\gamma\ (\bettersupterm \mathop{\cup}
  \arr{\beta})^*\ t\gamma$ by the induction hypothesis.
\item If $s = \apps{\meta{Z}{t_1,\dots,t_\mac}}{s_1}{s_n}$ and $t_i
  \bsuptermeq{A} t$ for some $1 \leq i \leq k$ with $(Z : i) \in A$,
  then we can write $\gamma(Z) = \abs{x_1 \dots x_n}{w}$ (where $n
  \leq \mac$), and $s\gamma = \apps{w[x_1:=t_1\gamma,\dots,x_n:=t_n
  \gamma]}{(t_{n+1}\gamma)}{(t_\mac\gamma)}$.
  Since $\gamma$ respects $A$, either $x_i$ occurs in $w$ or $i > n$;
  therefore $\gamma(Z) \bettersuptermeq t_i\gamma$.
  We again complete by the induction hypothesis.
\qed
\end{itemize}
\end{proof}

In fact, the text is ambiguous regarding the definition of $\SDP$ when
an AFSM is not properly applied, since $\minarity{\Rules}(\afun)$ may
not be uniquely defined.  However, the result holds for \emph{any}
choice of $\minarity{\Rules}(\afun)$.  In the following lemma, we only
use that the elements of $\SDP(\Rules)$ are DPs
$\apps{\afun^\sharp}{\ell_1}{\ell_\mia} \arrdp
\apps{\bfun^\sharp}{p_1}{p_i}\ (A)$ where $\apps{\afun^\sharp}{
\ell_1}{\ell_\mia} \arrz r$ is a rule and there exist $p_{i+1} \dots
p_n$ such that $r \bsuptermeq{A} \apps{\bfun}{p_1}{p_n}$.

\begin{lemma}\label{lem:dpreduce}
For $\ell^\sharp \arrdp p^\sharp\ (A) \in \SDP(\Rules)$ such that
$\FMV(p) \subseteq \FMV(\ell)$, and substitution $\gamma$ on domain
$\FMV(\ell)$ such that $\gamma$ respects the meta-variable conditions
in $A$: both $\ell\gamma$ and $p\gamma$ are terms and $\ell\gamma\ 
(\arr{\Rules} \mathop{\cup} \supterm)^+\ p\gamma$.
\end{lemma}

\begin{proof}
By definition of $\SDP$ and the fact that no fresh meta-variables
occur on the right, there is a rule $\ell \arrz r$ such that
$p\ (A) \in \cand(r)$, so there are $r_1,\dots,r_n$ such that
$r \bsuptermeq{A} \apps{p}{r_1}{r_n}$.
Clearly, we have $\ell\gamma \arr{\Rules} r\gamma$ by that rule,
and $r\gamma\ (\supterm \mathop{\cup} \arr{\Rules})^*\ (\apps{p}{r_1
}{r_n})\gamma \suptermeq p\gamma$ by \refLemma{lem:candidatereduce}
(using that $\bettersupterm$ is a sub-relation of $\supterm$).
We are done because $\arr{\beta}$ is included in $\arr{\Rules}$.
\qed
\end{proof}

This allows us to draw the required conclusion:

\oldcounter{\thmReverseChain}{\thmReverseChainSec}
\reverseChainTheThm
\startappendixcounters

\begin{proof}
Let $\unsharp{s_i}, \unsharp{t_i}$ denote the terms $s_i,t_i$ with all
$\sharp$ marks removed.
An infinite $(\SDP(\Rules),\Rules)$-dependency chain that does not use
any DPs where fresh meta-variables are introduced on the right-hand
side provides a sequence $(s_i,t_i)$ for $i \in \N$ such that for all
$i$, $\unsharp{s_i}\ (\arr{\Rules} \mathop{\cup} \supterm)^+\ 
\unsharp{t_i}$ (by \refLemma{lem:dpreduce}), and $\unsharp{t_i}
\arrr{\Rules} \unsharp{s_{i+1}}$.  Thus, we obtain an infinite
$\arr{\Rules} \mathop{\cup} \supterm$ sequence, which provides an
infinite $\arr{\Rules}$ sequence due to monotonicity of $\arr{\Rules}$.
\qed
\end{proof}

 \subsection{Original static dependency pairs}\label{app:sdp:original}
 
 Since the most recent work on static dependency pairs has been defined
 for a polymorphic variation of the HRS formalism, it is not evident
 from sight how our definitions relate.  Here, we provide context by
 showing how the definitions from~\cite{kus:iso:sak:bla:09,suz:kus:bla:11}
 apply to the restriction of HRSs that can be translated to AFSMs.
 
 It should be noted that HRSs, as translated to AFSMs, should be seen
 as $\eta$-expanded rules; in practice, for $\ell \arrz r$ we have
 that $\ell \bettersuptermeq s$ or $r \bettersuptermeq s$ implies that
 either $s$ is an abstraction, or $s$ has base type.  This definition
 implies that the system is properly applied, but is much stronger.
 We will refer to this restriction as \emph{fully applied}.
 
 \begin{definition}\label{def:pfp1}
 An AFSM $(\F,\Rules)$ is plain function passing following
 \cite{kus:iso:sak:bla:09} if:
 \begin{itemize}
 \item
 for all rules $\apps{\afun}{\ell_1}{\ell_\maa} \arrz r$ and all
   $Z \in \FMV(r)$:
   if $Z$ does not have base type, then there are variables
   $x_1,\dots,x_n$ and some $i$ such that $\ell_i = \abs{x_1 \dots
   x_n}{\meta{Z}{x_{j_1},\dots,x_{j_\mia}}}$.
 \end{itemize}
 
 An AFSM $(\F,\Rules)$ is plain function passing following
 \cite{suz:kus:bla:11} if:
 \begin{itemize}
 \item
   for all rules $\apps{\afun}{\ell_1}{\ell_\maa} \arrz r$ and all
   $Z \in \FMV(r)$: there are
      some variables
   $x_1,\dots,x_\mia$ and
   some $i \leq \maa$ such that $\ell_i \safesup \meta{Z}{x_1,\dots,
   x_k}$, where the relation $\safesup$ is given by:
   \begin{itemize}
   \item
     $s \safesup s$,
   \item
     $\abs{x}{t} \safesup s$ if $t \safesup s$,
   \item
     $\apps{x}{t_1}{t_n} \safesup s$ if $t_i \safesup s$ for some
     $i$ with $x \in \V \setminus \FV(t_i)$%
   \item
     $\apps{\identifier{f}}{t_1}{t_n} \safesup s$ if $t_i \safesup
     s$ for some $t_i$ of base type.%
     \footnote{The authors of \cite{suz:kus:bla:11} refer to such
     subterms as \emph{accessible}.  We do not use this terminology,
     as it does not correspond to the accessibility notion in
     \cite{bla:jou:oka:02,bla:jou:rub:15} which we follow here.
     In particular, the accessibility notion we use considers the
     relation $\gracsortup$, which corresponds to the positive/negative
     inductive types in \cite{bla:jou:oka:02,bla:jou:rub:15}.  This is
     not used in \cite{suz:kus:bla:11}.}
   \end{itemize}
 \end{itemize}
 
 In addition, in both cases right-hand sides of rules are assumed to be
 presented in $\beta$-normal form and are fully applied.
 \end{definition}
 
 The definitions of PFP in \cite{kus:iso:sak:bla:09,suz:kus:bla:11}
 also capture some non-pattern HRSs,
 but these cannot be represented as AFSMs.
 Note that the key difference
 between $\safesup$ and $\suptermeq$ for patterns is that the former
 is not allowed to descend into a non-base argument of a function symbol.
 The same difference applies when comparing $\safesup$ with $\gracc$:
 $\safesup$ also cannot descend into the accessible higher-order
 arguments.
 
 \begin{example}
 The rules from \refEx{ex:mapintro} are PFP
 following both defi\-nitions.
 The rules from \refEx{ex:deriv} are PFP following
 \cite{suz:kus:bla:11} but not following \cite{kus:iso:sak:bla:09}.
 The rules from \refEx{ex:ordrec} are not PFP
 in either definition, since $\symb{lim}\ F \safesup F$ does
 not hold (although they \emph{are} AFP).
 \end{example}
 
 For a PFP AFSM, static dependency pairs are then
 defined as \emph{pairs} $\ell^\sharp \arrdp \apps{\identifier{f}^{
 \sharp}}{p_1}{p_m}$.
 This allows for a very simple notion of chains, even closer to the one
 in the first-order setting than our \refDef{def:chain}.
 
 \begin{definition}\label{def:sdpchain1}
 A static dependency chain following
 \cite{kus:iso:sak:bla:09,suz:kus:bla:11} is an infinite sequence
 $[(\ell_0 \arrdp p_0,\gamma_0),(\ell_1 \arrdp p_1,\gamma_1),\ldots]$
 where
 $p_i\gamma_i \arrr{\Rules} \ell_{i+1}\gamma_{i+1}$ for all $i$.
 It is \emph{minimal} if each $p_i\gamma_i$ is terminating under
 $\arr{\Rules}$.
 \end{definition}
 
 Both papers present a counterpart of Theorems \ref{thm:basicchain}
 and \ref{thm:sschain} that roughly translates to the following:
 
 \begin{theorem}[\cite{kus:iso:sak:bla:09,suz:kus:bla:11}]\label{thm:pfp1}
 Let $\Rules$ be plain function passing following either definition in
 \refDef{def:pfp1}.  Let $\P = \{ \ell^\sharp \arrdp \apps{
 \identifier{f}^\sharp}{p_1}{p_\maa} \mid \ell \arrz r \in \Rules \wedge
 r \suptermeq \apps{\identifier{f}}{p_1}{p_\maa} \wedge \identifier{f}
 \in \Defineds \wedge \maa = \arity(\identifier{f}) \}$.  If $\arr{\Rules}$
 is non-terminating, then there is an infinite minimal static
 dependency chain with all $\ell_i \arrdp p_i \in \P$.
 \end{theorem}
 
 Note that the chains are proved \emph{minimal}, but not
 \emph{computable} (which is a new definition in the current paper).
 
 However, there is no counterpart to \refThm{thm:chainreverse}: this
 result relies on the presence of meta-variable conditions, which are
 not present in the static DPs from the literature.
 
 \medskip
 Note that $\gracc$ corresponds to $\safesup$ (from
 \refDef{def:pfp1}) if $\greqsort$
  equates all sorts (as then always
 $\Acc(\afun) = \{$ the indices of all base type arguments
 of $\afun\}$).  Thus,
 \refDef{def:apfp} includes both notions from \refDef{def:pfp1}.

\section{Dependency pair processors}\label{app:processors}

In this appendix, we prove the soundness -- and where applicable
completeness -- of all DP processors defined in the text.

We first observe:

\begin{lemma}\label{lem:subsetcomplete}
If $\Proc$ maps every DP problem to a set of problems such that for
all $(\P',\Rules',m',f') \in \Proc(\P,\Rules,m,f)$ we have that $\P'
\subseteq \P$, $\Rules' \subseteq \Rules$, $m' \succeq m$ and $f' =
f$, then $\Proc$ is complete.
\end{lemma}

\begin{proof}
$\Proc(\P,\Rules,m,f)$ is never \texttt{NO}.  Suppose $\Proc(\P,
\Rules,m,f)$ contains an infinite element $(\P',\Rules',m',f')$; we
must prove that then $(\P,\Rules,m,f)$ is infinite as well.  This is
certainly the case if $\arr{\Rules}$ is non-terminating, so assume
that $\arr{\Rules}$ is terminating.  Then certainly $\arr{\Rules'}
\mathop{\subseteq} \arr{\Rules}$ is
terminating as well, so $(\P',\Rules',m',f')$ can be infinite only
because there exists an infinite $(\P',\Rules')$-chain that is
$\AlterRules$-computable if $m' = \static_\AlterRules$, minimal if
$m' = \minimal$ and formative if $f' = \formative$.  By definition,
this is also a $(\P,\Rules)$-dependency chain, which is formative if
$f = f' = \formative$.  Since $\arr{\Rules}$ is terminating, this
chain is also minimal.
If we have $m = \static_\AlterRules$, then also $m' = \static_\AlterRules$
(since $\static_\AlterRules$ is maximal under $\succeq$)
and the chain is indeed $\AlterRules$-computable.
\qed
\end{proof}

\subsection{The dependency graph}

The dependency graph processor lets us split a DP problem into
multiple smaller ones.
To prove soundness of its main processor, we first prove a helper
result.

\begin{lemma}\label{lem:graphend}
Let $\adpprob = (\P,\Rules,m,f)$ and $G_\theta$ an approximation of
its dependency graph.  Then for every infinite $\adpprob$-chain
$[(\rho_0,s_0,t_0), (\rho_1,s_1,t_1), \ldots]$
there exist $n \in \N$ and a cycle
$C$ in $G_\theta$ such that for all $i > n$: $\theta(\rho_i) \in C$.
\end{lemma}

\begin{proof}
We claim (**):
for all $i \in \N$, there is an edge from $\theta(\rho_i)$ to
$\theta(\rho_{i+1})$.
By definition of \emph{approximation}, the claim follows if $DG$ has
an edge from $\rho_i$ to $\rho_{i+1}$.  But this is obvious: by
definition of a chain, if
$[(\rho_0,s_0,t_0), (\rho_1,s_1,t_1), \ldots]$
is a
dependency chain, then so is $[(\rho_i,s_i,t_i),(\rho_{i+1},s_{i+1},
t_{i+1})]$.

Now, having (**), the chain traces an infinite path in
$G_\theta$.  Let $C$ be the set of nodes that occur infinitely often
on this path; then for every node $d$ that is not in $C$, there is
an index $n_d$ after which $\theta(\rho_i)$ is never $d$ anymore.
Since $G_\theta$ is a finite graph, we can take $n :=
\max(\{ n_d \mid d$ a node in $G_\theta \wedge d \notin C \})$.  Now
for every pair $d,b \in C$: because they occur infinitely often, there
is some $i > n$ with $\theta(\rho_i) = d$ and there is $j > i$ with
$\theta(\rho_j) = b$.  Thus, by (**) there is a path in $G_\theta$
from $d$ to $b$.  Similarly, there is a path from $b$ to $d$.
Hence, they are on a cycle.
\qed
\end{proof}

Note that to find a chain with all $\theta(\rho_i) \in C$, we do not
need to modify the original chain at all: the satisfying chain is a
tail of the original chain.  Hence, the same flags apply to the
resulting chain.
This makes it very easy to prove correctness of the main processor:

\oldcounter{\procDPGraph}{\procDPGraphSec}
\DPGraphTheProc
\startappendixcounters

\begin{proof}
Completeness follows by \refLemma{lem:subsetcomplete}.
Soundness follows because if $(\P,\Rules,m,\linebreak
f)$ admits an infinite
chain, then by \refLemma{lem:graphend} there is a cycle $C$ such that
a tail of this chain is mapped into $C$.  Let $C'$ be the strongly
connected component in which
$C$ lies, and $\P' = \{ \rho \in \P \mid \theta(\rho) \in
C' \}$.  Then clearly the same tail lies in $\P'$, giving an infinite
$(\P',\Rules,m,f)$-chain, and $(\P',\Rules,m,f)$ is one of the elements
of the set returned by the dependency graph processor.
\qed
\end{proof}

The dependency graph processor is essential to prove termination in
our framework because it is the only processor defined so far that
can map a DP problem to $\emptyset$.

\subsection{Processors based on reduction triples}

\oldcounter{\procRedpair}{\procRedpairSec}
\redpairTheProc
\startappendixcounters

\begin{proof}
Completeness follows by \refLemma{lem:subsetcomplete}.  Soundness
follows because every infinite $(\P_1 \uplus \P_2,\Rules)$-chain
$[(\rho_0,s_0,t_0), (\rho_1,s_1,t_1), \ldots]$
with $\P_1,\P_2,\Rules$ satisfying
the given properties induces an infinite $\pgt \cup \pge \cup \rge$
sequence, and every occurrence of a DP in $\P_1$ in the chain
corresponds to a $\pgt$ step in the sequence.  By compatibility of
the relations, well-foundedness guarantees that there can only be
finitely many such steps, so there exists some $n$ such that
$[(\rho_{n},s_{n},t_{n}), (\rho_{n+1},s_{n+1},t_{n+1}), \ldots]$
is an infinite
$(\P_2,\Rules)$-chain.

To see that we indeed obtain the sequence, let $i \in \N$.  Denote
$\rho_i := \ell \arrdp\ p\ (A)$, and let $\gamma$ be a substitution
on domain $\FMV(\ell) \cup \FMV(p)$ such that $s_i = \ell\gamma$ and
$t_i = p\gamma$.  Meta-stability gives us that $s_i = \ell\gamma\ 
(\pge \cup \pgt)\ p\gamma = t_i$.  As $\arr{\Rules}$ is included in
$\rge$ by meta-stability and monotonicity, and because $t_i \arrr{
\Rules} s_{i+1}$, we have $t_i \rge s_{i+1}$.
Thus, $s_i (\pge \cup \pgt) \cdot \rge s_{i+1}$.
Moreover, a $\pgt$ step is used if $\rho_i \in \P_1$.
\qed
\end{proof}

Now that we have seen a basic processor using reduction triples,
soundness of the base-type processor presented in the text follows
easily.

\oldcounter{\procBasetypeRedpair}{\procBasetypeRedpairSec}
\basetypeRedpairTheProc
\startappendixcounters

\begin{proof}
Completeness follows by \refLemma{lem:subsetcomplete}.  Soundness
follows by soundness of \refThm{thm:redpairproc}: let $(\rge,\pge,\pgt)$
be a reduction triple satisfying the requirements above, and for
$R \in \{\pge,\pgt\}$ define $R'$ as follows: for $s : \atype_1
\arrtype \dots \arrtype \atype_\maa \arrtype \asort$ and $t :
\btype_1 \arrtype \dots \arrtype \btype_n \arrtype \bsort$, let
$s\ R'\ t$ if for all $u_1 : \atype_1,\dots,u_\maa : \atype_\maa$
there exist $w_1 : \btype_1,\dots,w_n : \btype_n$ such that
$\apps{s}{u_1}{u_\maa}\ R\ \apps{t}{w_1}{w_n}$.  We claim that
$(\rge,\pge',\pgt')$ is a reduction triple satisfying the requirements
of \refThm{thm:redpairproc}, which implies soundness of the present
processor.

It is clear that $\pge'$ and $\pgt'$ satisfy the requirements of
\refThm{thm:redpairproc}: if $\ell \arrdp p\ (A) \in \P_1$, then
for any $u_1,\dots,u_\maa$ we let $w_1:=\bot_{\btype_1},\dots,w_n:=
\bot_{\btype_n}$ and have $\apps{\ell}{u_1}{u_\maa} \pgt
\apps{p}{w_1}{w_n}$ by meta-stability of $\pgt$; the same holds for
$\pge$.  It remains to be seen that $\pgt'$ and $\pge'$ are both
transitive and meta-stable, that $\pge'$ is
reflexive and that $\pgt'$
is well-founded.
\begin{itemize}
\item Meta-stability: given that $\ell \pgt' p$ and $\gamma$ is a
  substitution on domain $\FMV(\ell) \cup \FMV(p)$, we must see that
  $\ell\gamma \pgt' p\gamma$ (the case for $\pge'$ follows in the
  same way).  Let $u_1,\dots,u_\maa$ be arbitrary terms and $\delta
  := \gamma \cup [Z_1:=u_1,\dots,Z_\maa:=u_\maa]$; then
  $(\apps{\ell}{Z_1}{Z_\maa})\delta \pgt (\apps{p}{\bot_{\btype_1}}{
  \bot_{\btype_n}})\delta$, so indeed $\ell\delta = \ell\gamma \pgt'
  p\gamma$.
\item Transitivity: if $s \pgt t \pgt v$ then for all $\vec{u}$
  there exist $\vec{w}$ such that $s\ \vec{u} \pgt t\ \vec{w}$, and
  for all $\vec{w}$ there exist $\vec{q}$ such that $t\ \vec{w}
  \pgt v\ \vec{q}$; thus, also $s\ \vec{u} \pgt v\ \vec{q}$.  The
  case for $\pge$ is similar.
\item
Reflexivity
of $\pge'$: always $s \pge' s$ since for all $u_1,
  \dots,u_\maa$ we have $s\ \vec{u} \pge s\ \vec{u}$.
\item Well-foundedness of $\pgt'$: suppose $s_1 \pgt' s_2 \pgt' \dots$
  and let $\vec{u_1}$ be a sequence of variables;
  we find
  $\vec{u_2},\vec{u_3},\dots$ such that
  $s_1\ \vec{u_1} \pgt' s_2\ \vec{u_2} \pgt' \dots$ as in the case for
  transitivity.
\qed
\end{itemize}
\end{proof}

\subsection{Rule removal without search for orderings}

There is very little to prove: the importance is in the definition.

\oldcounter{\procFormative}{\procFormativeSec}
\formativeTheProc
\startappendixcounters

\begin{proof}
Completeness follows by \refLemma{lem:subsetcomplete}.  Soundness
follows by definition of a formative rules approximation (a formative
infinite
$(\P,\Rules)$-dependency chain can be built using only rules in
$\FR(\P,\Rules)$).
\qed
\end{proof}

The practical challenge lies in proving that a given formative rules
approximation really is one.  The definition of a good approximation
function is left to future work.

\subsection{Subterm criterion processors}

Next, we move on to the subterm processors.  We first present the
basic one -- which differs little from its first-order counterpart,
but is provided for context.

\oldcounter{\procSubtermCriterion}{\procSubtermCriterionSec}
\subtermCriterionTheProc
\startappendixcounters

\begin{proof}
Completeness follows by \refLemma{lem:subsetcomplete}.  Soundness
follows because an infinite $(\P,\Rules,m,f)$-chain with the
properties above induces an infinite sequence $\project(s_0)
\suptermeq \project(t_0) \arrr{\Rules} \project(s_1) \suptermeq
\project(t_1) \arrr{\Rules} \dots$.
Since the chain is minimal (either because $m = \minimal$, or by
\refLemma{lem:computableminimal} if $m = \static_\AlterRules$),
$\project(p_0)$ is terminating, and therefore it is terminating under
$\arr{\Rules} \mathop{\cup} \supterm$.
Thus, there is some index $n$ such
that for all $i \geq n$: $\project(s_i) = \project(t_i) =
\project(s_{i+1})$.  But this can only be the case if
$\project(\ell_i) = \project(p_i)$.
But then the tail of the chain starting at
position $n$ does not use any pair in $\P_1$, and is therefore an
infinite $(\P_2,\Rules,m,f)$-chain.
\qed
\end{proof}

We now turn to the proof of the computable subterm
criterion processor.  This proof is very similar to the one for the
normal subterm criterion, but it fundamentally uses the definition of
a computable chain.

\oldcounter{\procStaticSubtermCriterion}{\procStaticSubtermCriterionSec}
\staticSubtermCriterionTheProc
\startappendixcounters

\begin{proof}
Completeness follows by \refLemma{lem:subsetcomplete}.  Soundness
follows because, for $C := C_\AlterRules$ the computability predicate
corresponding to $\arr{S}$, an infinite $(\P,\Rules,\linebreak
\static_\AlterRules,f)$-chain
induces an infinite $\accreduce{C} \cup \arr{\Rules}$ sequence starting
in the $C$-computable term $\project(s_1)$, with always $s_i
(\accreduce{C} \cup \arr{\Rules})^* t_i$ if $\rho_i \in \P_1$ and
$\project(s_i) = \project(t_i)$ if $\rho_i \in \P_2$; like in the proof
of the subterm criterion, this proves that the chain has a tail that is
a $(\P_2,\Rules,\static_\AlterRules,f)$-chain because, by definition
of $C$, $\arr{\Rules} \cup \accreduce{C}$ is terminating on
$C$-computable terms.

It remains to be seen that we indeed have $\project(s_i)\ (\accreduce{C}
\cup\ \arr{\Rules})^+
 \project(t_i)$ whenever $\rho_i \in \P_1$.  So
suppose that $\rho$ is a dependency pair $\ell \arrdp p\ (A) \in \P_1$
such that $\project(\ell) \sqsupset \project(p)$; we must see that
$\project(\ell\gamma)\ (\accreduce{C} \cup \arr{\beta})^+
\project(p\gamma)$ for any substitution $\gamma$ on domain
$\FMV(\ell) \cup \FMV(r)$ such that $v\gamma$ is $C$-computable for
all $v,B$ such that $r \bsuptermeq{B} v$ and $\gamma$ respects $B$.

Write $\ell = \apps{\afun}{\ell_1}{\ell_\mia}$ and $p = \apps{\bfun}{
p_1}{p_n}$; then $\project(\ell\gamma) = \ell_{\nu(\afun)}\gamma$ and
$\project(p\gamma) = p_{\nu(\bfun)}\gamma$.
Since, by definition of a dependency pair, $\ell$ is closed, we also
have $\FV(\ell_{\nu(\afun)}) = \emptyset$.
Consider the two possible reasons why $\ell_{\nu(\afun)} \sqsupset
p_{\nu(\bfun)}$.

\begin{itemize}
\item $\ell_{\nu(\afun)} \gracc p_{\nu(\bfun)}$:
  since both sides have base type by assumption and $\ell_{\nu(\afun)}$
  is closed, 
  by \refLemma{lem:preservecompacchelper} also
  $\ell_{\nu(\afun)}\gamma\ (\accreduce{C} \cup \arr{\beta})^*
  p_{\nu(\bfun)}\gamma$.
\item $\ell_{\nu(\afun)} \gracc \meta{Z}{x_1,\dots,x_\mia}$ and
  $p_{\nu(\bfun)} = \apps{\meta{Z}{u_1,\dots,u_\mac}}{v_1}{v_n}$:
  denote $\gamma(Z) = \abs{x_1 \dots x_\mia}{q}$ and also $\gamma(Z)
  \approxp \abs{x_1 \dots x_\mac}{q'}$.  Then we can write $q =\linebreak
  \abs{x_{\mia+1} \dots x_i}{q''}$ as well as $q' = \apps{q''}{x_{i+1}}{
  x_\mac}$ for some $\mia \leq i \leq \mac$.
  Moreover:
  \[
  p_{\nu(\bfun)}\gamma = \apps{q'[x_1:=u_1\gamma,\dots,
  x_\mac:=u_\mac\gamma]}{v_1\gamma}{v_n\gamma}
  \]

  By definition of an $\AlterRules$-computable chain, $v_j\gamma$ is computable
  for each $1 \leq j \leq n$, and $u_j\gamma$ is computable for each
  $1 \leq j \leq \mac$ such that $x_j \in \FV(q')$.  Write $v_j' :=
  v_j\gamma$ and let $u_j' := u_j\gamma$ if $x_j \in \FV(q')$,
  otherwise $u_j' :=$ a fresh variable; then all $u_j'$ and $v_j'$
  are computable, and still:
  \[
  \begin{array}{cl}
    & p_{\nu(\bfun)}\gamma\\
  = & \apps{q'[x_1:= u_1',\dots,x_\mac:=u_\mac']}{v_1'}{v_n'}\\
  = & \apps{\apps{q''[x_1:=u_1',
      \dots,x_i:=u_i']}{u_{i+1}'}{u_{\mac'}'}}{v_1'}{v_n'}
  \end{array}
  \]
  On the other hand, by \refLemma{lem:preservecompacchelper} and the
  observation that $\FV(\ell_{\nu(\afun)}) = \emptyset$, we have
  $\ell_{\nu(\afun)
  }\gamma\ (\accreduce{C} \cup \arr{\beta})^+\
  \apps{\apps{q[x_1:=u_1',
  \dots,
  x_\mia:=u_\mia']}{u_{\mia+1}'}{u_\mac}'}{v_1'}{v_n'}$,
  and as
  $q = \abs{x_{\mia+1} \dots x_i}{q''}$ this term $\beta$-reduces to
  $\apps{\apps{q''[x_1:=u_1',\dots,x_i:=u_i']}{u_{i+1}'}{
  u_{\mac'}'}}{v_1'}{v_n'} = p_{\nu(\bfun)}\gamma$.
\qed
\end{itemize}
\end{proof}

\subsection{Non-termination}

Soundness and completeness of the non-termination processor
in \refThm{def:nontermproc} are both direct consequences of
\refDef{def:dpproblem} and \refDef{def:proc}.

\section{Experimental results}\label{sec:experiments}

Finally, while the main paper focuses on a theoretical
exposition, we here present an experimental evaluation of the
results in this paper.  The work has been implemented in the second
author's termination tool \wanda, using \emph{higher-order polynomial
interpretations} \cite{fuh:kop:12} and a \emph{recursive path ordering}
\cite[Chapter~5]{kop:12} for reduction triples.  Both methods rely on
an encoding of underlying constraints into SAT.  The search for a sort
ordering, and a projection function for the subterm criterion, is also
delegated to SAT.  The non-termination processor has not been implemented
(\wanda performs some loop analysis, but outside the DP framework; this
is a planned future improvement), and the subterm criterion processor
has been merged with the computable subterm criterion processor.

In addition to the results in this paper, \wanda includes a
search for monotonic termination orderings outside the DP framework
(using successive rule removal), a dynamic DP framework (following
\cite{kop:raa:12}), and a mechanism \cite{fuh:kop:11}
within both DP frameworks to delegate some first-order parts of
the AFSM to a first-order termination tool (here we use
\aprove\ \cite{gie:asc:bro:emm:fro:fuh:hen:ott:plu:sch:str:swi:thi:17}).

We have evaluated the power of our techniques on the
\emph{Termination Problems Database} \cite{tpdb}, version 10.5.
Of the 198 benchmarks in the category \emph{Higher Order Union Beta},
153 are accessible function passing.  Comparing the power of static
DPs versus dynamic DPs or no DP framework gives the following results
(where \emph{Time} is the average runtime on success in seconds):

\smallskip\noindent
\begin{tabular}{|l|l|l|}
\hline
\textbf{Technique} & \textbf{Yes} & \textbf{Time} \\
\hline
Only rule removal & \phantom{0}92 & \ 0.32 \\
\hline
Static DPs with techniques from this paper & 124 & \ 0.07 \\
\hline
Dynamic DPs with techniques from this paper & 132 & \ 0.52 \\
\hline
Static DPs with delegation to a first-order prover &
  129 & \ 0.58 \\
\hline
Dynamic DPs with delegation to a first-order prover &
  137 & \ 0.93 \\
\hline
Static and dynamic DPs with delegation to a first-order prover &
  150 & \ 0.70 \\
\hline
Non-terminating? & \phantom{0}16 & \ \\
\hline
\end{tabular}

\smallskip
While static DPs have a slightly lower success rate than dynamic DPs,
their evaluation is much faster, since they allow for greater
modularity: the dynamic setting includes dependency pairs
where the right-hand does not have a (marked) defined symbol at the
head (e.g., $\map\ F\ (\cons\ H\ T) \arrdp F\ H$), which make both the
subterm criterion processor and the dependency graph processor harder
to apply. The combination of static and dynamic DPs performs
substantially better than either style alone: although the gains can
be seen as modest if the size of the data set is not taken into account
(153 versus 166 YES+NO, giving an 8\% increase),
the numbers indicate a 29\% \emph{decrease} in failure rate (from 45
to 32).

We have also compared the individual techniques in this paper by
disabling them from the second test above.  This gives the table
below:

\smallskip\noindent
\begin{minipage}{0.39\textwidth}
\begin{tabular}{|l|l|l|}
\hline
\textbf{Disabled} & \textbf{Yes} & \textbf{Time} \\
\hline
Formative rules & 124 & \ 0.07 \\
\hline
Usable rules & 124 & \ 0.10 \\
\hline
Subterm criterion & 121 & \ 0.16 \\
\hline
Graph & 121 & \ 0.35 \\
\hline
Reduction triples & \phantom{0}92 & \ 0.01 \\
\hline
Nothing & 124 & \ 0.07 \\
\hline
\end{tabular}
\end{minipage}
\parbox{0.6\textwidth}{Note that none of the techniques individually
give much power except for reduction triples:
where one method is disabled,
another can typically pick up the slack.  If all processors except
reduction triples are disabled, only 105 benchmarks are proved.
Most  processors do individually give a significant \emph{speedup}.
Only formative rules does not; this processor is useful in the dynamic
setting, but does not appear to be so here.}

\medskip
Evaluation pages for these and further experiments
are available at:
\begin{center}
\url{https://www.cs.ru.nl/~cynthiakop/experiments/esop19/}
\end{center}

\end{document}